\documentclass[11pt]{article}
\usepackage[margin = 1in]{geometry}

\usepackage{natbib}
\newcount\Comments  %

\newcommand{\kibitz}[2]{\ifnum\Comments=1{\color{#1}{#2}}\fi}

\newcount\AddComments  %
\newcommand{\kibitzAdd}[2]{\ifnum\AddComments=1{\color{#1}{#2}}\fi}

\newcount\vsqCounter
\newcommand{\vsq}[1]{\ifnum\vsqCounter=1{\vspace{#1}}\fi}

\newcommand{\driverthickness}{thick}

\newcommand{\driverOneColor}{violet}
\newcommand{\driverTwoColor}{orange}
\newcommand{\driverThreeColor}{lightblue}

\newcommand{\driverOneTextColor}{violet}
\newcommand{\driverTwoTextColor}{orange}
\newcommand{\driverThreeTextColor}{lightblue}

\newcommand{\driverOneRiderColor}{violet}
\newcommand{\driverTwoRiderColor}{orange}
\newcommand{\driverThreeRiderColor}{lightblue}

\newcommand{\driverOneLgdColor}{violet}
\newcommand{\driverTwoLgdColor}{orange}
\newcommand{\driverThreeLgdColor}{lightblue}

\usepackage{etex}
\usepackage{latexsym}
\usepackage{amssymb}
\usepackage{amsmath}
\usepackage{amsfonts}
\usepackage{bbm}
\usepackage{ifthen}
\usepackage{enumerate}
\usepackage{psfrag}
\usepackage{floatflt}
\usepackage{verbatim}
\usepackage{dsfont}
\usepackage{scalefnt}
\usepackage{epsfig,epstopdf}
\usepackage{calc}
\usepackage{graphicx}
\usepackage{subcaption}
\usepackage{multicol}
\usepackage{tikz}
\usetikzlibrary{shapes.multipart,arrows,automata}
\usetikzlibrary{patterns,calc,intersections,through,backgrounds,decorations.pathreplacing}
\usepackage{xstring}

\usepackage{amsthm}		%
\usepackage{thmtools, thm-restate}
\usepackage{algorithmic}
\usepackage[linesnumbered,ruled,vlined]{algorithm2e}

\usepackage{color, xcolor} 
\colorlet{darkblue}{blue!40!black}
\definecolor{auburn}{rgb}{0.43, 0.21, 0.1}
\definecolor{orange}{rgb}{1, 0.5, 0}
\definecolor{lightblue}{rgb}{0.1176, 0.5647, 1}

\definecolor{matlabblue}{rgb}{0    0.4470    0.7410}
\definecolor{matlabred}{rgb}{0.8500    0.3250    0.0980}
\definecolor{matlabyellow}{rgb}{ 0.9290    0.6940    0.1250}
\definecolor{matlabpurple}{rgb}{0.4940    0.1840    0.5560}
       
\usepackage[colorlinks,
	citecolor = darkblue,
	urlcolor = black,
	linkcolor = darkblue,
	breaklinks = true]{hyperref}

\usepackage{afterpage}

\theoremstyle{plain}
\newtheorem{theorem}{Theorem}

\theoremstyle{definition}
\newtheorem{definition}{Definition} 
\newtheorem{example}{Example}
\newtheorem{lemma}{Lemma}

\newcommand{\horizon}{T}					%
\newcommand{\timeSet}{[\horizon]}
\newcommand{\activeTimeSet}{[\horizon - 1]}

\newcommand{\loc}{\calL}					%
\newcommand{\dist}{\delta}					%
\newcommand{\trips}{\mathcal{T}} 			%

\newcommand{\driverSet}{\mathcal{D}}		%
\newcommand{\nd}{m}						%
\newcommand{\driverEntrance}{\beta}				%
\newcommand{\te}{\underline{\tau}}			%
\newcommand{\tl}{\bar{\tau}}				%
\newcommand{\re}{\ell}						%

\newcommand{\cost}{c}
\newcommand{\pathCost}{\lambda} 			
\newcommand{\exitCost}{\kappa}

\newcommand{\riders}{\mathcal{R}}		%
\newcommand{\origin}{o}					%
\newcommand{\dest}{d}					%
\newcommand{\tr}{\tau}					%
\newcommand{\val}{v}					%
\newcommand{\surplus}{w}				%

\newcommand{\pathSet}{\mathcal{Z}}		%
\renewcommand{\path}{Z}					%

\newcommand{\actPathSet}{\tilde{\pathSet}}
\newcommand{\actpath}{\tilde{z}}

\newcommand{\price}{p}					%
\newcommand{\rPayment}{q}			%
\newcommand{\dPayment}{r}			%

\newcommand{\state}{s}					%
\newcommand{\history}{h}				%
\newcommand{\historySet}{\mathcal{H}}	%

\newcommand{\action}{\alpha}			%
\newcommand{\actionSet}{\mathcal{A}}	%

\newcommand{\strategy}{\sigma} 		%

\newcommand{\actualPayment}{\hat{\dPayment}}
\newcommand{\actualDriverUtil}{\hat{\pi}}
\newcommand{\actualRiderPayment}{\hat{\rPayment}}
\newcommand{\actualAssignment}{\hat{x}}

\newcommand{\econ}{E}

\newcommand{\sw}{W}						%
\newcommand{\node}{n}					%
\newcommand{\edgeCost}{\gamma}			%
\newcommand{\edge}{e}					%
\newcommand{\graph}{G}					%
\newcommand{\edgeSet}{\mathcal{E}}		%
\newcommand{\nodeSet}{\mathcal{N}}		%
\newcommand{\driverNode}{D}				%
\newcommand{\sink}{S}
\newcommand{\riderEdge}{R}

\newcommand{\resGraph}{\tilde{\graph}}		%
\newcommand{\resEdgeSet}{\tilde{\edgeSet}}	%
\newcommand{\resEdge}{\tilde{\edge}}		%
\newcommand{\resRiderEdge}{\tilde{\riderEdge}}

\newcommand{\flow}{f}						%

\newcommand{\tail}{\partial^+}
\newcommand{\head}{\partial^-}

\newcommand{\capacity}{	\zeta}
\newcommand{\minCap}{\underline{\capacity}}
\newcommand{\maxCap}{\bar{\capacity}}

\newcommand{\flowMatrix}{F}

\newcommand{\barphi}{\bar{\varphi}}
\newcommand{\undphi}{\underline{\varphi}}
\newcommand{\barmu}{\bar{\mu}}
\newcommand{\undmu}{\underline{\mu}}
\newcommand{\barp}{\bar{\price}}
\newcommand{\undp}{\underline{\price}}

\newcommand{\CS}{CS}
\newcommand{\FCS}{CS$_F$}

\newcommand{\cs}[1]{(\CS-\ref{item:cs#1})}
\newcommand{\fcs}[1]{(\FCS-\ref{item:fcs#1})}

\newcommand{\shorteq}{\hspace{-0.2em}=\hspace{-0.2em}}

\newcommand{\set}[2]{\left\lbrace \left. #1 ~\right|~  #2\right\rbrace}
\newcommand{\tprime}{t'}

\newcommand{\0}{^{(0)}}
\newcommand{\supt}{^{(t)}}
\newcommand{\suptprime}{^{(t')}}

\newcommand{\calL}{\mathcal{L}}

\newcommand{\setR}{\mathbb{R}}
\newcommand{\setZ}{\mathbb{Z}}
\newcommand{\setN}{\mathbb{N}}

\newcommand{\txtif}{~\mathrm{if}~}
\newcommand{\txtfor}{~\mathrm{for}~}
\newcommand{\txtst}{~\mathrm{s.t.}~}
\newcommand{\txtand}{~\mathrm{and}~}

\renewcommand{\th}{^{\mathrm{th}}}

\newcommand{\one}[1]{\mathds{1} \{ #1\}}

\newcommand{\mytimes}{\times}
\newcommand{\myqed}{\hfill $\square$}

\title{Spatio-Temporal Pricing for Ridesharing Platforms%
\thanks{
Thanks to Itai Ashlagi, Moshe Babaioff, Yakov Babichenko, John Byers, Yeon-Koo Che, Peter Frazier, Keith Chen, Amos Fiat, Sergey Gitlin, Ramesh Johari, Myrto Kalouptsidi, Cinar Kilcioglu, Fuhito Kojima, Scott Kominers, Robin Lee,  Kevin Leyton-Brown, Shengwu Li, Jake Marcinek, Reshef Meir, Paul Milgrom, Hamid Nazerzadeh, Michael Ostrovsky, Ariel Pakes, Garrett van Ryzin, Amin Saberi, Lior Seeman, Peng Shi, Rakesh Vohra, E. Glen Weyl, Adam Wierman, %
conference and seminar participants, and anonymous referees and editors for helpful feedback.
Ma was supported by a Siebel scholarship. Fang was partially supported by Harvard Center for Research on Computation and Society fellowship.
}
\author{Hongyao Ma%
\thanks{Decision, Risk, and Operations Division, Columbia Business School, 3022 Broadway, Uris Hall 423,
New York, NY, 10027, USA. Email: hongyao.ma@columbia.edu.} 
\and Fei Fang%
\thanks{School of Computer Science, Carnegie Mellon University, Hamerschlag Dr, 4126 Wean Hall, Pittsburgh, PA 15213, USA. Email: feifang@cmu.edu.}
\and David C. Parkes%
\thanks{John A. Paulson School of Engineering and Applied Sciences, Harvard University, 33 Oxford Street, Maxwell Dworkin 229, Cambridge,MA 02138, USA. Email: parkes@eecs.harvard.edu.}
}
}

\begin{document}

\maketitle

\begin{abstract}
Ridesharing platforms match drivers and riders to trips, using dynamic prices to balance supply and demand. A challenge is to set prices that are appropriately smooth in space and time, so that drivers with the flexibility to decide how to work will nevertheless choose to accept their dispatched trips, rather than drive to another area or wait for higher prices or a better trip. 
In this work, we propose a complete information model that is simple yet rich enough to incorporate spatial imbalance and temporal variations in supply and demand--- conditions that lead to market
failures in today's platforms.
We introduce the {\em Spatio-Temporal Pricing (STP) mechanism}. The mechanism is incentive-aligned, in that it is a subgame-perfect equilibrium for drivers to always accept their trip dispatches.
From any history onward, the equilibrium outcome of the STP mechanism is welfare-optimal, envy-free, individually rational, budget balanced, and core-selecting.
We also prove the impossibility of achieving the same economic properties in a dominant-strategy equilibrium. Simulation results show that the STP mechanism can achieve substantially improved social welfare and earning equity than a myopic mechanism.
\end{abstract}

\section{Introduction} \label{sec:intro}

Uber connected its first rider to a driver in San Francisco in the summer of 2010. %
Within one decade's time, ridesharing platforms such as Uber and Lyft have radically changed the way people get around in urban areas. 
Comparing with traditional taxi systems, a distinct feature of ridesharing platforms is the emphasis on reliable transportation. For example, Uber's mission is stated as
``transportation as ubiquitous and reliable as running water, everywhere, for everyone''~\citep{travis,ubermissionChicago}, and
Lyft's mission is ``to provide the best, most reliable service possible by making sure drivers are on the road, when and where you need them most'' \citep{lyftmission}.
When demand exceeds supply, these platforms use dynamic ``surge" pricing to guarantee that rider wait times would not exceed a few minutes~\citep{rayle2014app}.

In addition to reliability for riders, the platforms also provide the flexibility for drivers to drive on their own schedule. Uber, for example, advertises itself as ``work that put you first--- drive when you want, earn what you need" \citep{uberdrive}, 
and Lyft promises drivers ``To drive or not to drive? It's really up to you"~\citep{lyftdrive}.
The ``real-time flexibility'' to decide when and where to drive is an important reason that drivers drive for Uber~\citep{hall2016analysis}, and substantially increases both driver supply and driver surplus in comparison to alternative, less flexible arrangements~\citep{chen2019value}. More recently, Uber also started to provide drivers in some markets the option to accept only trips they want, based on the trip destinations and expected earnings~\citep{driversSeatPart1}.

Despite their success, there remain a number of problems with the pricing and dispatching rules governing these ridesharing platforms, leading %
to various kinds of market failure and undercutting the endeavor to provide reliable yet flexible transportation. A particular concern, is that trips may be mispriced relative to each other, incentivizing drivers to cherry-pick~\citep{cook2018gender,changeRuleAdjustStrategy,castro2021randomized}.\footnote{There are also other incentive problems, including inconsistencies across classes of service, competition among platforms, drivers' bonuses and off-platform incentives. In the interest of simplicity, we only model a single class of service and ignore cross-platform competition.}
Drivers can also benefit %
from strategic behavior in the following scenarios, where there is spatial imbalance and temporal variation of supply and demand:

\vspace{-0.2em}
\begin{enumerate}[$\bullet$]
	\setlength\itemsep{0.0em}
	\item (Spatial mispricing) When the price is substantially higher for trips that start in location $A$ (e.g. downtown Austin as in Figure~\ref{fig:ride_austin_space}) in comparison to an adjacent location $B$ (e.g. neighborhoods in the north and across the river), drivers in location $B$ that are close to the boundary can usefully decline trips. This spatial mispricing leads to drivers' ``chasing the surge''--- turning off a ridesharing app while relocating to another location~\citep{dontchase,keithChenTalk}. %
	\item (Temporal mispricing) When large events such as a sports game will soon end, and right before bars are required to close, drivers can anticipate that prices will increase substantially in order to balance supply and demand (see Figure~\ref{fig:ride_austin_time}). In this case, many drivers around the stadium or downtown will decline trips %
	and even go off-line in order to wait in place~\citep{gridwise}. %
	\item (Network externalities) The origin-based ``surge pricing" used by many platforms does not correctly factor market conditions at the destination of a trip. %
	This incentivizes drivers to decline trips to destinations where the continuation payoffs are low, e.g. quiet suburbs with low prices and long wait times~\citep{cancellationMarketWatch,Uber2021NetworkPricingBlog}. %
\end{enumerate}

\newcommand{\imageHeight}{1.5}

\begin{figure}[t!]
\vspace{-0.5em}
\centering
\begin{subfigure}[t]{0.48\textwidth}
	\centering
    \includegraphics[height= \imageHeight in]{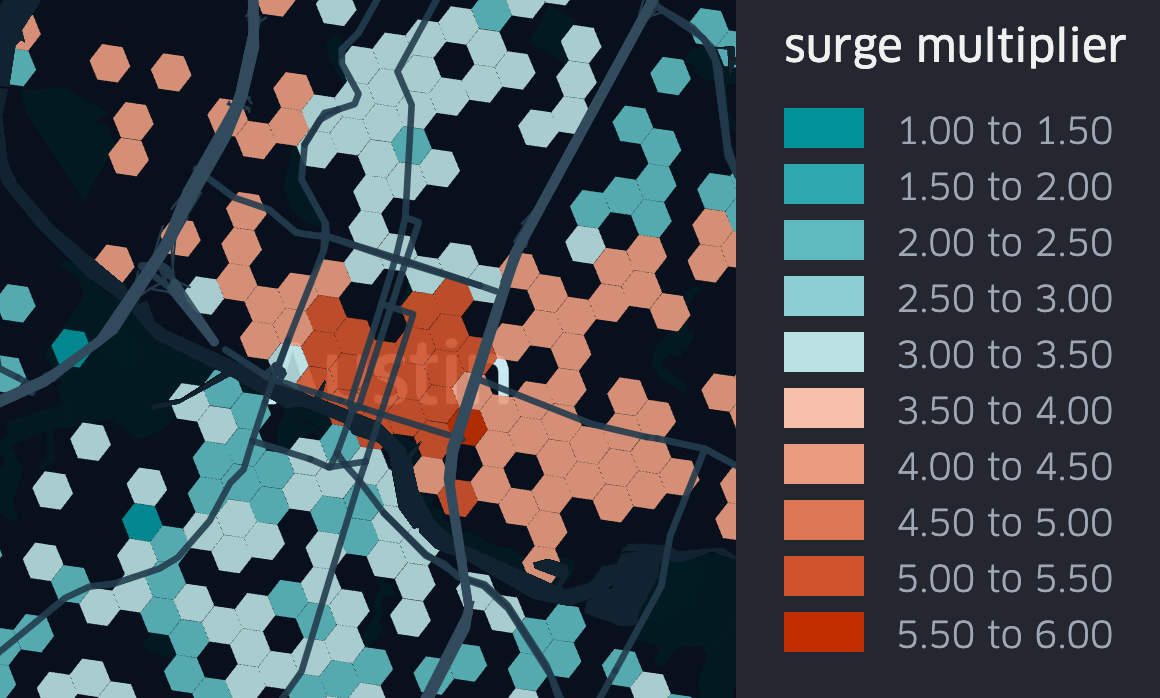}
    \caption{Spatial mispricing.\label{fig:ride_austin_space}}
\end{subfigure}%
\begin{subfigure}[t]{0.48 \textwidth}
	\centering
    \includegraphics[height=\imageHeight in]{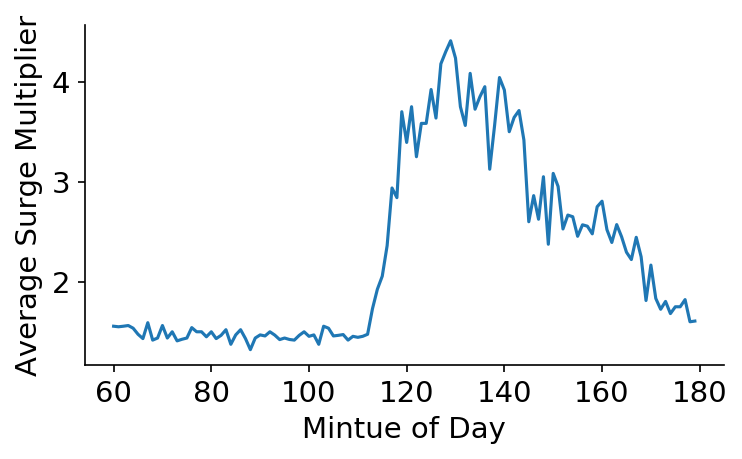}
    \caption{Temporal mispricing. \label{fig:ride_austin_time}}
\end{subfigure}%
\caption[Compact for the LOF]{
Illustration of pricing problems %
using data made public by %
RideAustin.\footnotemark\hspace{-0.01em} %
(a) shows %
the surge multiplier by trip origin, averaged over trips starting between 1am-3am on %
New Year's Eve 2017. (b) plots the average surge multiplier for trips originating from downtown Austin on March 17, 2017 (during the South by Southwest festival). Bars in Austin are required %
to close at 2am.  \vspace{-0.5em}
\label{fig:ride_austin_smoothness_failures}    
}
\end{figure}

\footnotetext{\url{https://data.world/ride-austin}}

These pricing problems undercut the mission of reliable transport, with even high willingness-to-pay riders unable to get access to reliable service for certain trips, such as trips leaving the stadium before a game ends, and trips going to a quiet suburb.    
These kinds of mispricing also makes it difficult %
to accommodate drivers' idiosyncratic preferences, for example over locations, since such features are used extensively by drivers to strategize for better earnings~\citep{destinationFilterDrop}. 
This can also lead to inequity, with demonstrated learning effects leading to differences in drivers' long-run earnings (e.g. a gender gap in driver hourly earnings~\citep{cook2018gender}), with potential consequences around driver churn from the platform.

Simple fixes by limiting dispatching transparency or drivers' flexibility are not fully effective. 
For example, when a platform hides trip destinations %
before the pick-up, experienced drivers will call riders to ask about trip details, and cancel those trips that are not worthwhile~\citep{cook2018gender}. 
Nor does the imposition of penalties on drivers solve these problems, since drivers may decide to go offline, or choose not to participate in the platform from certain locations or times. %

We conceptualize many of the problems with today's platforms as arising from prices failing  to be appropriately ``smooth'' in space and time--- if prices for trips are higher in one location then they should be appropriately higher in adjacent locations; if demand would soon increase in a location then the current prices should already be appropriately higher; and if destinations differ in continuation payoffs then trip prices to these destinations need to reflect this.
With appropriately smooth prices that correctly reflect the on-trip
cost and network externality of completing each trip, drivers who retain the
flexibility to decide how to work will still choose to accept any trip to which they are dispatched, providing reliability for riders. 
Many of the strategic behaviors %
are also symptoms of inefficiencies in dispatching, for example when dispatching drivers to low-priced trips that send them away from a sports stadium, five minutes before a game ends. 
Correctly designed, ridesharing platforms can succeed in optimally orchestrating trips and providing reliable transpiration for riders, while leaving drivers with flexibility to decide how to work. %

\vsq{-0.1em}

\subsection{Our results} \label{sec:intro_contribution}

In this work, we propose a framework for studying pricing and dispatching in the context of a ridesharing
platform. 
The model that we introduce is simple, but is the first in the literature to be rich enough to incorporate spatial imbalance and temporal variation of supply and demand--- market conditions under which mispricing and strategic behavior arise in today's platforms. 
We propose the \emph{Spatio-Temporal Pricing mechanism} (STP), which achieves the following properties that we consider very important for a sharing economy platform:
\vsq{-0.1em}
\begin{enumerate}[$\bullet$]
	\setlength\itemsep{0.0em}
	\item Welfare-optimality: maximizing total rider values minus driver costs. 
	\item Incentive-alignment: the prices are appropriately smooth in space and time, such that drivers will always choose to accept any dispatched trips.
	\item Robustness: the mechanism updates the downstream plans after deviations from its %
	dispatches.
	\item Temporal-consistency: plans are computed and updated based on the current state but not past history, without using penalties or time-extended contracts.
	\item Envy-freeness: drivers at the same location and time do not envy each other's future payoffs; riders requesting the same trips do not envy each other's outcomes.
	\item Core-selecting: no coalition of riders and drivers can make a better plan among themselves.	
\end{enumerate}

Welfare-optimality and incentive-alignment are necessary, but are not enough to guarantee efficient and reliable operation. Robustness ensures the other properties from any point of time onward regardless of past deviations. This is important in the face of erroneous predictions, mistakes by participants, or unmodeled idiosyncratic preferences, and without robustness any solution would be brittle and poorly suited to practice. Temporal-consistency matters, since using penalties (or threatening to fire drivers, or shut down the system) is incompatible with the spirit of the sharing economy, and the real-time flexibility of being able to choose how to work. 
Envy-freeness and core-selecting properties relate to fairness, and to
the long-run health of the marketplace. An envy-free mechanism removes
unnecessary fluctuations in daily income that depend on lucky
dispatches, and reduces the long-run inequity in earnings from
``learning-by-doing.''%
A core outcome ensures that no group of riders and drivers are %
disadvantaged, getting a worse outcome than the best plan they can form among themselves. %
In practice, a violation of the core makes a platform vulnerable to a competitor that may %
take-over such disadvantaged parts of the network. %

\paragraph{The model.} We work in a complete information, discrete time, multi-period and multi-location model, addressing the challenge of promoting desirable behavior by drivers in the absence of time-extended contracts. %
At the beginning of each time period, based on the history, current positioning of drivers, and current and future demand, the STP mechanism dispatches each available driver to a rider trip, or to relocate, or to exit the platform for the planning horizon. The mechanism also determines a payment to be made if the driver follows the dispatch. Each driver then decides whether to follow the dispatch, or to decline and stay, or to relocate to any location, or to exit. After observing the driver actions in a period, the mechanism collects payments from the riders who are picked-up, and makes payments to the drivers who followed the dispatches.
The main assumptions that we make are:   
\vsq{-0.2em}
\begin{enumerate}[(i)]
	\setlength\itemsep{0.0em}
	\item Complete information about supply and demand over a finite planning horizon,
	\item Impatient, price-taking riders, with a value for being picked-up at a particular time and location (and without preferences over drivers), and
	\item Drivers who each face the same costs for completing the same trip from some origin to some destination at some particular time (and without idiosyncratic preference over riders or locations), and who are willing to provide trips until the end of the planning horizon. %
\end{enumerate}

We do allow for heterogeneity in rider values and trip details (the origin, destination, and time of a trip). For drivers, we allow them to become available at different times and locations, and we also model the distinction between a driver who is already driving in the platform (for example, finishing a trip), and a driver who has not yet started driving and thus needs to make an entry decision (for example, dropping off a child at school at a specific location and time, and willing to drive afterwards). We also allow a driver who is asked to exit the platform earlier than their intended exit time to incur a one-time exit cost, modeling the forgone opportunity of outside options after the driver has been driving in the platform for some time.

\paragraph{Main results.} 

We first show that the welfare-optimal planning problem can be reduced to a minimum cost flow (MCF) problem. The integrality of the linear program (LP) of  MCF guarantees the existence of anonymous (not depending on the identity of rider or driver), origin-destination, competitive equilibrium (CE) prices, allowing the price of a trip to depend on market conditions at both the origin and destination.
The lattice structure of the dual LP also implies that drivers' total utilities among all CE plans form a lattice. 
The STP mechanism uses {\em driver-pessimal CE prices}, computing a driver-pessimal CE plan at the beginning of the planning horizon, as well as after any deviations from the current plan. 
This induces an extensive-form game among the drivers, where the total payoff to each driver is determined by the mechanism's dispatch and payment rules, as well as the actions taken by the other drivers. 
The main result is that the STP mechanism satisfies all the desiderata outlined above.

The STP mechanism uses driver-pessimal rather than driver-optimal CE prices. 
The driver-optimal analog of the STP mechanism%
reflects the payments of a Vickrey-Clarke-Groves (VCG) mechanism, but fails to align incentives. %
In a driver-pessimal plan, a driver's continuation payoff from some location and time onward is equal to the welfare-gain in the economy from adding an additional driver at this location and time.
The $M^\natural$ concavity of the MCF problems~\citep{murota2003discrete} 
implies a stronger substitution between drivers at the same location at the same time (than between drivers at different locations or times), allowing us to reason about the welfare gains and prove that accepting the mechanism's dispatches %
forms a subgame-perfect equilibrium.
Under a driver-optimal plan, each driver's continuation payoff is equal to the welfare-loss in the economy from losing the driver. %
This ``marginal product'' can increase over time, as the set of trips that can be
completed by the rest of the drivers becomes smaller. 
This breaks incentives so that a driver can usefully deviate from a suggested dispatch
and trigger a plan update.

We also prove an impossibility result, that no dominant-strategy mechanism has the same economic properties. 
For three stylized scenarios (the end of an event, the morning rush hour, and trips to and from the airport with unbalanced flows), we compare the STP mechanism with a \emph{myopic pricing mechanism} that simply clears the market for each location at each time, without taking future demand or supply into consideration. 
Extensive simulation results %
suggest that STP achieves substantially higher social welfare, and highlight the failure of incentive alignment and the large variance in driver earnings due to non-smooth prices in myopic mechanisms.

The main operational insight from this paper is that each trip should
be priced (as it is under the STP mechanism) as the welfare
contribution of an extra driver at the origin of the trip, minus the
welfare contribution of an extra driver at the destination of the
trip, plus the cost to a driver to complete this trip. This
 correctly ``prices in'' the externality imposed on the system
by moving a driver from the origin to the destination.
As an example, an
extra driver at a stadium five minutes before a game ends is very
valuable, so with this approach to pricing, the price will go up
smoothly before a game ends, providing riders with high
willingness-to-pay reliable access to transportation.
Prices  determined in this way will also be smooth in space, reducing drivers'
incentives to ``chase the surge.''
  From this insight comes the
  importance to practice of  estimating the welfare contribution of an extra
  driver,  likely the expected welfare contribution in  an actual
  deployment, using
  data from the current, typically suboptimal, operations.

\subsection{Related Work} \label{sec:related_work}

To the best of our knowledge, this current paper is unique in that it considers both multiple locations and multiple time periods, along with rider demand, rider willingness-to-pay, and driver supply that can vary across both space and time.

\citet{banerjee2015pricing} adopt a queuing-theoretic approach in analyzing the effect of dynamic pricing on the revenue and throughput of ridesharing platforms, assuming %
stationary system state. %
When the platform correctly estimates supply and demand, the optimal dynamic pricing strategy %
does not achieve better performance than the optimal static pricing strategy. However, dynamic pricing is more robust to fluctuations and to mis-estimation of system parameters. %

By analyzing %
a continuum model %
 with stationary demand and unlimited driver supply at fixed costs, \citet{Bimpikis2016} %
show that a %
platform's profit is maximized when the demand pattern across %
locations is balanced. They show in simulation that in comparison to a single price, origin-based pricing improves profit, while there is not a substantial gain %
from using origin-destination based pricing. Our model is quite distinct, in that it is not a continuum or stationary model, does not have unlimited driver supply, %
and is focused on welfare.
\citet{banerjee2017pricing} model a shared vehicle system as a continuous-time Markov chain, and establish approximation guarantees for a static, state-independent policy %
w.r.t. the optimal, state-dependent policy.

\citet{castillo2017surge} study the ``wild goose chase" phenomena in platforms that myopically dispatch the closest drivers to rider requests. When demand much exceeds supply, drivers spend too much time driving to pick up riders, leading to long wait times and decreased welfare and revenue. Under a model with stationary supply and demand,  %
the authors establish the importance of dynamic pricing for these %
myopic dispatching schemes, in keeping enough open cars to avoid inefficient long pick-ups. %
\citet{yan2020dynamic} provide a review of matching and dynamic pricing in ridesharing platforms. 
\citet{garg2020driver} study a dynamic stochastic model, and show that when the ``surge'' pricing needs to be origin-based and cannot depend on the length of the trip, additive surge is more incentive compatible for drivers in practice than is multiplicative surge.

There are various empirical studies of the Uber platform as a two-sided marketplace~\citep{hall2017labor,hall2016analysis,cohen2016using}, analyzing the labor market of Uber's drivers, the longer-term labor market equilibration, and consumer surplus. By analyzing drivers' hourly earnings, \citet{chen2019value} show that drivers' reservation wages vary significantly over time, and that the real-time flexibility of being able to choose when to work increases %
driver surplus and driver supply.
In regard to dynamic pricing, \citet{chen2015dynamic} show 
that surge pricing increases the supply of drivers at times when the surge pricing is high, and \citet{lu2018surge} show that surge pricing incentivizes drivers to relocate to higher surge areas.
A case study into an outage of Uber's surge pricing during the 2014-2015 New Year's Eve %
found a large increase in riders' waiting times,  %
and a large decrease in the percentage of requests completed~\citep{hall2015effects}.

\smallskip

In this work, we made use of the connections between LP duality %
and market equilibrium to prove the existence and structure of CE. %
These connections are widely applied in many matching and assignment settings~\citep{shapley1971assignment,bikhchandani2002package,parkes2000iterative}. The existing literature, however, %
does not provide a framework to study the dynamic incentive properties
in our setting, in regard to drivers' decisions about accepting dispatches.

In combinatorial auctions,
the connection between LP duality and CE allows the design of
iterative auctions that can be interpreted as primal-dual algorithms
that solve the optimal allocation
problem~\citep{parkes2000iterative,mishra2007ascending}.
The challenge is to elicit valuation functions, and
these designs induce
an extensive-form game among agents who participate in each round of
the auction, and before an allocation is determined.
This is different from our setting,
where the challenge is not asymmetric information but rather
agency, and where drivers' strategic behavior is realized
after a plan has been computed.

       	The {\em matching with substitutes} literature
        \citep{kelso1982job,gul1999walrasian} establishes the
        existence and lattice structure of CE payoffs, for
        two-sided matching where the preference of an agent over the
        agents %
        on the other side satisfy the {\em gross substitutes} (GS) condition. Our problem cannot be considered as two-sided matching with substitutes, since from a driver's perspective, riders on different segments of the same path may be complements to each other.
	
The literature on {\em trading networks} studies economic models where agents in a network can transact via bilateral contracts~\citep{hatfield2013stability,hatfield2015chain,ostrovsky2008stability}. %
A CE exists when preferences satisfy a {\em full substitution}
property, and the utilities of agents on either end of an acyclic
network forms a lattice.
Although the optimal planning problem we study can be
reduced to a trading network, where drivers and riders trade the right
to use a car for the rest of the planning horizon, it is not
  apparent how to use this mapping to establish the incentive
  properties of a ridesharing mechanism
(see Appendix~\ref{appx:trading_network}). %

Dynamic variations of the VCG mechanism~\citep{athey2013efficient,bergemann2010dynamic,cavallo2009efficient} are able to truthfully implement efficient decision policies, where agents receive private information over time. The payment to an agent in a single period in the dynamic VCG mechanism is equal to the flow marginal externality imposed on the other agents by its presence in the current period~\citep{cavallo2009efficient}.
These mechanisms are not suitable for our problem, because the existence of a driver for only one period may exert negative externality on the rest of the economy by inducing suboptimal positioning of the rest of the drivers in the subsequent time periods. As a result, some drivers may be paid negative payments for certain periods of time, which would lead them to decline dispatches. See Appendix~\ref{appx:dynamic_vcg} for examples and detailed discussions.

Principal-agent problems are, of course, studied extensively in
contract
theory~\citep{bolton2005contract,salanie2005economics}. When
information asymmetry arises before the time of contracting,
with private information, this is a problem of {\em adverse
  selection}.
When information asymmetry arises after the time of contracting,
through hidden actions, this is a problem of {\em moral hazard}. 
Where contracts cannot be perfectly enforced, {\em relational incentive contracts}~\citep{levin2003relational} %
are self-enforcing by threatening to terminate an agent following poor performance. 
In our model, there is neither private information nor hidden
actions. Rather, the challenge that we face is one of incentive alignment in the absence of time-extended contracts, so that %
drivers retain the flexibility  to decide on actions without incurring penalties or facing termination threats.

\section{Preliminaries} 
\label{sec:preliminaries}

Let $\horizon$ be the length of the planning horizon, starting at time $t = 0$ and ending at time $t = \horizon$. We adopt a discrete time model, and refer to each time point $t$ as ``time $t$", and call the duration between time $t$ and time $t + 1$ a \emph{time period}. We may think about each time period as $\sim5$ minutes, and with $\horizon = 6$ the planning horizon would be half an hour. Trips start and end at time points. 
Denote $\timeSet = \{0, 1, \dots, \horizon \}$ and $\activeTimeSet = \{ 0, 1, \dots, \horizon - 1 \}$.

Let $\loc = \{A, B, \dots, \}$ be a set of $|\loc|$ discrete
locations, and we adopt $a$ and $b$ to denote generic locations. For all $a,b \in \loc$ and $t \in \timeSet$, the triple $(a,b,t)$ denotes a \emph{trip} with origin $a$, destination $b$, starting at time $t$.
Each trip can represent (i) taking a rider from $a$ to $b$ at time $t$, (ii) relocating without a rider from $a$ to $b$ at time $t$, and (iii) staying in the same location for one period of time (in which case $a = b$). 
Let the distance $\dist: \loc \mytimes \loc \rightarrow \setN $ be the number of time periods needed to travel between locations, so that trip $(a,b,t)$ ends at $t + \dist(a, b)$.\footnote{We can also allow the distance between a pair of locations to change over time, modeling the changes in traffic conditions, i.e. a trip from $a$ to $b$ starting at time $t$ ends at time $t + \dist(a,b,t)$. This does not affect the results presented in this paper, and we keep $\dist(a,b)$ for simplicity of notation.} 
We allow $\dist(a, b) \neq \dist(b, a)$ for locations $a \neq b$, modeling asymmetric traffic flows.
We assume $\dist(a, b) \geq 1$ for all $a, b \in \loc$, and $\dist(a,a) = 1$ for all $a \in \loc$. 
Set $\trips \triangleq \set{(a,b,t)}{a \in \loc, ~b \in \loc, ~ t \in \{0, 1,\dots, \horizon - \dist(a,b)\}}$ denotes the set of all feasible trips within the planning horizon.

Let $\driverSet$ denote the set of drivers, with $\nd \triangleq |\driverSet|$. Each driver $i \in \driverSet$ is characterized by
{\em type} $\theta_i = (\driverEntrance_i, \re_i, \te_i, \tl_i)$--- driver $i$ is able to enter the platform at location $\re_i$ and time $\te_i$, and plans to exit the platform at time $\tl_i$ (with $\te_i < \tl_i$).
$\driverEntrance_i$ indicates driver $i$'s entrance status. 
A driver with $\driverEntrance_i = 0$ has not yet entered the platform, and needs to make an entry decision--- consider a driver who is willing to drive after dropping her daughter at school at location $\re_i$ at time $\te_i$.
A driver with $\driverEntrance_i = 1$ has already entered the platform (she may be completing an earlier trip, or relocating to another location), and will become available to pick up again at $(\re_i, \te_i)$ . 
Here we make the assumption (S1) that {\em driver types are known to
  the mechanism and that all drivers stay until at least the end of
  the planning horizon, and do not have a
  preference over riders or location, including where they finish their last trip in the planning horizon.}

A driver who completes a trip $(a,b,t) \in \trips$ incurs a cost $\cost_{a,b,t} \geq 0$, which models the cost of time, driving, fuel, wear-and-tear, etc. A driver who has already entered the platform may exit earlier than her intended exit time, in which case she will not be able to complete any trip in the remainder %
of this planning horizon. Exiting $\Delta$ periods earlier than time $\horizon$ incurs a one-time cost of $\exitCost_\Delta \geq 0$ (with $\exitCost_0 = 0$), modeling the forgone opportunity of outside employment options, after driving for the platform for some time. A driver with $\driverEntrance_i = 0$ who does not enter the platform at $(\re_i, \te_i)$ does not incur any cost, and will not enter at a later time. %
Drivers have quasi-linear utilities, and seek to maximize the total payments received over the planning horizon minus the total costs.

Denote $\riders$ as the set of riders, each intending to take a single trip during the planning horizon.  The {\em type} of rider $j \in \riders$ is $(\origin_j, \dest_j, \tr_j, \val_j)$, where $\origin_j$ and $\dest_j$ are the trip origin and destination, $\tr_j$ the requested start time, and $\val_j\geq 0$ the {\em value} for the trip.\footnote{
The value $\val_j$ models the rider's willingness-to-pay over and above a base payment that covers the additional cost that a driver incurs for picking up a rider in comparison to just relocating (e.g., extra wear-and-tear, loss of privacy, inconvenience). This base payment is always made for a matched trip, and allows us to model a driver's cost as depending on  the origin, destination, and time of a trip, and irrespective of whether there is a passenger in the car. With this, the prices we determine are the amount to  pay on top of the base amount.}
We assume (S2) {\em that riders are impatient, only value trips starting at $\tr_j$, are not willing to relocate or walk from a drop-off point to their actual, intended destination, and do not have preference over drivers.}
Rider utility is quasi-linear, with utility $\val_j - \price$ to rider $j$ for a trip at (incremental to base) price $\price$.

We assume the platform has complete information about supply and demand over the planning horizon (travel times, trip costs, driver and rider types, including driver entry during the planning horizon). We assume drivers have the same information, and that this is common knowledge amongst drivers (more generally, it is sufficient that it be common knowledge amongst drivers that the platform has the correct information).
Unless otherwise noted, we assume properties~(S1), (S2), and complete, symmetric information throughout the paper. Detailed discussions on the effect of relaxing these assumptions are provided in Section~\ref{sec:conclusion}.

\medskip

At each time $t$, a driver is \emph{en route} if she started her last trip from $a$ to $b$ at time $\tprime$ (with or without a rider), and $t < \tprime + \dist(a,b)$. A driver is \emph{available} if she has entered or is able to enter the platform, and has not yet exited, and is not \emph{en route}. 
A driver who is available at time $t$ and location $a$ is able to complete a pick-up at this location and time. We allow a driver to drop-off a rider and pick-up another rider in the same location at the same time point (see Appendix~\ref{appx:cont_time}).

A {\em path} is a sequence of tuples $(a,b,t)$, representing driver entrance, exit, and the trips she takes over the planning horizon. Let $\pathSet_i$ denote the set of all \emph{feasible paths} of driver $i$, with $\path_{i,k}\in \pathSet_i $ to denote her $k\th$ feasible path. The path $\path_{i,0}$ includes no trip: for a driver with $\driverEntrance_i = 0$,  $\path_{i,0}$ models the option to not enter the platform at all; for a driver s.t. $\driverEntrance_i = 1$,  $\path_{i,0}$ models the option exit immediately at $(\re_i, \te_i)$. For each $k = 1, \dots, |\pathSet_i|$, $\path_{i,k}$ is a path that starts at $(\re_i, \te_i)$, with the starting time and location of each successive trip equal to the ending time and location of the previous trip. Denote $(a,b,t) \in \path_{i,k}$ if path $\path_{i,k}$ includes (or {\em covers}) trip $(a,b,t)$, and let $\pathCost_{i,k}$ be the total cost of the $k\th$ path to driver $i$. We know that $\pathCost_{i,0} = 0$ if $\driverEntrance_i = 0$, $\pathCost_{i,0} = \exitCost_{T - \te_i}$ if $\driverEntrance_i = 1$, and for $k>1$, $\pathCost_{i,k} = \sum_{(a,b,t)\in  \path_{i,k}} \cost_{a,b,t} + \exitCost_\Delta$, if path $\path_{i,k}$ ends $\Delta$ periods earlier than $T$.

Driver $i$ who takes the path $\path_{i,k}$ is able to pick up rider $j$ if $(\origin_j, \dest_j, \tr_j) \in \path_{i,k}$, however, a path specifies only the movement in space and time, and does not specify whether a rider is picked up for each of the trips on the path. 
Let an \emph{action path} for driver $i$ be a sequence of tuples, each of them can either be of the form $(a,b,t)$, representing a relocation trip from $a$ to $b$ at time $t$ without a rider, or be of the form $(a,b,t,j)$, in which case the driver sends rider $j$ from $a$ to $b$ at time $t$ (thus requiring $(a,b,t) = (\origin_j, \dest_j, \tr_j)$). 
Let $\actPathSet_{i}$ be the set of all feasible action paths of driver $i$ (the feasibility of an action path is similar to that of a path).
For an action path $\actpath_i \in \actPathSet_{i}$, denote $(a,b,t) \in \actpath_i$ or $(a,b,t,j) \in \actpath_i$ if the action path includes a relocation or rider trip from $a$ to $b$ at time $t$. A driver taking action path $\actpath_i$ that is \emph{consistent} with path $\path_{i,k}$ (i.e. results in the same movement in space and time) incurs a total cost of $\pathCost_{i,k}$. 

\begin{example}
\label{exmp:toy_econ_1}

\newcommand{\nodeScaleI}{0.8}
\newcommand{\hdistI}{3cm}
\newcommand{\vdistI}{1.9cm}	

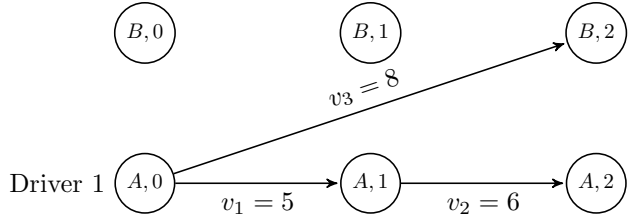
\begin{figure}[t!]
\vspace{-0.5em}
\centering
\begin{tikzpicture}[->,>=stealth',shorten >=1pt,auto,node distance=2cm,semithick][font = \small]
\tikzstyle{vertex}=[fill=white,draw=black,text=black,scale=0.9]

\node[state]         (A0) [scale = \nodeScaleI] {$A,0$};
\node[state]         (B0) [above of=A0, node distance = \vdistI, scale=\nodeScaleI] {$B,0$};
\node[state]         (A1) [right of=A0, node distance = \hdistI, scale=\nodeScaleI] {$A,1$};
\node[state]         (B1) [right of=B0, node distance = \hdistI, scale=\nodeScaleI] {$B,1$};
\node[state]         (A2) [right of=A1, node distance = \hdistI, scale=\nodeScaleI] {$A,2$}; 
\node[state]         (B2) [right of=B1, node distance = \hdistI, scale=\nodeScaleI] {$B,2$};

\path (A0) edge	node[pos=0.5, sloped, below] {$\val_1=5$} (A1);
\path (A1) edge	node[pos=0.5, sloped, below] {$\val_2=6$} (A2);
\path (A0) edge	node[pos=0.5, sloped, above] {$\val_3=8$} (B2);

\node[text width=3cm] at (-0.3, 0) {Driver 1};

\end{tikzpicture}
\caption{The economy in Example~\ref{exmp:toy_econ_1}, with two locations $A$, $B$, two time periods and three riders.  \vspace{-0.3em}
\label{fig:toy_econ_1}}
\end{figure}

The planning horizon is $\horizon = 2$ and there are two locations $\loc = \{A, B\}$ with distance $\dist(A,A) = \dist(B,B) = 1$ and $\dist(A,B) = \dist(B,A) = 2$. See Figure~\ref{fig:toy_econ_1}. Trip costs are $2$ per period of time, i.e. $\cost_{a,b,t} = 2 \dist(a,b)$ for all $(a,b,t) \in \trips$, and the opportunity cost of exiting early is $\exitCost_\Delta = \Delta$.
There is one driver, who has not yet entered the platform (i.e. $\driverEntrance_1 = 0$), but is able to enter at time $\te_1 = 0$ at location $\re_1 = A$, and plans to leave at time $\tl_1 = 2$. There are three riders with:
\vspace{-1em}
\begin{multicols}{2}
\begin{enumerate}[$\bullet$]
	\setlength\itemsep{0.0em}
	\item Rider 1: $\origin_1 = A$, $\dest_1 = A$, $\tr_1 = 0$, $\val_1=5$, 
	\item Rider 2: $\origin_2 = A$, $\dest_2 = A$, $\tr_2 = 1$, $\val_1=6$,
	\item Rider 3: $\origin_3 = A$, $\dest_3 = B$, $\tr_3 = 0$, $\val_3 = 8$.
\end{enumerate}
\end{multicols}
\vspace{-0.8em}

In addition to not entering the platform at all, which corresponds to path $\path_{1,0}$ with cost $\pathCost_{1,0} = 0$, there are three more feasible paths for driver 1: $\path_{1,1} = ((A,A,0),~(A,A,1))$, $\path_{1,2} = ((A,B,0))$, and $\path_{1,3} = ((A,A,0))$. 
In $\path_{1,3}$, the driver exits one period before the end of planning horizon. The path costs are $\pathCost_{1,1} = \cost_{A,A,0} + \cost_{A,A,1} = 2$, $ \pathCost_{1,2} = \cost_{A,B,0} = 4$, and $\pathCost_{1,3} = \cost_{A,A,0} + \exitCost_1 = 3$. Path $((A,A,0),~(A,B,1))$ is infeasible, since the last trip ends later than the driver's leaving time. Similarly, paths $((A,B,0),~(B,B,1))$ and $((A,A,0), (B,B,1))$ are infeasible.

In addition to not entering, there are eight feasible actions paths of rider $1$. $((A,B,0))$, relocating from $A$ to $B$ at time $0$, and $((A,B,0,3))$, sending rider $3$ from $A$ to $B$ at time $0$, are both consistent with the path $\path_{1,2}$, and both have cost $4$. 
Four action paths, $((A,A,0),(A,A,1))$, $((A,A,0,1),(A,A,1))$, $((A,A,0),(A,A,1,2))$, $((A,A,0,1),(A,A,1,2))$, are consistent with $\path_{1,1}$ and have cost $4$. Both $((A,A,0))$ and $((A,A,0,1))$ are consistent with $\path_{1,3}$ and have cost $3$. 
\end{example}

We now provide an informal timeline of a {\em ridesharing mechanism} (see Section~\ref{sec:st_pricing} for a formal definition). 
At each time point $t \in \activeTimeSet$, given the history of trips, current positioning and availability of drivers, and current and future driver supply and rider demand for trips:
\begin{enumerate}[1.]
	\setlength\itemsep{0.0em}

	\item The ridesharing mechanism determines for each rider with trip start time $t$, whether a driver will be dispatched to pick her up, and if so, the price of her trip.
	\item The mechanism dispatches available drivers to pick up riders, to relocate, or to exit (for drivers already in the platform), or not to enter (for drivers who have not entered, with $\te_i = t$ and $\driverEntrance_i = 0$). The mechanism also determines the payments offered to drivers for accepting the dispatches.
	\item Each available driver decides whether to accept the dispatch, or to deviate and either stay in the same location, or relocate, or exit/not enter. A driver may still decide to enter the platform even if asked not to do so. %
	The mechanism collects and makes payments based on driver actions.
\end{enumerate}

Any undispatched, available driver makes their own choices of actions. %
We assume that any driver already \emph{en route} will continue their current trip. A driver's payment in a period in which the driver declines a dispatch is zero, so that drivers are not charged penalties for deviation.

As a baseline, we define the following myopic pricing mechanism. 
For each rider $j \in \riders$, denote the per-period surplus of her trip as $\surplus_j \triangleq (\val_j -\cost_{\origin_j, \dest_j, \tr_j})/\dist(\origin_j, \dest_j)$.

\begin{definition}[Myopic pricing mechanism] \label{defn:myopic_mech}
At each time point $t \in \activeTimeSet$, for each location $a \in \loc$, the {\em myopic pricing mechanism} dispatches available drivers at $(a,t)$ to riders %
  with $(\origin_j,\tr_j) = (a,t)$ and $\surplus_j \geq 0$, in decreasing order of $\surplus_j$. The mechanism sets a market clearing rate $\rho_{a,t}$ (i.e. between highest unallocated $\surplus_j$ and lowest allocated $\surplus_j$), and sets prices $\price_{a,b,t} = \dist(a,b) \rho_{a,t} + \cost_{a, b, t}$ for each destination $b \in \loc$, which is offered to all dispatched drivers and collected from all riders. 
\end{definition}

The market clearing prices may not be unique, and a fully defined myopic mechanism must provide a rule for picking a particular set of  prices. This mechanism has anonymous, origin-based pricing, and is very simple in ignoring the need for smooth pricing, or future supply and demand. But its simplicity means that it fails to optimize social welfare or to set prices that are spatially and temporally smooth, leading to various market failures. 

\begin{example}[Super Bowl example] \label{exmp:super_bowl} 

\newcommand{\nodeScaleII}{0.8}
\newcommand{\hdistII}{3.5}
\newcommand{\vdistII}{2}	

\begin{figure}[t!]
\centering
\begin{tikzpicture}[->,>=stealth',shorten >=1pt, auto, node distance=2cm,semithick][font = \small]
\tikzstyle{vertex} = [fill=white,draw=black,text=black,scale=0.9]

\draw[line width=0.25mm] 	(-1, \vdistII*2+0.8) -- (\hdistII * 3 + 1, \vdistII*2+0.8);

\node[text width=1.2cm] at (0, \vdistII*2+1.2) {9:50pm};
\node[text width=1.3cm] at (\hdistII, \vdistII*2+1.2) {10:00pm};
\node[text width=3cm] at (\hdistII + 0.2, 5.7) {Super Bowl Ends};
\node[text width=1.3cm] at (\hdistII*2, \vdistII*2+1.2) {10:10pm};
\node[text width=1.3cm] at (\hdistII*3, \vdistII*2+1.2) {10:20pm};

\node[state]         (A0) [scale = \nodeScaleII] {$A,0$};
\node[state]         (B0) [above of=A0, node distance = \vdistII cm, scale=\nodeScaleII] {$B,0$};
\node[state]         (A0) [scale = \nodeScaleII] {$A,0$};
\node[state]         (C0) [above of=B0, node distance = \vdistII cm, scale=\nodeScaleII] {$C,0$};

\node[state]         (A1) [right of=A0, node distance = \hdistII cm, scale=\nodeScaleII] {$A,1$};
\node[state]         (B1) [right of=B0, node distance = \hdistII cm, scale=\nodeScaleII] {$B,1$};
\node[state]         (C1) [right of=C0, node distance = \hdistII cm, scale=\nodeScaleII] {$C,1$};

\node[state]         (A2) [right of=A1, node distance = \hdistII cm, scale=\nodeScaleII] {$A,2$}; 
\node[state]         (B2) [right of=B1, node distance = \hdistII cm, scale=\nodeScaleII] {$B,2$};
\node[state]         (C2) [right of=C1, node distance = \hdistII cm, scale=\nodeScaleII] {$C,2$};

\node[state]         (A3) [right of=A2, node distance = \hdistII cm, scale=\nodeScaleII] {$A,3$}; 
\node[state]         (B3) [right of=B2, node distance = \hdistII cm, scale=\nodeScaleII] {$B,3$};
\node[state]         (C3) [right of=C2, node distance = \hdistII cm, scale=\nodeScaleII] {$C,3$};

\node[text width=3cm] at (-0.3, 2 * \vdistII + 0.2) {{\color{\driverOneTextColor} Driver 1}};
\node[text width=3cm] at (-0.3, 2 * \vdistII - 0.2) {{\color{\driverTwoTextColor} Driver 2}};
\node[text width=3cm] at (-0.3, \vdistII) {{\color{\driverThreeTextColor} Driver 3}};

\path (C0) edge	node[pos=0.53, sloped, above] {{\color{\driverOneRiderColor}$\boldsymbol{\val_1 \shorteq 20}$},  \hspace{0.7em} {\color{\driverTwoRiderColor} $\boldsymbol{\val_2 \shorteq 30}$}} (B1);
\path (B0) edge	node[pos=0.17, sloped, above] {$\val_3 \shorteq 10$} (C1);
\path (B0) edge	node[pos=0.5, sloped, above] {{\color{\driverThreeRiderColor} $\boldsymbol{\val_4 \shorteq 20}$}} (A1);
\path (B1) edge	node[pos=0.5, sloped, above] {{\color{\driverOneRiderColor} $\boldsymbol{\val_5 \shorteq 20}$} } (B2);
\path (C1) edge	node[pos=0.5, sloped, above] {$\val_6 \shorteq 100$} (B2);

\draw[->]  (3.9, 4)  to[out=-2, in= 120 ](10.15, 0.1);

\node[text width=3cm] at (9, 3.3) {$\val_7 \shorteq 100$};
\node[text width=3cm] at (9.5, 2.9) {$\val_8 \shorteq 90$};
\node[text width=3cm] at (10, 2.5) {$\val_9 \shorteq 80$};

\draw[dashed, \driverthickness, \driverOneColor] (0.27, \vdistII * 2 - 0.3) -- (\hdistII - 0.38, \vdistII + 0.07 );
\draw[dashdotted, \driverthickness, \driverTwoColor] (0.2, \vdistII * 2 - 0.4) -- (\hdistII - 0.38, \vdistII - 0.1  );
\draw[dotted, \driverthickness, \driverThreeColor] (0.27, \vdistII * 1 - 0.3) -- (\hdistII - 0.4,  + 0.07 );

\draw[dashed, \driverthickness, \driverOneColor] (\hdistII + 0.4, \vdistII - 0.15) -- (\hdistII * 2 - 0.35, \vdistII - 0.15);

\draw[dashed, \driverthickness, \driverOneColor](\hdistII * 3 + 1.5, 1) -- (\hdistII * 3 + 2.3, 1);
\draw[dashdotted, \driverthickness, \driverTwoColor](\hdistII*3 + 1.5, 0.5) -- (\hdistII*3 + 2.3, 0.5);
\draw[dotted, \driverthickness, \driverThreeColor](\hdistII * 3 + 1.5, 0) -- (\hdistII * 3 + 2.3, 0);

\node[text width=0.32cm] at (\hdistII * 3 + 2.5, 1.1) {{ \color{\driverOneLgdColor} $z_1$}};
\node[text width=0.32cm] at (\hdistII * 3 + 2.5, 0.6) {{ \color{\driverTwoLgdColor} $z_2$}};
\node[text width=0.32cm] at (\hdistII * 3 + 2.5, 0.1) {{ \color{\driverThreeLgdColor}$z_3$}};

\node[text width=1cm] at (1.3, \vdistII + 0.15) {{\color{gray}\emph{10}}};
\node[text width=1cm] at (2.35, \vdistII + 0.15) {{\color{gray}\emph{10}}};
\node[text width=1cm] at (1.9, \vdistII /2 - 0.4) {{\color{gray}\emph{10}}};

\node[text width=1cm] at (\hdistII  + 1.2, \vdistII * 1.5) {{\color{gray}\emph{100}}};
\node[text width=1cm] at (\hdistII  + 2, \vdistII - 0.3) {{\color{gray}\emph{10}}};
\node[text width=1cm] at (\hdistII * 2 + 2.2, \vdistII /2 ) {{\color{gray}\emph{200}}};

\end{tikzpicture}
\caption{A Super Bowl game: time $0$ plan under the myopic pricing mechanism.  Color coded paths $\{z_i\}_{i \in \driverSet}$ and rider values in bold indicate the movement of drivers in space and time, as well as the riders picked up by each driver. Numbers in italics below trips is the set of lowest origin-based market clearing prices.  \label{fig:super_bowl_myopic} 
}
\end{figure}

Consider the economy in Figure~\ref{fig:super_bowl_myopic}, modeling the end of a sports event, with $\horizon = 3$ time periods. Each time period is 10 minutes. Time $t = 0$ is 9:50pm, and 10 minutes before the game ends. There are three locations $A$, $B$ and $C$ with symmetric distances $\dist(A,A) = \dist(B,B) = \dist(C,C) = \dist(A,B) = \dist(B,A) = \dist(B,C) = \dist(C,B) = 1$ and $\dist(A,C) = \dist(C,A) = 2$.
Drivers $1$ and $2$ enter at location $C$ at time $0$, while driver $3$ enters at $B$ at time $0$, with exit times $\tl_i = \horizon$ for all $i \in \driverSet$. Riders' trips and values are as shown in the figure. The game ends at location $C$ at time $1$, where many riders with high values will request rides. 
Trips cost $10$ per time period, i.e. $\cost_{a,b,t} = 10\dist(a,b)$, $\forall (a,b,t) \in \trips$, and early exiting costs are $\exitCost_{\Delta} = 5 \Delta$.

\smallskip
\noindent{\textit{Suboptimal welfare.}} 
Under the myopic pricing mechanism, at time $0$, drivers $1$ and $2$
are dispatched to pick up riders $1$ and $2$, respectively, and driver
$3$ is dispatched to pick up rider $4$. At time $1$, driver $1$ picks
up rider $5$. Assuming optimal exiting (i.e. driver $2$ exits at time
$2$ at a cost of $\exitCost_{1} = 5$, while drivers $2$ and $3$ exit
at time $1$ and each incurs a cost of $\exitCost_{2} = 10$), the total
social welfare achieved %
is only $\val_1 + \val_2
+ \val_4 + \val_5 - 10 \times 4 - 5 - 10 - 10 = 25$. 
This also illustrates that the myopic pricing mechanism does not achieve any constant fraction of the optimal welfare.

\smallskip
\noindent{\textit{Unsmooth prices in space and time.}} 
The set of market clearing rates for $C$ at time $0$ is $\rho_{C,0} \in [0, \surplus_1] = [0, 10]$, thus the possible market clearing prices for the trip $(C,B,0)$ is $\price_{C,B,0} \in [10,~20]$. The price for the $(B,B,1)$ trip is $\price_{B,B,1} = \cost_{B,B,1} = 10$, since there is excess supply. 
At time $1$, since no driver is able to pick up the four riders at location $C$, the lowest market clearing rate is $\rho_{C,1} = \surplus_6 = 90$. The prices therefore must be at least $\price_{C,B,1} \geq \dist(C,B) \rho_{C,1} + \cost_{C,B,1} = 100$ and $\price_{C,A,1} \geq \dist(C,A)\rho_{C,1} + \cost_{C,A,1} = 200$. The set of lowest market clearing prices are shown in italics in Figure~\ref{fig:super_bowl_myopic}, below the edges corresponding to the trips. We can see that the price for the $(C,B)$ trip jumps from $10$ to $100$ within one period of time. Moreover, there exist large gaps between prices for trips originating from neighboring locations (compare e.g. $\price_{C,B,1} = 100$ and $\price_{B,B,1} = 10$).

\smallskip

\noindent{\textit{Incentivizing strategic behavior.}} 
Since $\price_{C,B,0} \leq 20$, the highest possible total utility to
driver $1$ under any myopic pricing mechanism would be $(20-10) +
(10-10) - 5 = 5$, and the utility to driver $2$ will not exceed
$(20-10) - 10 = 0$. 
Note that when all drivers follow the dispatches, the outcome fails to be envy-free since the two drivers who start at the same location and time have different total payoffs.
Now suppose driver $1$ deviates from the dispatch, and stays in $C$ until time $1$. The mechanism would then dispatch her to pick up rider $6$, and driver $1$ would be paid the new market clearing price of at least $\dist(C,B) \surplus_7 + \cost_{C,B,1} = 50$. This is a useful deviation, since by exiting at time $2$, her utility is now at least $- 10 + (50 - 10) - 5 = 25$. %
Also observe that by pricing in a myopic manner, strategic drivers are rewarded substantially higher earnings than  drivers who are straightforward and accept all dispatches.
\end{example}

\section{A Static CE Mechanism} \label{sec:static_mech}  

In this section, we formulate the welfare-optimal planning problem,
define competitive equilibrium (CE) prices, and prove a welfare
theorem, the core equivalence, and a lattice structure of drivers'
utilities among all CE plans.

\subsection{Plans} 

A {\em plan} describes the paths taken by all drivers until the end of the planning horizon, rider pick-ups, as well as payments for riders and drivers for each trip associated with these paths.

Formally, a plan is the 4-tuple $(x, \actpath, \rPayment, \dPayment)$, where: $x$ is the indicator of rider pick-ups, i.e. for all riders $j \in \riders$, $x_j = 1$ if rider $j$ is picked-up according to the plan, and $x_j = 0$ otherwise; $\actpath$ is a vector of action paths, where $\actpath_i \in \actPathSet_i$ is the dispatched action path taken by driver $i$; $\rPayment_j$ denotes the payment made by rider $j$, $\dPayment_{i,t}$ denotes the payment made to driver $i$ at time $t$, and let $\dPayment_i \triangleq \sum_{t = 0}^\horizon \dPayment_{i,t}$ denote the total payment to driver $i$. 
If $\actpath_i$ is consistent with $\path_{i,k}$, the $k\th$ feasible path of driver $i$, then the driver incurs a total cost of  $\pathCost_{i,k}$, and her utility is $\pi_i \triangleq \dPayment_i - \pathCost_{i,k}$.

A plan $(x, \actpath, \rPayment, \dPayment)$ is \emph{feasible} if for each rider $j \in \riders$, $x_j = \sum_{i \in \driverSet} \one{(\origin_j, \dest_j, \tr_j, j) \in \actpath_i} \in \{0,~1\}$, where $\one{\cdot}$ is the indicator function.
Unless stated otherwise, when we mention a plan in the rest of the paper, it is assumed to be feasible.
For the {\em budget balance} (BB) of a plan, we need: \vsq{-0.5em}
\begin{align}
	\sum_{j \in \riders} \rPayment_j \geq \sum_{i \in \driverSet} \dPayment_i, \label{equ:BB}
\end{align}
with strict budget balance if \eqref{equ:BB} holds with equality. 
A plan is \emph{individually rational for riders} if  \vsq{-0.5em}
\begin{align*}
	x_j \val_j \geq \rPayment_j, ~ \forall j \in \riders.
\end{align*}
A plan is \emph{individually rational for drivers} if $\pi_i \geq 0$ for all $i \in \driverSet$ s.t. $\driverEntrance_i = 0$, i.e. drivers that are not yet in the platform do not get negative utility from participating. 
A plan is {\em envy-free for riders} if no rider strictly prefers the outcome of another rider requesting the same trip, that is \vsq{-0.5em}
\begin{align}
	x_jv_j - \rPayment_j \geq x_{j'} \val_j - \rPayment_{j'} \text{ for all } j,~j' \in \riders \txtst \origin_j = \origin_{j'},~\dest_j = \dest_{j'}, \txtand \tr_j = \tr_{j'}. 
\end{align}
A plan is \emph{envy-free for drivers} if any pair of drivers with the same type have the same utility: \vsq{-0.5em}
\begin{align}
	\pi_i = \pi_{i'} \text{ for all } i,~i' \in \driverSet \txtst \te_i = \te_{i'},~  \re_i = \re_{i'}, \txtand \driverEntrance_i = \driverEntrance_{i'}.
\end{align}

Plans with a particular kind of anonymity structure can also be defined by associating prices with every possible trip.
\begin{definition}[Anonymous trip prices]  \label{defn:anon_trip_prices} A plan $(x, \actpath, \rPayment, \dPayment)$ uses \emph{anonymous trip prices} if there exist prices $\price = \{ \price_{a,b,t} \}_{ (a,b,t) \in \trips}$ such that for all $(a,b,t) \in \trips$, we have:
\begin{enumerate}[(i)]
	\setlength\itemsep{0.0em}
	\item all riders taking the same $(a,b,t)$ trip are charged the same payment $\price_{a,b,t}$, and there is no payment by riders who are not picked up, %
	and
	\item all drivers that are dispatched on a rider trip from $a$ to $b$ at time $t$ are paid the same amount $\price_{a,b,t}$ for the trip at time $t$, and there is no other payment to or from any driver.
\end{enumerate}
\end{definition}
Given dispatches $(x,\actpath)$ and anonymous trips prices $\price$, all payments are fully determined: the total payment to driver $i $ is $\dPayment_i = \sum_{j \in \riders} \one{(\origin_j,\dest_j,\tr_j,j) \in \actpath_i} \price_{\origin_j,\dest_j,\tr_j}$ and
the payment made by rider $j$ is $\rPayment_j = x_j \price_{\origin_j,
  \dest_j, \tr_j}$. For this reason, we will represent
plans with anonymous trip prices as $(x, \actpath,\price)$. 
By construction, plans with anonymous trip prices are strictly budget balanced.

\begin{definition}[Competitive equilibrium] \label{defn:CE}
A plan with anonymous trip prices $(x, \actpath,\price)$ forms a \emph{competitive equilibrium} (CE) if: 
\begin{enumerate}[(i)]
	\setlength\itemsep{0.0em}
	\item (rider best response) all riders $j \in \riders$ that can afford the ride are picked up, i.e.  $\val_j > \price_{\origin_j, \dest_j, \tr_j} \Rightarrow x_j = 1$, and all riders that  are picked up can afford the price: $ x_j = 1 \Rightarrow \val_j \geq \price_{\origin_j, \dest_j, \tr_j}$,
	\item (driver best response) $\forall i \in \driverSet$, $\pi_i  =  \max_{k = 0, \dots, |\pathSet_i|}  \left\lbrace \sum_{(a,b,t) \in  \path_{i,k} } \max \{ \price_{a,b,t} , 0 \} - \pathCost_{i,k} \right\rbrace$, i.e. each driver achieves the highest possible utility given prices and the set of feasible paths.%
\end{enumerate}
\end{definition}

Given any set of anonymous trip prices $\price$, let anonymous trip prices $\price^+$ be defined as $\price^+_{a,b,t} \triangleq \max\{\price_{a,b,t}, ~ 0\}$ for each $(a,b,t) \in \trips$. 

\begin{restatable}{lemma}{lemCEPlan} \label{lem:CE_plan_properties} Given any CE plan $(x, \actpath, \price)$, the plan with anonymous prices $(x, \actpath, \price^+)$ also forms a CE, and has the same driver and rider payments and utilities as those under $(x, \actpath, \price)$. 
\end{restatable}

The lemma implies that when studying the set of possible rider and driver payments and utilities among all CE outcomes, it is without loss to consider only anonymous trip prices that are non-negative. 
We leave the full proof of this lemma to Appendix~\ref{appx:proof_lem_CE_plan}. 
Intuitively, prices must be non-negative for any trip that is requested by any rider, thus changing prices from $\price$ to $\price^+$ does not affect the payments for any rider or driver, or the best response on the riders' side. The driver best response
property also continues to hold, since $\max \{ \price_{a,b,t} , 0 \} = \max \{ \price^+_{a,b,t} , 0 \}$ for all $(a,b,t) \in \trips$.

For a given mechanism that dispatches all available drivers at all times, and with knowledge of supply and demand and assuming drivers follow suggested dispatches, we can compute the intended outcome through the planning horizon. %
We call this the  ``time $0$ plan",  which consists of the assignments of 
riders, the action paths taken by drivers, and the payment schedule. 

\subsection{Optimal Plans and CE Prices}

The welfare-optimal planning problem can be formulated as an integer linear program (ILP) that determines rider pick-ups and driver paths, followed by an assignment of riders to drivers whose paths cover the rider trips. 
Let $x_j$ be the indicator that rider $j \in \riders$ is picked up, and $y_{i,k}$ be the indicator that driver $i$ takes
$\path_{i,k}$, her $k\th$ feasible path in $\pathSet_i$. We have: \vsq{-0.5em}
\begin{align}
	\max_{x,y} ~& \sum_{j \in \riders} x_j \val_j - \sum_{i \in \driverSet} \sum_{k = 0}^{|\pathSet_i|} y_{i,k}\pathCost_{i,k} \label{equ:ilp} \\
	\txtst & \sum_{j \in \riders} x_j  \one{(\origin_j, \dest_j, \tr_j) = (a,b,t)} \leq \sum_{i \in \driverSet} \sum_{k=0}^{|\pathSet_i|}  y_{i,k} \one{(a,b,t) \in \path_{i, k}}, & \forall (a,b,t) \in \trips \label{equ:ilp_trip_capacity} \\ 
		& \sum_{k = 0}^{|\pathSet_i|} y_{i,k} = 1, & \forall i \in \driverSet \label{equ:ilp_driver_constraint} \\
		& x_j \in \{0, 1\}, & \forall j \in \riders \\
		& y_{i,k} \in \{0, 1\}, & \forall i \in \driverSet,~k = 1, \dots, |\pathSet_i| 
\end{align}

Constraint \eqref{equ:ilp_driver_constraint} requires that each driver takes exactly one path (which includes the path $\path_{i,0}$ representing not entering/exiting immediately). 
The feasibility constraint \eqref{equ:ilp_trip_capacity} requires that for each trip $(a,b,t) \in \trips$, the number of riders who request this trip and are picked up is no greater than the total number of drivers whose  paths cover this trip.
Once the rider pick-ups $x$ and driver paths $y$ are computed, \eqref{equ:ilp_trip_capacity} guarantees that each rider with $x_j = 1$ can be assigned to a driver.%

Relaxing the integrality constraints on variables $x$ and $y$, we obtain the following linear program (LP) relaxation of the ILP: \vsq{-0.5em}
\begin{align}
	\max_{x,y} ~& \sum_{j \in \riders} x_j \val_j - \sum_{i \in \driverSet} \sum_{k = 0}^{|\pathSet_i|} y_{i,k}\pathCost_{i,k} \label{equ:lp} \\
	\txtst & \sum_{j \in \riders} x_j  \one{ (\origin_j, \dest_j, \tr_j) = (a,b,t) } \leq \sum_{i \in \driverSet} \sum_{k = 0}^{|\pathSet_i|}  y_{i,k} \one{(a,b,t) \in \path_{i, k}}, & \forall (a,b,t) \in \trips \label{equ:lp_trip_capacity_constraint} \\ 
		& \sum_{k = 0}^{|\pathSet_i|} y_{i,k} = 1, & \forall i \in \driverSet \label{equ:lp_driver_constraint} \\
		& x_j \leq 1, & \forall j \in \riders \label{equ:lp_rider_constraint} \\
		& x_j \geq 0,  & \forall j \in \riders \label{equ:lp_rider_nonneg}  \\
		& y_{i,k} \geq 0, & \forall i \in \driverSet,~k = 1, \dots, |\pathSet_i| \label{equ:lp_driver_nonneg}
\end{align}

We refer to \eqref{equ:lp} as the primal LP.
The constraint $y_{i,k} \leq 1$, that each path is taken by each driver at most once is guaranteed by imposing \eqref{equ:lp_driver_constraint} and \eqref{equ:lp_driver_nonneg}, and is omitted. 

\begin{restatable}[Integrality]{lemma}{thmLPIntegrality} \label{thm:lp_integrality}
There exists an integer optimal solution to the linear program~\eqref{equ:lp}. 
\end{restatable}

We leave the proof of this lemma to Appendix~\ref{appx:proof_thm_lp_integrality}, showing there a correspondence to a minimum cost flow (MCF) problem, where drivers flow through a network with vertices corresponding to (location, time) pairs, edges corresponding to trips, and with edge costs equal to driver's costs minus riders' values. The MCF has integral optimal solutions due to total-unimodularity, and this reduction to MCF can also be used to efficiently solve for the optimal plans.

Let $\price_{a,b,t}$, $\pi_i$ and $u_j$ denote the dual variables corresponding to the primal constraints \eqref{equ:lp_trip_capacity_constraint}, \eqref{equ:lp_driver_constraint} and \eqref{equ:lp_rider_constraint}, respectively. The dual LP of \eqref{equ:lp} is as follows: \vsq{-0.2em}
\begin{align}
	\min~ &\sum_{i \in \driverSet} \pi_i + \sum_{j \in \riders} u_j & \label{equ:dual} \\
	\txtst & \pi_{i} \geq \sum_{(a,b,t) \in \path_{i,k}} \price_{a,b,t} - \pathCost_{i,k} &  \forall k = 0, 1, \dots, |\pathSet_{i}|, ~ \forall i \in \driverSet  \label{equ:dual_cnst_driver} \\
	& u_j \geq \val_j - \price_{\origin_j, \dest_j, \tr_j},  & \forall j \in \riders \label{equ:dual_cnst_rider} \\
		 & \price_{a,b,t} \geq 0, &\forall (a,b,t) \in \trips & \label{equ:dual_price_nonneg}\\
	 & u_j
	  \geq 0, & \forall j \in \riders \label{equ:dual_util_nonneg}
\end{align}

\begin{restatable}[Welfare Theorem]{lemma}{thmOptPlans} \label{thm:optimal_plans}
A dispatching $(x, \actpath)$ is welfare-optimal if and only if there exists anonymous trip prices $\price$ s.t.~the plan $(x, \actpath, \price)$ forms a competitive equilibrium.
Such CE plans always exist and are efficient to compute. Moreover, these plans are strictly budget balanced, and are individually rational and envy-free for both riders and drivers.
\end{restatable}

See Appendix~\ref{appx:proof_thm_opt_plan} for the proof of this lemma. Briefly, the dual variables $\pi$ and $u$ can be interpreted as the utilities of drivers and riders, when the anonymous trip prices are given by $\price$. We then make use of Lemma~\ref{lem:CE_plan_properties}, and the standard observations about complementary slackness conditions and their connection with competitive equilibria~\citep{parkes2000iterative,bertsekas1990auction}. 
By integrality,  CE plans always exist, and can be efficiently computed by solving the primal and dual LPs of the MCF problem.

\smallskip

For two driver utility profiles $\pi = (\pi_1, \dots, \pi_\nd)$, $\pi' = (\pi'_1, \dots, \pi_\nd')$ that correspond to CE plans,  let the \emph{join} $\bar{\pi} = \pi \vee \pi'$ and the \emph{meet}
$\underline{\pi} = \pi \wedge \pi'$ be defined as $\bar{\pi}_i \triangleq \max\{\pi_i, \pi_i'\}$ and $\underline{\pi}_i \triangleq \min \{ \pi_i, \pi_i'\}$ for all $i\in \driverSet$.
The following lemma shows that drivers' utilities among all CE outcomes form a lattice, meaning that there exist CE plans where driver utilities are given by $\bar{\pi}$ or $\underline{\pi}$. 

The lemma also shows a connection between the top/bottom of the lattice and the welfare differences from losing/replicating a driver, which plays an important role in establishing the incentive properties of the STP mechanism.
Denote $\sw(\driverSet,~\riders)$ as the highest welfare achievable by drivers $\driverSet$ and riders $\riders$ (i.e. the optimal objective of \eqref{equ:lp}). 
For each driver $i \in \driverSet$, define the {\em social welfare gain from replicating driver} $i$, and the {\em social welfare loss from losing driver} $i$, as: \vsq{-0.2em}
\begin{align}
	\Phi_{\driverNode_i} & \triangleq \sw(\driverSet \cup \{i'\},~\riders) - \sw(\driverSet,~\riders), \label{equ:welfare_gain_driver}\\
	\Psi_{\driverNode_i}  & \triangleq \sw(\driverSet,~\riders) - \sw(\driverSet \backslash \{ i \},~\riders), \label{equ:welfare_loss_driver}
\end{align}
where driver $i'$ with $\theta_{i'} = \theta_i$ is a replica of driver $i$. A \emph{driver-optimal plan} has a driver utility profile at the top of the lattice, and a \emph{driver-pessimal plan} has a utility profile at the bottom of the lattice.

\begin{restatable}[Lattice Structure]{lemma}{lemOptPesPlans} \label{thm:opt_pes_plans} Drivers' utility profile $\pi$ among all CE outcomes form a lattice. Moreover, for each driver $i \in \driverSet$,  $\Phi_{\driverNode_i}$ and  $\Psi_{\driverNode_i}$ are equal to utility of driver $i$ in the driver-pessimal and driver-optimal CE plans, respectively.
\end{restatable}

We leave the proof of this lemma to Appendix~\ref{appx:proof_lemma_opt_pes_plans}. 
The lattice structure follows from the correspondences between driver utilities, the dual LP \eqref{equ:dual}, and the dual of the flow LP, and the fact that optimal dual solutions of MCF form a lattice.
Standard arguments on shortest paths in the residual graph~\citep{ahuja1993network}, and the connection between optimal dual solutions and subgradients (w.r.t. flow boundary conditions), then imply the correspondence between welfare gains/losses and driver pessimal/optimal utilities.

\medskip

A plan is in the \emph{core} if no coalition of riders and drivers can break out of this plan and make a plan among themselves, s.t. all drivers and riders in the coalition get at least their utilities from the original plan, and at least one of the drivers or riders is strictly better off. %

\begin{restatable}[Core Equivalence]{lemma}{lemCoreCE} \label{lem:core_equal_CE} All CE plans are in the core. Moreover, for any budget-balanced core outcome $(x, \actpath, \rPayment, \dPayment)$, there exists prices $\price$ such that the plan with anonymous prices $(x, \actpath, \price)$ forms a CE, and has the same driver and rider total utilities.
\end{restatable}

See Appendix~\ref{appx:proof_lem_core} for the proof of this lemma. Intuitively, any CE plan is in the core since for any $\driverSet' \subseteq \driverSet$ and $\riders' \subseteq \riders$, the highest achievable coalitional welfare $\sw(\driverSet',~\riders')$ is no greater than the sum of utilities of all driver and riders in this coalition under any CE plan. Given any core outcome, we can construct anonymous trip prices $\price$ that support the outcome in CE, and have the same driver and rider total payments: $\rPayment_j = x_j \price_{\origin_j, \dest_j, \tr_j} $ and $\dPayment_i = \sum_{j \in \riders}\one{ (\origin_j, \dest_j, \tr_j, j) \in \actpath_i  }  \price_{\origin_j, \dest_j, \tr_j}$.%

\smallskip

We revisit the Super Bowl example, and show that CE plans employ prices that are more smooth in space and time, in comparison to the the outcome under the myopic pricing mechanism.

\addtocounter{example}{-1} 
\begin{example}[Continued] \label{exmp:super_bowl_cont}

\newcommand{\nodeScaleII}{0.8}
\newcommand{\hdistII}{3.2}
\newcommand{\vdistII}{1.7}	

\begin{figure}[t!]
\centering
\vsq{-0.5em}
\begin{tikzpicture}[->,>=stealth',shorten >=1pt, auto, node distance=2cm,semithick][font = \small]
\tikzstyle{vertex} = [fill=white,draw=black,text=black,scale=0.9]

\draw[line width=0.25mm] 	(-1, \vdistII*2+0.5) -- (\hdistII * 3 + 1, \vdistII*2+0.5);

\node[text width=1.2cm] at (0, \vdistII*2+0.8) {9:50pm};
\node[text width=1.3cm] at (\hdistII, \vdistII*2+0.8) {10:00pm};
\node[text width=3cm] at (\hdistII + 0.2, \vdistII*2 + 1.2) {Super Bowl Ends};
\node[text width=1.3cm] at (\hdistII*2, \vdistII*2+0.8) {10:10pm};
\node[text width=1.3cm] at (\hdistII*3, \vdistII*2+0.8) {10:20pm};

\node[state]         (A0) [scale = \nodeScaleII] {$A,0$};
\node[state]         (B0) [above of=A0, node distance = \vdistII cm, scale=\nodeScaleII] {$B,0$};
\node[state]         (A0) [scale = \nodeScaleII] {$A,0$};
\node[state]         (C0) [above of=B0, node distance = \vdistII cm, scale=\nodeScaleII] {$C,0$};

\node[state]         (A1) [right of=A0, node distance = \hdistII cm, scale=\nodeScaleII] {$A,1$};
\node[state]         (B1) [right of=B0, node distance = \hdistII cm, scale=\nodeScaleII] {$B,1$};
\node[state]         (C1) [right of=C0, node distance = \hdistII cm, scale=\nodeScaleII] {$C,1$};

\node[state]         (A2) [right of=A1, node distance = \hdistII cm, scale=\nodeScaleII] {$A,2$}; 
\node[state]         (B2) [right of=B1, node distance = \hdistII cm, scale=\nodeScaleII] {$B,2$};
\node[state]         (C2) [right of=C1, node distance = \hdistII cm, scale=\nodeScaleII] {$C,2$};

\node[state]         (A3) [right of=A2, node distance = \hdistII cm, scale=\nodeScaleII] {$A,3$}; 
\node[state]         (B3) [right of=B2, node distance = \hdistII cm, scale=\nodeScaleII] {$B,3$};
\node[state]         (C3) [right of=C2, node distance = \hdistII cm, scale=\nodeScaleII] {$C,3$};

\node[text width=3cm] at (-0.3, 2 * \vdistII + 0.2) {{\color{\driverOneTextColor} Driver 1}};
\node[text width=3cm] at (-0.3, 2 * \vdistII - 0.2) {{\color{\driverTwoTextColor} Driver 2}};
\node[text width=3cm] at (-0.3, \vdistII) {{\color{\driverThreeTextColor} Driver 3}};

\path (C0) edge	node[pos=0.6, sloped, above] {$\val_1 \shorteq 20$ \hspace{0.7em} $\val_2 \shorteq 30$} (B1);
\path (B0) edge	node[pos=0.2, sloped, above] {{\color{\driverThreeRiderColor}  $\boldsymbol{\val_3 \shorteq 10}$}} (C1);
\path (B0) edge	node[pos=0.5, sloped, above] {$\val_4 \shorteq 20$} (A1);
\path (B1) edge	node[pos=0.5, sloped, above] {$\val_5 \shorteq 20$} (B2);
\path (C1) edge	node[pos=0.6, sloped, above] {{\color{\driverThreeRiderColor}  $\boldsymbol{\val_6 \shorteq 100}$}} (B2);

\draw[->]  (\hdistII + 0.35, \vdistII * 2)  to[out=-4, in= 130 ](\hdistII * 3 - 0.35, 0.1);

\node[text width=3cm] at (\hdistII * 2 + 2.2, \vdistII + 0.8) {{\color{\driverOneRiderColor}  $\boldsymbol{\val_7 \shorteq 100}$}};
\node[text width=3cm] at (\hdistII * 2 + 2.7, \vdistII + 0.5) {{\color{\driverTwoRiderColor}  $\boldsymbol{\val_8 \shorteq 90}$}};
\node[text width=3cm] at (\hdistII * 2 + 3.2, \vdistII + 0.2) {$\val_9 \shorteq 80$};

\draw[dashed,\driverOneColor, \driverthickness] (0.4, \vdistII *2 + 0.08) -- (\hdistII  - 0.35, \vdistII * 2  + 0.08);
\draw[dashdotted, \driverTwoColor, \driverthickness] (0.4, \vdistII *2 - 0.08) -- (\hdistII  - 0.35, \vdistII * 2  - 0.08);
\draw[dotted, \driverThreeColor, \driverthickness] (0.4, \vdistII + 0.07 ) -- (\hdistII - 0.27, \vdistII * 2 - 0.27);

\draw[dashed, \driverOneColor, \driverthickness]  (\hdistII + 0.4, \vdistII * 2 + 0.1)  to[out=-4, in= 130 ](\hdistII* 3 - 0.3, 0.2);
\draw[dashdotted, \driverTwoColor, \driverthickness]  (\hdistII + 0.4, \vdistII*2 - 0.1)  to[out=-4, in= 130 ](\hdistII* 3 - 0.34, - 0.08);
\draw[dotted, \driverThreeColor, \driverthickness] (\hdistII + 0.35, \vdistII * 2 - 0.3) -- (2 * \hdistII - 0.38, \vdistII + 0.07 );

\draw[dashed, \driverOneColor, \driverthickness](\hdistII * 3 + 1.1, 1) -- (\hdistII * 3 + 2.1, 1);
\draw[dashdotted, \driverTwoColor, \driverthickness](\hdistII*3 + 1.1, 0.5) -- (\hdistII*3 + 2.1, 0.5);
\draw[dotted, \driverThreeColor, \driverthickness](\hdistII * 3 + 1.1, 0) -- (\hdistII * 3 + 2.1, 0);

\node[text width=0.4cm] at (\hdistII * 3 + 2.3, 1) {{\color{\driverOneLgdColor} $z_1$}};
\node[text width=0.4cm] at (\hdistII * 3 + 2.3, 0.5) {{\color{\driverTwoLgdColor} $z_2$}};
\node[text width=0.4cm] at (\hdistII * 3 + 2.3, 0.0) {{\color{\driverThreeLgdColor} $z_3$}};

\node[text width=1cm] at (1.5, \vdistII + 0.15) {{\color{gray}\emph{0}}};
\node[text width=1cm] at (2.7, \vdistII + 0.15) {{\color{gray}\emph{55}}};
\node[text width=1cm] at (1.8, \vdistII /2 - 0.25) {{\color{gray}\emph{70}}};

\node[text width=1cm] at (\hdistII  + 1.3, \vdistII * 1.5-0.1) {{\color{gray}\emph{75}}};
\node[text width=1cm] at (\hdistII  + 1.9, \vdistII - 0.2) {{\color{gray}\emph{20}}};
\node[text width=1cm] at (\hdistII * 2 + 2, \vdistII /2 ) {{\color{gray}\emph{80}}};

\end{tikzpicture}
\vsq{-1.5em}
\caption{The Super Bowl example: the driver pessimal competitive
  equilibrium plan. Color coded paths $\{z_i\}_{i \in \driverSet}$ and rider values in bold indicate the movement of drivers in space and time, as well as the riders picked up by each driver. Numbers in italics %
  is the set of driver-pessimal CE prices. \label{fig:super_bowl_spatio_temporal} }\vsq{-0.5em}
\end{figure}

For the Super Bowl example introduced in Section~\ref{sec:preliminaries}, the driver-pessimal CE plan is as shown in Figure~\ref{fig:super_bowl_spatio_temporal}. All drivers stay at or re-position to location $C$, and pick up riders with high values at time $1$. %
The total rider value is $300$, and the total trip costs and exit costs incurred by the drivers are $80$ and $5$, respectively. This results in %
an optimal welfare of $215$, substantially higher than the welfare of 25 achieved under myopic pricing.

Trip prices are shown in italics, below the edges corresponding to the trips. For each feasible path of each driver, the total prices minus costs is $50$, which is the welfare gain from replicating the driver (an additional driver at $(C,0)$ or $(B,0)$ will be dispatched to $(C,1)$ and pick up rider $9$, improving rider values by $80$ and incurring a total cost of $30$). 
The outcome forms a CE, that there is no other path with a higher utility for any driver, and all riders are happy with their whether they are picked-up given the prices. We can also verify that there is no driver or rider envy, and that the outcome is in the core.

In contrast to the myopic pricing mechanism, where the price for the $(C,B)$ trip jumped from $10$ to $100$ between 9:50pm and 10pm, the CE prices started to increase more smoothly \emph{before} the end of the game in anticipation of higher future demand. Intuitively, sending a driver away from location $C$ right before the game ends is costly to the economy, and this is properly reflected  in the higher trip prices at 9:50pm. 
\end{example}

\subsection{The Static CE Mechanism}

Given the existence of welfare-optimal CE plans, we may consider a \emph{static CE mechanism}, which announces a CE plan at time $0$, and never again updates the plan even after driver deviations. Rather, each driver can choose to take any feasible path, but can only pick up riders that are dispatched to her, and is only paid for the subset of these rider trips that are completed.

\begin{definition}[Static CE mechanism] \label{defn:static_CE_mech} 
A static CE mechanism announces a CE plan $(x, \actpath, \price)$ at the beginning of the planning horizon. Each driver $i \in \driverSet$ then decides on the actual action path $\actpath'_i$ that she takes, and gets paid $ \hat{\dPayment}_i = \sum_{j \in \riders} \price_{\origin_j, \dest_j, \tr_j} \one{(\origin_j, \dest_j, \tr_j,j) \in \actpath_i,~(\origin_j, \dest_j, \tr_j,j) \in \actpath_i'}$. Each rider $j \in \riders$ pays $\rPayment_j$ only if she is picked up.
\end{definition}

A static CE mechanism can be defined for any set of CE prices.
Driver best response guarantees that no alternative path gives any driver a higher total utility, thus it is a dominant strategy for each driver to follow the dispatched action path $\actpath_i$.

\begin{theorem}\label{thm:static_CE_mech} A static CE mechanism implements an optimal CE plan in dominant strategy.
\end{theorem}

In addition, if all riders and drivers follow the plan, the outcome under a static CE mechanism is budget balanced,  and envy-free for both riders and drivers.
The CE property also ensures that every rider that is picked up is happy to take the trip at the offered price, and that no rider who is not picked up has positive utility for the trip.\footnote{Still, Example~\ref{exmp:toy_econ_1_continued} in Appendix~\ref{appx:rider_incentives} shows that truthful reporting of a rider's value need not be a dominant strategy (and this can be the case whichever CE prices are selected).}

The optimal static mechanism enjoys many good properties. By not updating the plan, however, a static CE mechanism is fragile to driver deviations, which could occur for many reasons: mistakes, unexpected contingencies, unexpected traffic, %
or unmodeled idiosyncratic preferences, etc.
The Super Bowl example demonstrates this lack of robustness: once a driver has deviated, the resulting outcome in the subsequent periods may no longer be reliable or welfare-optimal.

\addtocounter{example}{-1} 
\begin{example}[Continued] \label{exmp:super_bowl_cont_2}

Suppose that driver $3$ in the Super Bowl example did not follow the plan at time $0$ to pick-up rider $3$, %
but stayed in location $B$ until time $1$. 
Under a static CE mechanism with driver-pessimal CE plan (as shown in Figure~\ref{fig:super_bowl_spatio_temporal}), the effect of this deviation and not updating the plan %
is that driver $3$ is no longer able to pick up rider $6$ at time $1$, who strictly prefers to be picked up given the original price of $\price_{C,B,1} = 75$. 
Driver $2$, who was supposed to pick up rider $8$ %
is actually able to pick up rider $6$ instead of rider $8$, and this would lead to a higher welfare. Moreover, driver $3$ is now able to pick up rider $5$, however, she wouldn't be dispatched to do so. 
\myqed
\end{example}

One may think of a naive fix for this robustness issue of the static CE mechanisms, simply repeating the computation of the plan at all times. The following example shows that the mechanism that recomputes a driver pessimal plan at all times fails to be incentive compatible.
Similarly, we show that the mechanism that repeatedly recomputes a {\em driver-optimal} plan is not envy-free and also have incentive issues (see Example~\ref{exmp:always_compute_driver_opt} in Appendix~\ref{appx:examples}).

\begin{example}\label{exmp:toy_econ_3} 

\newcommand{\nodeScaleI}{0.8}
\newcommand{\hdistI}{4}
\newcommand{\vdistI}{1.7}	

\begin{figure}[t!]
\centering
\begin{tikzpicture}[->,>=stealth',shorten >=1pt,auto,node distance=2cm,semithick][font = \small]
\tikzstyle{vertex}=[fill=white,draw=black,text=black,scale=0.9]

\node[state]         (A0) [scale = \nodeScaleI] {$A,0$};
\node[state]         (B0) [above of=A0, node distance = \vdistI cm, scale=\nodeScaleI] {$B,0$};
\node[state]         (A1) [right of=A0, node distance = \hdistI cm, scale=\nodeScaleI] {$A,1$};
\node[state]         (B1) [right of=B0, node distance = \hdistI cm, scale=\nodeScaleI] {$B,1$};
\node[state]         (A2) [right of=A1, node distance = \hdistI cm, scale=\nodeScaleI] {$A,2$}; 
\node[state]         (B2) [right of=B1, node distance = \hdistI cm, scale=\nodeScaleI] {$B,2$};

\path (B1) edge	node[pos=0.5, sloped, above] {$\boldsymbol{\val_1 \shorteq 8}$} (B2);

\path (A1) edge	node[pos=0.5, sloped, above] {$\boldsymbol{\val_2 \shorteq 6}$, $\val_3 \shorteq 5$, $\val_4 \shorteq 4$} (A2);

\node[text width=3cm] at (-0.3, \vdistI) {Driver 1};
\node[text width=3cm] at (-0.3, 0) {Driver 2};

\draw[dashed] (0.4, \vdistI - 0.1) -- (\hdistI - 0.38, \vdistI  - 0.1);
\draw[dashed] (\hdistI + 0.4, \vdistI - 0.1) -- (\hdistI * 2 - 0.38, \vdistI  - 0.1);

\draw[dashdotted] (0.4,  - 0.1) -- (\hdistI - 0.38, - 0.1);
\draw[dashdotted] (\hdistI + 0.4, - 0.1) -- (\hdistI * 2 - 0.38, - 0.1);

\draw[dashed](\hdistI * 2 + 1.5, 0.5) -- (\hdistI * 2 + 2.5, 0.5);
\draw[dashdotted](\hdistI * 2 + 1.5, 0) -- (\hdistI * 2 + 2.5, 0);

\node[text width=0.4cm] at (\hdistI * 2 + 2.7, 0.5) {$z_1$};
\node[text width=0.4cm] at (\hdistI * 2 + 2.7, 0) {$z_2$};

\node[text width=1cm] at (\hdistI * 1 + 2.35,  \vdistI - 0.3) {\emph{5}};
\node[text width=1cm] at (\hdistI * 1 + 2.35, -0.3) {\emph{5}};

\end{tikzpicture}
\caption{The economy in Example~\ref{exmp:toy_econ_3} and the driver pessimal CE plan computed at time $0$. %
\label{fig:toy_econ_3}}
\end{figure}
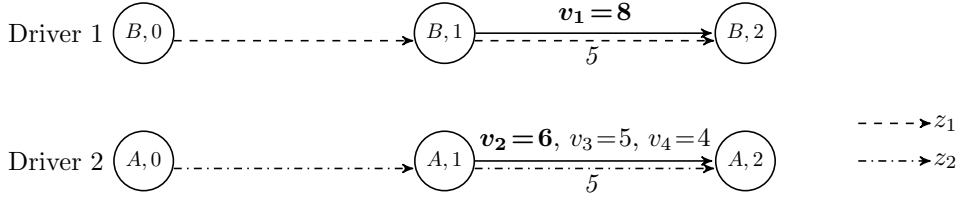

Consider the economy as shown in Figure~\ref{fig:toy_econ_3}, where there are two locations $\loc = \{A, B\}$ with distances $\dist(A,A) = \dist(A,B) = \dist(B,A) = \dist(B,B) = 1$. Assume for simplicity that all trip costs and opportunity costs are zero: $\cost_{a,b,t}= 0$ for all $(a,b,t) \in \trips$, and $\exitCost_\Delta = 0$ for all $\Delta = 0, 1, \dots, \horizon$. In the driver-pessimal plan computed at time $0$ as shown in the figure, the anonymous trip prices are $\price_{B,B,1} = \price_{A,A,1} = 5$.
Assume that both drivers $1$ and $2$ follow the plan at time $0$, and reach $(B,1)$ and $(A,1)$ respectively. If the mechanism re-computes the plan at time $1$, the new driver-pessimal plan would set a new price of $0$ for the trip $(B,B,1)$--- the updated lowest market-clearing price for the trip. Therefore, if driver $1$ follows the mechanism at all times, her total payment and utility would actually be $0$.
Now consider the scenario where driver $2$ follows the mechanism at
time $0$, but driver $1$ deviates and relocates to $A$, so that both
drivers are at location $A$ at time $1$. At time $1$, when the
mechanism recomputes a driver-pessimal plan, both drivers would take
the trip $(A,A,1)$ and pick up riders $2$ and $3$ respectively. The
updated price for the trip $(A,A,1)$ would be $4$, and
this is a useful deviation for driver $1$.  \qed
\end{example}

The challenge is to achieve robustness, but at the same time handle the new strategic considerations that can occur as a result of drivers being able to trigger re-planning through
deviations.

\section{The Spatial-Temporal Pricing Mechanism} \label{sec:st_pricing}

In this section, we introduce the Spatio-Temporal Pricing mechanism, and prove our main result, that it is a subgame-perfect equilibrium for drivers to always follow the mechanism's dispatch.

\subsection{A Dynamic Mechanism}

We first formally define a dynamic ridesharing mechanism, that can use the history of actions to update the plan forward from the current state.

Let $\state_t = (\state_{1,t}, \state_{2,t}, \dots, \state_{\nd, t})$ denote the \emph{state} of the ridesharing platform at time $t$, where each $\state_{i,t}$ describes the state of driver $i \in \driverSet$. If driver $i$ has entered the platform and is available at time $t$ at location $a \in \loc$, denote $\state_{i,t} = (1, a,t)$.
Otherwise, if driver $i$ is  {\em en route}, finishing the trip from $a$ to $b$ that she started at time $\tprime < t$ s.t. $\tprime + \dist(a,b) > t$, denote $\state_{i,t} = (a,b,\tprime)$ if she is relocating with no rider, or $\state_{i,t} = (a, b, \tprime, j)$ if she is taking a rider $j$ from $a$ to $b$ at time $\tprime$. %
For drivers that had already exited or decided not to enter, denote $\state_{i,t} = \phi$. 
For drivers with $\te_i \geq t$, i.e. who enters or is able to enter now or in the future, $\state_{i,t} = ( \driverEntrance_i, \re_i, \te_i)$. The initial state of the platform is $\state_0 = ((\driverEntrance_1, \re_1, \te_1), \dots, (\driverEntrance_{\nd}, \re_{\nd}, \te_{\nd}))$.

At each time $t$, each driver $i$ takes an \emph{action} $\action_{i,t}$. An available driver $i$ with $\state_{i,t} = (1,a,t)$ or $\state_{i,t} = (0,a,t)$ may (enter and then) relocate to any location $b$ within reach by the end of the planning horizon (i.e. $b \in \loc$ s.t. $t + \dist(a,b) \leq \horizon$), which we denote $\action_{i,t}= (a, b, t)$. She may pick up a rider $j \in \riders$ with $\tr_j = t$ and $\origin_j = a$, in which case we write $\action_{i,t} = (a, \dest_j, t, j)$. She may also decide to exit (if $\driverEntrance_i = 1$) or not enter (if $\driverEntrance_i = 0$), for both cases we denote $\action_{a,t} = \phi$. For a driver $i$ that is {\em en route} at time $t$, (i.e. $\state_{i,t} = (a,b,\tprime)$ or $\state_{i,t} = (a,b,\tprime,j)$ for some $\tprime$ s.t. $\tprime + \dist(a,b) > t$), $\action_{i,t} = \state_{i,t}$--- the only available action is to finish the current trip. For driver $i$ with $\te_i > t$, denote $\action_{i,t} = \state_{i,t} = (\driverEntrance_i, \te_i, \re_i)$. A driver with $\state_{i,t} = \phi$ takes no more actions: $\action_{i,t} = \state_{i,t} = \phi$. 

The action $\action_{i,t}$ taken by driver $i$ at time $t$ determines her state $\state_{i,t+1}$ at time $t+1$: 
\begin{itemize}
	\setlength\itemsep{0.0em}
	\item (will complete trips at $t+1$) if $\action_{i,t} = (a,b,\tprime)$ or $\action_{i,t} = (a,b,\tprime, j)$ s.t. $\tprime + \dist(a,b) = t+1$, then $\state_{i,t+1} = (1, b, t+1)$,  i.e. becoming available at time $t+1$ at the
destination of their trips,\footnote{Here we assume that a driver
that declines the mechanism's dispatch and decide to relocate from $a$ to $b$ also does so in time $\dist(a,b)$. We can also handle
drivers who move more slowly when deviating, just as long as the
mechanism knows when and where the driver will become available
again.}
	\item  (still \emph{en route}) if $\action_{i,t} = (a,b,\tprime)$ or $\action_{i,t} = (a,b,\tprime, j)$  s.t. $\tprime + \dist(a,b) > t+1$, then $\state_{i, t+1} = \action_{i,t}$, 
	\item (not yet entered) for $i \in \driverSet$ s.t. $\action_{i,t} = (\driverEntrance_i, \te_i, \re_i)$, we have $\state_{i, t+1} = (\driverEntrance_i, \te_i, \re_i)$,
	\item (already exited / never entered) if $\action_{i,t} = \phi$, then $\state_{i,t+1} = \phi$. 
\end{itemize}

Let $ \action_t = ( \action_{1,t}, \action_{2,t}, \dots,
\action_{\nd,t})$ be the \emph{action profile} of all drivers at time $t$, and let
\emph{history} $\history_t \triangleq (\state_0, \action_0, \state_{1},
\action_1, \dots, \state_{t-1}, \action_{t-1}, \state_{t})$, with
$\history_0 = (\state_0)$.  Finally, let $\driverSet_t(\history_{t})=
\{i \in \driverSet ~|~ \state_{i,t} = (1,a,t) \text{ or } \state_{i,t} = (0,a,t) \text{ for some } a \in \loc \}$ be the set of drivers available at time $t$.

\begin{definition}[Dynamic ridesharing mechanism]   \label{defn:ride_sharing_mech} 
A dynamic ridesharing mechanism is defined by its {\em  dispatch rule} $\action^\ast$, {\em driver payment rule} $\dPayment^\ast$ and {\em rider payment rule} $\rPayment^\ast$.
At each time $t$, given history $\history_{t}$ and rider information $\riders$, the mechanism:
\begin{itemize}%
	\setlength\itemsep{0.0em}
	\item uses its dispatch rule $\action^\ast$ to determine for each of a subset of available drivers, a {\em dispatch action} $\action_{i,t}^\ast (\history_{t})$ to either pick up a rider, or to relocate, or to exit/not enter.
	\item uses its driver payment rule $\dPayment^\ast$ to determine, for each dispatched driver,  a payment $\dPayment_{i,t}^\ast(\history_{t})$  in the event the driver takes the action ($\dPayment_{i,t}^\ast(\history_{t}) = 0$ for available drivers that are not dispatched).
	\item dispatches each {\em en route} driver to keep driving (i.e. $\action_{i,t}^\ast(\history_t) = \state_{i,t}$), and does not make any payment to driver $i$ in this period: $\dPayment_{i,t}^\ast (\history_t) = 0$.
	\item determines for drivers entering in the future ($i \in \driverSet$ s.t. $\te_i > t$), and drivers who had already exited ($i \in \driverSet$ s.t. $\state_{i,t} = \phi$), $\action_{i,t}^\ast(\history_t) = \state_{i,t}$ and $\dPayment_{i,t}^\ast (\history_t) = 0$. 
	\item uses its rider payment rule $\rPayment^\ast$ to determine, for each rider who receives a dispatch at time $t$, %
the payment $\rPayment_j^\ast (\history_t)$ in the event that 
the rider is picked up.
\end{itemize}

Each driver then decides on which action $\action_{i,t} \in \actionSet_{i,t}(\history_{t})$ to take, where $\actionSet_{i,t}(\history_{t})$ is the set of actions available to agent $i$ at time $t$ given history $\history_t$. For an available driver at $(a,t)$ with dispatched action $\action_{i,t}^\ast(\history_{t})$, $\actionSet_{i,t}(\history_{t}) = \{\action_{i,t}^\ast(\history_{t})\} \cup \{(a,b,t)~|~b \in \loc \txtst t + \dist(a,b) \leq \horizon \} \cup \{\phi\}$, i.e. the driver can either take the dispatched action, or to relocate to any location, or to exit or not enter; if an available driver at $(a,t)$ is not dispatched, $\action_{i,t}^\ast(\history_{t})$, $\actionSet_{i,t}(\history_{t}) = \{(a,b,t)~ |~b \in \loc \txtst t + \dist(a,b) \leq \horizon \} \cup \{\phi\}$; for an {\em en route} driver, or a driver that enters in the future, or a driver that has already exited, $\actionSet_{i,t}(\history_{t}) = \{\state_{i,t}\}$. After observing the action profile $\alpha_t$, the mechanism pays each dispatched driver $\actualPayment_{i,t} (\action_{i,t}, \history_t)= \dPayment_{i,t}^\ast(\history_{t}) \one{\action_{i,t} = \action_{i,t}^\ast}$, and charges each rider $j \in \riders$ with $\tr_j = t$ the amount $\actualRiderPayment_j (\action_t) = \rPayment^\ast_j (\history_t) \sum_{i \in \driverSet_t} \one{\action_{i,t}= (\origin_j, \dest_j, t, j)}$.
\end{definition}

A mechanism is \emph{feasible} if (i) at any time it is possible for each available driver to take the trip that is dispatched to her, i.e. $\forall t$, $\forall \history_t$, $\forall i \in \driverSet_t$, if $\state_{i,t} = (1, a,t)$ or $\state_{i,t} = (0, a, t)$ for some $a \in \loc$, $\action_{i,t}^\ast (\history_t) \in \set{(a,b,t)}{b \in \loc,~t + \dist(a,b) \leq \horizon} \cup \set{(\origin_j,   \dest_j, \tr_j, j)}{ j \in \riders,~\tr_j = t,~\origin_j = a}$, (ii) no rider is picked-up more than once, i.e. $\forall t$, $\forall \history_t$, $\forall j \in \riders$ s.t. $\tr_j = t$, $\sum_{i \in \driverSet_t} \one{\action_{i,t}^\ast(\history_t) = (\origin_j, \dest_j, \tr_j, j)} \leq 1$, and (iii) unavailable drivers are not dispatched.
We can see from Definition~\ref{defn:ride_sharing_mech} that there is no payment to or from unavailable or undispatched drivers, or a dispatched driver $i$ who deviated from $\action^\ast_{i,t}(\history_t)$ at time $t$, or riders who are not picked up.\sloppy

Let $\historySet_t$ be the set of all possible {\em histories} up to time $t$.
A \emph{strategy} $\strategy_{i}$ of driver $i$ defines for
all times $t \in \activeTimeSet$ and all histories $\history_t \in
 \historySet_t$, the action she takes
$\action_{i,t} = \strategy_{i}(\history_t) \in
\actionSet_{i,t}(\history_t)$. For a mechanism that always dispatches
all available drivers, $\strategy_i^\ast$
denotes the \emph{straightforward strategy} of always following the
mechanism's dispatches at all times.
Let $\strategy = (\strategy_1, \dots, \strategy_{\nd})$ be the
{\em strategy profile}, with $\strategy_{-i} =
(\strategy_1, \dots, \strategy_{i-1}, \strategy_{i+1},
\dots, \strategy_{\nd})$. The strategy profile $\strategy$,
together with the initial state $\state_0$ and the rules
of a mechanism, determine all actions and
payments of all drivers through the planning horizon.
Let $\strategy_i |_{\history_t}$, $\strategy |_{\history_t}$ and
$\strategy_{-i} |_{\history_t}$ denote the strategy profile from time $t$ and history $\history_t$ onward for driver $i$, all drivers, and all drivers but $i$, respectively.

For each rider $j \in \riders$, let $\actualAssignment_j(\strategy)\in \{0,1\}$ be the indicator that rider $j$ is 
picked-up given strategy $\strategy$, and 
let $\actualRiderPayment_j(\strategy) = \actualAssignment_j(\strategy) \rPayment^\ast_j(\history_{\tr_j})$ be her actual payment.
For each driver $i \in \driverSet$, $\actualPayment_i(\strategy) \triangleq \sum_{t=0}^{\horizon-1} \actualPayment_{i,t}(\strategy_i(\history_t), \history_{t})$ denotes the total actual payments made to driver $i$, where drivers follow $\strategy$ and the history $\history_t$ is induced by the initial state and strategy $\strategy$. 
Let $\actualDriverUtil_{i,t}(\strategy_i(\history_t), \history_{t})$ be the actual utility driver $i$ gets at time $t$ given history $\history_t$ and strategy $\strategy_i$. We know that if $\strategy_i(\history_t) = (a,b,t)$ or $\strategy_i(\history_t) = (a,b,t,j)$, then $\actualDriverUtil_{i,t}(\strategy_i(\history_t), \history_{t}) = \actualPayment_{i,t}(\strategy_i(\history_t), \history_{t}) - \cost_{a,b,t}$; if $\strategy_i(\history_t) = \phi$ and $\state_{i,t} = (1,a,t)$ for some $a \in \loc$, then $\actualDriverUtil_{i,t}(\strategy_i(\history_t), \history_{t}) = \actualPayment_{i,t}(\strategy_i(\history_t), \history_{t}) - \exitCost_{\horizon - t}$. For every other scenario, $\actualDriverUtil_{i,t}(\strategy_i(\history_t), \history_{t}) = \actualPayment_{i,t}(\strategy_i(\history_t), \history_{t})$. 
Denote $\actualDriverUtil_i(\strategy) \triangleq \sum_{t=0}^{\horizon-1} \actualDriverUtil_{i,t}(\strategy_i(\history_t), \history_{t})$ as driver $i$'s total utility. 

Fixing driver and rider types, a ridesharing mechanism induces a finite horizon extensive form game. At each time point $t$, each driver decides on an action $\alpha_{i,t} = \strategy_i(\history_t)\in \actionSet_{i,t}(\history_t)$ to take based on strategy $\strategy_i$ and the history $\history_t$, and receives utility $\actualDriverUtil_{i,t}(\alpha_{i,t}, \history_t)$.
The total utility $\actualDriverUtil_i(\strategy)$ to each driver is
determined by the rules of the mechanism.

We define the following properties.
\begin{definition}[Budget balance] A ridesharing mechanism is {\em  budget balanced} if for any set of riders and drivers, and any strategy profile $\strategy$ taken by the drivers, we have 
\begin{align}
	\sum_{j \in \riders} \actualRiderPayment_j (\strategy)  \geq \sum_{i \in \driverSet} \actualPayment_i(\strategy).
\end{align}
\end{definition}

\begin{definition}[Subgame-perfect incentive compatibility] A ridesharing mechanism that always dispatches all available drivers is \emph{subgame-perfect incentive compatible} (SPIC) for drivers if given any set of riders and drivers, following the mechanism's dispatches at all times forms a subgame-perfect equilibrium (SPE) among the drivers, meaning for all $t \in \activeTimeSet$, for any history $\history_{t} \in \historySet_t$, 
\begin{align}
	\sum_{\tprime = t}^{\horizon - 1} \actualDriverUtil_{i, \tprime}( \strategy_i^\ast |_{\history_t}, \strategy_{-i}^\ast |_{\history_t})
	 	\geq 
	 \sum_{\tprime = t}^{\horizon - 1} \actualDriverUtil_{i, \tprime}(\strategy_i |_{\history_t}, \strategy_{-i}^\ast |_{\history_t}),
	 ~\forall \strategy_i,~\forall i \in \driverSet.  \label{equ:ICSPE}
\end{align}
\end{definition}

A ridesharing mechanism is {\em dominant strategy incentive compatible (DSIC)} if for any driver, following the mechanism's dispatches at all time points that the driver is dispatched 
maximizes her total payment, regardless of the actions taken by the
rest of the drivers.

\begin{definition}[Individual rationality (IR)] A ridesharing mechanism that always dispatches all available drivers is \emph{individually rational in SPE for drivers} if for any set of riders and drivers, (i) the mechanism is SPIC for drivers, and (ii) assuming $\strategy^\ast$, drivers that have not yet entered do not get negative utility from participating, i.e. 
\begin{align*}
	\actualDriverUtil_i(\strategy^\ast) \geq  0 \txtfor i \in \driverSet \txtst \driverEntrance_i = 0.
\end{align*}
A ridesharing mechanism is {\em individually rational for riders} if for any set of riders and drivers, and any strategy profile $\strategy$ taken by the drivers,
\begin{align}
	\actualAssignment_j(\strategy) \val_j \geq \actualRiderPayment_j(\strategy), ~\forall i \in \riders.
\end{align}
\end{definition}

\begin{definition}[Envy-freeness in SPE] A ridesharing mechanism that always dispatches all available drivers is {\em envy-free in SPE for drivers} if for any set of riders and drivers, (i) the mechanism is SPIC for drivers, and (ii) for any time $t \in \activeTimeSet$, for all history $\history_{t} \in \historySet_t$, all drivers with the same state at time $t$ are paid the same total amount in the subsequent periods, assuming all drivers follow the mechanism's dispatches:
\begin{align}
	\sum_{\tprime = t}^{\horizon - 1} \actualDriverUtil_{i, \tprime}( \strategy^\ast|_{\history_t}) = \sum_{\tprime = t}^{\horizon - 1} \actualDriverUtil_{i', \tprime}( \strategy^\ast|_{\history_t}), ~ \forall i, i' \in \driverSet \txtst \state_{i,t} = \state_{i', t}. 
\end{align}
A ridesharing mechanism is {\em envy-free in SPE for riders} if (i) the mechanism is SPIC for drivers, and (ii) for all $j \in \riders$, for all possible $\history_{\tr_j} \in \historySet_{\tr_j}$, and all $j' \in \riders$ s.t. $(\origin_j, \dest_j, \tr_j) = (\origin_{j'}, \dest_{j'}, \tr_{j'})$
\begin{align}
	\actualAssignment_j(\strategy^\ast)\val_j - \actualRiderPayment_j(\strategy^\ast) \geq \actualAssignment_{j'}(\strategy^\ast)\val_j - \actualRiderPayment_{j'}(\strategy^\ast).
\end{align}
\end{definition}

\begin{definition}[Core-selecting] A ridesharing mechanism that always dispatches all available drivers is \emph{core-selecting} if for any set of riders and drivers, (i) the mechanism is SPIC, and (ii) for any time $t \in \activeTimeSet$ and any history $\history_t \in \historySet_t$ onward, the outcome under the straightforward strategy $\strategy^\ast$ is in the core. 
\end{definition}

Fix a mechanism with dispatch rule $\alpha^\ast$ and payment rules $\rPayment^\ast,~\dPayment^\ast$, where all available drivers are always dispatched. Recall that the outcome under the straightforward strategy $\strategy^\ast$ over the entire planning horizon can be computed at time $0$, and is called the {\em time $0$ plan} of the mechanism. %
If some driver deviated at time $t-1$ for some $t > 0$, the downward outcomes given the dispatching and payment rules, assuming all drivers follow $\sigma^\ast|_{\history_t}$, can be thought of as an updated time $t$ plan.

For any time $t \in \timeSet$, given any state $\state_t$ of the platform, let $\econ \supt (\state_t)$ represent the
\emph{time-shifted economy} starting at state $\state_t$, with
planning horizon $\horizon \supt = \horizon - t$, the same set of locations $\loc$
and distances $\dist$, and the remaining riders $\riders \supt = \{(\origin_j,
\dest_j, \tr_j - t, \val_j) ~|~ j \in \riders, \tr_j \geq t \}$.
For drivers, we have
$\driverSet \supt (\state_t) = \{ \theta_i\supt ~|~ i \in \driverSet \}$, with types $\theta_i\supt = (\driverEntrance_i \supt , \re_i\supt, \te_i\supt, \tl_i\supt)$ determined as follows: %
\begin{enumerate}[$\bullet$]
	\setlength\itemsep{0.0em}
	\item for available drivers $i \in \driverSet$ s.t.
$\state_{i, t} = (1, a, t)$ or $\state_{i,t} = (0, a, t)$ for some $a \in \loc$, let $\theta_i\supt = (\driverEntrance_i\supt, \re_i\supt, \te_i\supt, \tl_i\supt) = (1, a, 0, \tl_i-t)$ or $(1, a, 0, \tl_i-t)$, respectively, 
	\item for {\em en route} drivers
$i \in \driverSet$ s.t. $\state_{i, t} = (a, b, t')$ or $(a, b, t',
j)$ where $t'+ \dist(a,b) > t$, let $\theta_i\supt = (\driverEntrance_i \supt, \re_i\supt, \te_i\supt, \tl_i\supt) = (1, b,  t' + \dist(a,b) - t, \tl_i-t)$, 
	\item for each driver $i \in \driverSet$ with $\te_i >t$,  let $\theta_i\supt = (\driverEntrance_i \supt, \re_i\supt, \te_i\supt, \tl_i\supt) = (\driverEntrance_i, \re_i, \te_i-t, \tl_i-t)$, and
	\item exclude drivers that have already exited or chose not to enter.
\end{enumerate}

\begin{definition}[Temporal consistency] \label{defn:temp_invariance} A ridesharing mechanism is \emph{temporally consistent} if after deviation at time $t-1$, the updated plan is identical to that determined for economy $\econ \supt (\state_t)$.
\end{definition}

Upon deviation(s) at time $t-1$, a temporally consistent mechanism updates its plan from time $t$ onward as if $t$ is the beginning of the planning horizon, thus does not make use of time-extended contracts. %
A temporally inconsistent mechanism is able to trivially align incentives, by firing any driver who has deviated, or by threatening to ``shut down" after any deviation, for example. 

\subsection{The Spatio-Temporal Pricing Mechanism}

We define the STP mechanism by providing a method to plan or re-plan, this implicitly defining the dispatch and payment rules. For each $a \in \loc$ and $t \in \timeSet$, denote the welfare gain 
from an additional driver at $(a,t)$ that is already in the platform as,
\begin{align}
	\Phi_{a,t} \triangleq \sw(\driverSet \cup \{(1, a, t, \horizon) \}) - \sw(\driverSet), \label{equ:welfare_gain}
\end{align}
where $(1, a, t, \horizon)$ represents the type of this driver that stays until the end of the planning horizon.

\begin{definition}[Spatio-Temporal pricing mechanism] \label{defn:stp}  
The {\em spatio-temporal pricing (STP) mechanism} is a dynamic ridesharing mechanism that always dispatches all available drivers. Given economy $\econ\0$ at the beginning of the planning horizon, or economy $\econ\supt(\state_t)$ immediately after a deviation by one or more drivers, 
the mechanism completes the following planning step:
\begin{enumerate}[$\bullet$]
	\item {\em Dispatch rule:} To determine the dispatches ($\alpha^\ast$), compute an optimal solution $(x, y)$ to the LP~\eqref{equ:lp}, and dispatch each driver $i$ to take the path $\path_{i,k}$ for $k$ s.t. $y_{i,k} = 1$, and pick up riders with $x_j = 1$, %
	\item {\em Payment rules:} To determine driver and rider payments ($\dPayment^\ast$ and $\rPayment^\ast$), for each $(a,b,t) \in \trips$, set anonymous trip prices to be 
	\begin{align}
		\price_{a,b,t} = \Phi_{a,t} - \Phi_{b, t+\dist(a,b)} + \cost_{a,b,t}	\label{equ:stp_prices}
\end{align}		
\begin{enumerate}[-]
	\item For each rider $j \in \riders$, $\rPayment_j^\ast =  \price_{\origin_j, \dest_j, \tr_j} \sum_{i \in \driverSet} \one{\action^\ast_{i,\tr_j} = (\origin_j, \dest_j, \tr_j, j)} $, 
	\item For each driver $i \in \driverSet$, $\dPayment_{i,t}^\ast = \sum_{j \in \riders, \tr_j = t} \price_{\origin_j,\dest_j,t} \one{\action^\ast_{i,\tr_j} = (\origin_j, \dest_j, t, j)} $. %
\end{enumerate}
\end{enumerate}
\end{definition}

We now state the main result of the present paper.

\begin{restatable}{theorem}{spe} \label{thm:SPE}  
The spatio-temporal pricing mechanism is temporally consistent and subgame-perfect incentive compatible. It is also individually rational for riders and strictly budget balanced for any action profile taken by the drivers. From any history onward, the equilibrium outcome is welfare optimal, core-selecting, envy-free, and individually rational for drivers.
\end{restatable}

The proof of Theorem~\ref{thm:SPE} is provided in Appendix~\ref{appx:proof_thm_spe}.
We first show %
that the total utility of each driver under the STP mechanism is $\pi_i  = \Phi_{\driverNode_i}$, the welfare gain from replicating driver $i$. 
Setting $u_j = \max\{ \val_j - \price_{\origin_j, \dest_j, \tr_j} ,~ 0 \}$ for all $j \in \riders$, we show that $(\price, \pi, u)$ forms an optimal solution to the dual LP \eqref{equ:dual} by observing (i) $(\Phi, u)$ forms an optimal solution to the dual of the corresponding MCF problem (the proof of Lemma~\ref{thm:opt_pes_plans}), and (ii) a correspondence between the optimal solutions of the dual LP~\eqref{equ:dual} and the optimal solutions of the dual of the MCF (Lemma~\ref{lem:dual_payment_correspondence} in Appendix~\ref{appx:proof_thm_opt_plan}). 
This implies that the plan determined by the STP mechanism starting from any history onward forms a CE, and as a result is individually rational, budget balanced, envy-free, and resides in the core.

For incentive alignment, the single-deviation principle~\citep{osborne1994course} %
implies that we only need show that a single deviation from the  dispatching is not useful.
For any driver available at some  location $a$ and time $t$, her total utility from time $t$ onward, if all drivers follow the dispatches, is equal the welfare gain (at the time when the plan is computed) from adding an extra driver at location $a$ and time $t$. 
We establish that this welfare gain is weakly higher than the welfare gain for the economy (at time $t+1$) from replicating this driver at any location and time that the driver can deviate and relocate to.

For this, we use the $M^\natural$ concavity (and more specifically, the local exchange properties) of optimal objectives of MCF problems~\citep{murota2003discrete}. In particular, to maximize welfare, there is stronger substitution among drivers at the same location and time, than among drivers at different locations or times. This shows that declining the mechanism's dispatch to stay/relocate is not useful. We also show that none of (i) exiting earlier than dispatched, (ii) not entering/exiting when asked to, and (iii) entering when dispatched not to, is a useful deviation.

\newcommand{\nodeScaleII}{0.8}
\newcommand{\hdistII}{3.5}
\newcommand{\vdistII}{2}	

\addtocounter{example}{-2} 

\begin{example}[Continued] 

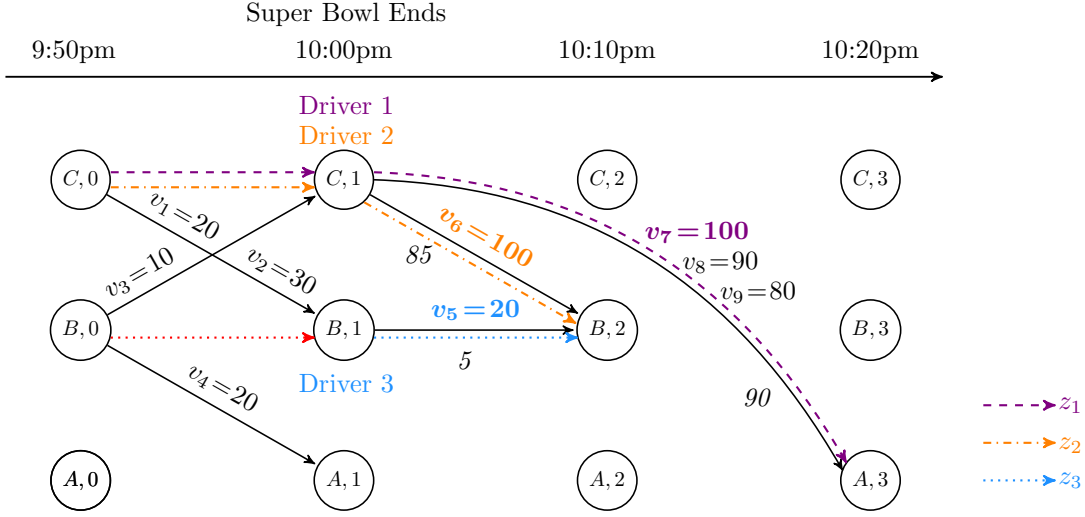
\begin{figure}[t!]
\centering
\begin{tikzpicture}[->,>=stealth',shorten >=1pt, auto, node distance=2cm,semithick][font = \small]
\tikzstyle{vertex} = [fill=white,draw=black,text=black,scale=0.9]

\draw[line width=0.25mm] 	(-1, \vdistII*2 + 1.37) -- (\hdistII*3 + 1, \vdistII*2 + 1.37);

\node[text width=1.3cm] at (0, \vdistII*2+1.7) {9:50pm};
\node[text width=1.3cm] at (\hdistII, \vdistII*2+1.7) {10:00pm};
\node[text width=3cm] at (\hdistII + 0.2, 6.2) {Super Bowl Ends};
\node[text width=1.3cm] at (\hdistII*2, \vdistII*2+1.7) {10:10pm};
\node[text width=1.3cm] at (\hdistII*3, \vdistII*2+1.7) {10:20pm};

\node[state]         (A0) [scale = \nodeScaleII] {$A,0$};
\node[state]         (B0) [above of=A0, node distance = \vdistII cm, scale=\nodeScaleII] {$B,0$};
\node[state]         (A0) [scale = \nodeScaleII] {$A,0$};
\node[state]         (C0) [above of=B0, node distance = \vdistII cm, scale=\nodeScaleII] {$C,0$};

\node[state]         (A1) [right of=A0, node distance = \hdistII cm, scale=\nodeScaleII] {$A,1$};
\node[state]         (B1) [right of=B0, node distance = \hdistII cm, scale=\nodeScaleII] {$B,1$};
\node[state]         (C1) [right of=C0, node distance = \hdistII cm, scale=\nodeScaleII] {$C,1$};

\node[state]         (A2) [right of=A1, node distance = \hdistII cm, scale=\nodeScaleII] {$A,2$}; 
\node[state]         (B2) [right of=B1, node distance = \hdistII cm, scale=\nodeScaleII] {$B,2$};
\node[state]         (C2) [right of=C1, node distance = \hdistII cm, scale=\nodeScaleII] {$C,2$};

\node[state]         (A3) [right of=A2, node distance = \hdistII cm, scale=\nodeScaleII] {$A,3$}; 
\node[state]         (B3) [right of=B2, node distance = \hdistII cm, scale=\nodeScaleII] {$B,3$};
\node[state]         (C3) [right of=C2, node distance = \hdistII cm, scale=\nodeScaleII] {$C,3$};

\node[text width=3cm] at (\hdistII + 0.9, 2 * \vdistII + 1) {{\color{\driverOneRiderColor} Driver 1}};
\node[text width=3cm] at (\hdistII+ 0.9, 2 * \vdistII + 0.6) {{\color{\driverTwoRiderColor} Driver 2}};
\node[text width=3cm] at (\hdistII+ 0.9, \vdistII-0.7) {{\color{\driverThreeRiderColor} Driver 3}};

\path (C0) edge	node[pos=0.55, sloped, above] {$\val_1 \shorteq 20$ \hspace{0.7em} $\val_2 \shorteq 30$} (B1);
\path (B0) edge	node[pos=0.2, sloped, above] { ${\val_3 \shorteq 10}$} (C1);
\path (B0) edge	node[pos=0.5, sloped, above] {$\val_4 \shorteq 20$} (A1);
\path (B1) edge	node[pos=0.5, sloped, above] {{\color{\driverThreeRiderColor} $\boldsymbol{\val_5 \shorteq 20}$}} (B2);
\path (C1) edge	node[pos=0.5, sloped, above] {{\color{\driverTwoRiderColor} $\boldsymbol{\val_6 \shorteq 100}$}} (B2);

\draw[->]  (3.9, 4)  to[out=-2, in= 120 ](10.15, 0.1);

\node[text width=3cm] at (9, 3.3) {{\color{\driverOneRiderColor} $\boldsymbol{\val_7 \shorteq 100}$}};
\node[text width=3cm] at (9.5, 2.9) {$\val_8 \shorteq 90$};
\node[text width=3cm] at (10, 2.5) {$\val_9 \shorteq 80$};

\draw[dashed, \driverOneColor, \driverthickness] (0.4, \vdistII *2 + 0.1) -- (\hdistII  - 0.35, \vdistII * 2  + 0.1);
\draw[dashdotted, \driverTwoColor, \driverthickness] (0.4, \vdistII *2 - 0.1) -- (\hdistII  - 0.35, \vdistII * 2  - 0.1);
\draw[dotted, red, \driverthickness] (0.4, \vdistII - 0.1) -- (\hdistII  - 0.35, \vdistII - 0.1);

\draw[dashed, \driverOneColor, \driverthickness]  (3.9, \vdistII * 2 + 0.1)  to[out=-2, in= 120 ](10.2, 0.2);
\draw[dashdotted, \driverTwoColor, \driverthickness] (\hdistII + 0.27, \vdistII * 2 - 0.3) -- (2 * \hdistII - 0.38, \vdistII + 0.07 );
\draw[dotted, \driverThreeColor, \driverthickness] (\hdistII +0.4, \vdistII - 0.1) -- (\hdistII *2  - 0.35, \vdistII - 0.1);

\draw[dashed, \driverOneColor, \driverthickness](\hdistII * 3 + 1.5, 1) -- (\hdistII * 3 + 2.5, 1);
\draw[dashdotted, \driverTwoColor, \driverthickness](\hdistII*3 + 1.5, 0.5) -- (\hdistII*3 + 2.5, 0.5);
\draw[dotted, \driverThreeColor, \driverthickness](\hdistII * 3 + 1.5, 0) -- (\hdistII * 3 + 2.5, 0);

\node[text width=0.4cm] at (\hdistII * 3 + 2.7, 1) {{\color{\driverOneLgdColor} $z_1$}};
\node[text width=0.4cm] at (\hdistII * 3 + 2.7, 0.6) {{ \color{\driverTwoLgdColor} $z_2$}};
\node[text width=0.4cm] at (\hdistII * 3 + 2.7, 0.0) {{\color{\driverThreeLgdColor} $z_3$}};

\node[text width=1cm] at (\hdistII * 2 + 2.3, 1.1) {\emph{90}};
\node[text width=1cm] at (\hdistII  + 2, 1.6) {\emph{5}};
\node[text width=1cm] at (\hdistII  + 1.3, 3) {\emph{85}};

\end{tikzpicture}

\caption{The Super Bowl example: replanning under the STP mechanism at time $1$ after driver $3$ deviated from the original plan and stayed in location $B$ until time $1$. \label{fig:super_bowl_spatio_temporal_replan}}
\end{figure}
Whereas the static CE mechanism fails to be welfare-optimal or envy-free for riders after driver $3$ deviates from the dispatch and stays in location $B$ until time $1$, the plan recomputed under the STP  mechanism at time $1$ is as illustrated in Figure~\ref{fig:super_bowl_spatio_temporal_replan}.
Driver $3$ is re-dispatched to pick up rider $5$ and then exit from $(B,2)$. Instead of picking up rider $8$ whose value is $90$, driver $2$ now picks up rider $6$ who was initially assigned to driver $3$.
If there existed an additional driver at $(C,1)$, the driver will be dispatched to pick up rider $8$, and contribute to a welfare gain of $\val_8 - \cost_{C,A,1} = 70$. An additional driver at $(A,3)$ has no effect on welfare, thus the price for the $(C,A,1)$ trip is updated to $\price_{C,A,1} =  \Phi_{C,1} - \Phi_{A,3} + \cost_{C,A,1} = 70 - 0 + 20 = 90$, and the utility of each driver from $(C,1)$ onward is equal to $\Phi_{C,1} = 70$. Similarly, we have $\price_{C,B,1} = \Phi_{C,1} - \Phi_{B,2} + \cost_{C,B,1} = 70 - (-5) + 10 = 85$ and $\price_{B,B,1} = \Phi_{B,1} - \Phi_{B,2} + \cost_{B,B,1} = -10 - (-5) + 10 = 5$. The utility of driver $3$ from $(B,1)$ onward is $\price_{B,B,1} - \cost_{B,B,1} - \exitCost_{1} = 5 - 10 - 5 = -10$. The outcome remains envy-free for riders, and welfare optimal from time $1$ onward. 
\end{example}
\addtocounter{example}{1}

\medskip

Under the STP mechanism, replanning can be triggered by the deviation of any driver, thus the utility of a driver is affected by the actions of others, and the mechanism is not DSIC. We also show in the following theorem that no mechanism can implement the desired properties in a dominant-strategy equilibrium. 
\begin{restatable}{theorem}{impossOfDSE} \label{thm:impossibility_of_DSE} Following the
  mechanism's dispatch at all times does not form a dominant strategy
  equilibrium under any dynamic ridesharing mechanism that is, from
  any history onward, (i) welfare-optimal, (ii) IR for riders, (iii) budget balanced, and (iv) envy-free for riders and drivers.
\end{restatable}

\begin{proof} We show that for the economy in Example~\ref{exmp:toy_econ_3}, as shown in Figure~\ref{fig:toy_econ_3}, under any mechanism that satisfies conditions (i)-(iv), following the dispatches at all times cannot be a DSE. We start by analyzing what must be the outcome at time $0$ under such a mechanism. At time $0$, optimal welfare is achieved by dispatching one of the two drivers to go to $(B,1)$ so that at time $1$ she can pick up rider $1$, and the other driver to go to $(A,1)$ to pick up rider $2$. Assume w.l.o.g.~that at time $0$, driver $1$ is dispatched to stay in $B$ and driver $2$ is dispatched to stay in $A$.

  Now consider the scenario where driver $2$ deviated, and took the
  trip $(A,B,0)$ at time $0$ instead. If driver $1$ followed the mechanism's dispatch at time $0$, both drivers are at $(B,1)$ at time $1$, and
  the welfare-optimal outcome is to pick up rider $1$. Individual
  rationality requires that the highest amount of payment we can
  collect from rider $1$ is $8$. Budget balance and envy-freeness of
  drivers then imply that drivers $1$ and $2$ are each paid at most
  $4$ at time $1$.
If driver $2$ is going to deviate at time $0$ and relocate to $B$, driver $1$ may
deviate from the mechanism's dispatch and relocate to $(A,1)$ instead.
In this case, at time $1$ it is welfare optimal for driver $1$ to pick
up rider $2$. Her payment for the trip $(A,A,1)$ is at least $5$, for
otherwise rider $3$ envies the outcome of rider $2$. This is better
than following the mechanism and get utility at most 4.
\end{proof}

A natural variation on the STP mechanism is the driver-optimal analog, which always computes a driver-optimal CE plan at the beginning of the planning horizon, or upon the deviation of any driver. Under this mechanism, a driver's continuation payoff from some location and time onward is equal to her ``marginal product'', i.e. the welfare-loss in the economy from losing a driver at this location and time. Despite the fact that it reflects the payments of a VCG mechanism, the following example shows that the driver-optimal mechanism is not incentive compatible. This is because a driver's marginal product can increase over time, as the set of trips that can be completed by the rest of the drivers becomes smaller. In these scenarios, the driver may deviate from the mechanism's dispatch, trigger a recomputation of the downstream plan, and get paid the updated, higher marginal product in the subsequent periods.

\begin{example} %
\label{exmp:toy_econ_2}
Consider the economy illustrated in Figure~\ref{fig:toy_econ_2}, with three locations, three time periods and symmetric distances $\dist(A,A) = \dist(B,B) = \dist(C,C) = \dist(B,C) = 1$, $\dist(A,B) = \dist(A,C) = 2$. All trip costs and early exit costs are zero. 
Two drivers enter the platform at time $0$ at location $B$, and three riders have types: \vspace{-1em}
\begin{multicols}{2}
\begin{enumerate}[$\bullet$]
	\setlength\itemsep{0.0em}
	\item Rider 1: $\origin_1 = C$, $\dest_1 = C$, $\tr_1 = 1$, $\val_1 = 1$, 
	\item Rider 2: $\origin_2 = C$, $\dest_2 = C$, $\tr_2 = 2$, $\val_2 = 5$,
	\item Rider 3: $\origin_3 = A$, $\dest_3 = A$, $\tr_3 = 1$, $\val_3 = 1$.
\end{enumerate}
\end{multicols}

\vspace{-0.5em}

\newcommand{\nodeScaleIII}{0.8}
\newcommand{\hdistIII}{3}
\newcommand{\vdistIII}{1.2}	

\begin{figure}[t!]
\centering
\begin{tikzpicture}[->,>=stealth',shorten >=1pt, auto, node distance=2cm,semithick][font = \small]
\tikzstyle{vertex} = [fill=white,draw=black,text=black,scale=0.9]

\node[state]         (A0) [scale = \nodeScaleIII] {$A,0$};
\node[state]         (B0) [above of=A0, node distance = \vdistIII cm, scale=\nodeScaleIII] {$B,0$};
\node[state]         (A0) [scale = \nodeScaleIII] {$A,0$};
\node[state]         (C0) [above of=B0, node distance = \vdistIII cm, scale=\nodeScaleIII] {$C,0$};

\node[state]         (A1) [right of=A0, node distance = \hdistIII cm, scale=\nodeScaleIII] {$A,1$};
\node[state]         (B1) [right of=B0, node distance = \hdistIII cm, scale=\nodeScaleIII] {$B,1$};
\node[state]         (C1) [right of=C0, node distance = \hdistIII cm, scale=\nodeScaleIII] {$C,1$};

\node[state]         (A2) [right of=A1, node distance = \hdistIII cm, scale=\nodeScaleIII] {$A,2$}; 
\node[state]         (B2) [right of=B1, node distance = \hdistIII cm, scale=\nodeScaleIII] {$B,2$};
\node[state]         (C2) [right of=C1, node distance = \hdistIII cm, scale=\nodeScaleIII] {$C,2$};

\node[state]         (A3) [right of=A2, node distance = \hdistIII cm, scale=\nodeScaleIII] {$A,3$}; 
\node[state]         (B3) [right of=B2, node distance = \hdistIII cm, scale=\nodeScaleIII] {$B,3$};
\node[state]         (C3) [right of=C2, node distance = \hdistIII cm, scale=\nodeScaleIII] {$C,3$};

\path (C1) edge	node[pos=0.5, sloped, above] {$\boldsymbol{\val_1=1}$} (C2);
\path (C2) edge	node[pos=0.5, sloped, above] {$\boldsymbol{\val_2=5}$} (C3);
\path (A2) edge	node[pos=0.5, sloped, above] {$\boldsymbol{\val_3=1}$} (A3);

\node[text width=3cm] at (-0.3, 1.4) {Driver 1};
\node[text width=3cm] at (-0.3, 1.0) {Driver 2};
    
\draw[dashed] (0.35, \vdistIII * 1 + 0.2) -- (2.7, \vdistIII * 2 - 0.2);
\draw[dashed] (\hdistIII + 0.35, \vdistIII * 2 - 0.2) -- (\hdistIII * 2 - 0.3, \vdistIII * 2 - 0.2);
\draw[dashed] (\hdistIII * 2 + 0.35, \vdistIII * 2 - 0.2) -- (\hdistIII * 3 - 0.3, \vdistIII * 2 - 0.2);

\draw[dashdotted] (0.35, 1.3) -- (5.6, -0.1);
\draw[dashdotted] (\hdistIII * 2 + .35, -0.2) -- (\hdistIII * 3 - 0.3, - 0.2);

\draw[dashed] (10.5, 0.5) -- (11.5, 0.5);
\draw[dashdotted] (10.5, 0.0) -- (11.5, 0.0);

\node[text width=0.4cm] at (11.8, 0.5) {$z_1$};
\node[text width=0.4cm] at (11.8, 0.0) {$z_2$};

\end{tikzpicture}
\caption{Illustration of the economy in Example~\ref{exmp:toy_econ_2} with three locations $A$, $B$, $C$, three time periods, two drivers starting at $(B, 0)$ and three riders with values $1$, $5$ and $1$, respectively. Under a welfare optimal plan, driver $1$ picks up riders $1$ and $2$ and driver 2 picks up rider $3$. 
\label{fig:toy_econ_2}}
\end{figure}

In a welfare-optimal dispatching as shown in Figure~\ref{fig:toy_econ_2}, driver $1$ is dispatched to take the path $z_1 = ((B,C,0),~(C,C,1),~(C,C,2))$ and to pick up riders $1$ and $2$. Driver 2 takes the path $z_2 = ((B,A,0),~(A,A,2))$ and picks up rider $3$. 
Every driver-optimal CE plan sets anonymous trip prices $\price_{A,A,2} = 1$ and $\price_{C,C,1} + \price_{C,C,2} = 1$, so that the total utility of each driver is equal $1$, the welfare loss if one driver is removed from the economy.

Assume that driver $2$ follows the mechanism's dispatch and starts to drive toward location $A$ at time $0$, we show a useful deviation of driver $1$ by rejecting the dispatched relocation to $C$ and staying in location $B$. %
At time $1$, when the platform updates the downstream plan, %
driver $2$ is already \emph{en route} to $A$ thus the only rider she is able to pick up in the future is rider $3$.
Driver $1$ would be asked to relocate to $C$ and then pick up rider $2$. The price $\price_{C,C,2}$ in the updated driver-optimal CE plan would be $5$, the welfare loss if the economy at time $1$ loses driver $1$ at time $1$. This is higher than driver $1$'s payment from following the dispatches at all times.
\myqed
\end{example}

This kind of useful deviation does not exist under the STP mechanism, since under the driver-pessimal CE plan each driver is paid the additional welfare the economy gains if we replicate this driver. It is always more useful to have an extra driver earlier, thus the ``replica welfare gain'' is monotonically non-increasing over time. %

A variation on the driver-optimal mechanism where drivers' payments are shifted in time is equivalent to the {\em dynamic VCG mechanism}~\citep{bergemann2010dynamic,cavallo2009efficient}. The dynamic VCG mechanism does not align incentives for drivers either, because the existence of a driver at some point of time may exert negative externality on the economy, in which case the payment to the driver would be negative for that time period. The driver would have incentives to decline the dispatch and avoid such payment. See discussions and examples in Appendix~\ref{appx:dynamic_vcg}.
The VCG mechanisms would address information asymmetry, but these examples highlight that the challenge we face is rather one of aligning incentives %
in the absence of time-extended contracts.

\section{Simulation Results} \label{sec:illustrations}

In this section, we compare, through numerical simulations, the performance of the STP mechanism against the myopic pricing mechanism, for three stylized scenarios: the end of a sporting event, the morning rush hour, and trips to and from the airport with unbalanced demand.

In addition to social welfare, we consider the \emph{time-efficiency} of drivers, which is defined as the proportion of time where the drivers have a rider in the car, divided by the
total time drivers spend on the platform.
We also consider the \emph{regret} to drivers for following the
straightforward strategy in a non incentive-aligned mechanism: the
highest additional amount a driver can gain by strategizing in
comparison to following a mechanism's dispatch, assuming that the rest of the drivers all follow the mechanism's dispatches at all times.
The analysis suggests that the STP mechanism achieves substantially
higher social welfare, as well as time-efficiency for drivers,
whereas, under the myopic pricing mechanism, prices are highly unstable, and drivers incur a high regret.

We define the myopic pricing mechanism to use the lowest market
clearing prices (which market clearing prices are chosen is unimportant for the results). 
In addition, since this mechanism need not always dispatch all available drivers, we model any available driver who is not dispatched as randomly choosing a location that is within reach, %
and relocates there if this trip cost is no greater than the cost for exiting immediately (in which case she exits).

\subsection{Scenario One: The End of a Sporting Event}

We first consider the scenario in Figure~\ref{fig:ex1_end_of_event}, modeling the end of a sport event.
There are three locations $\loc = \{A, B, C\}$ with unit distances $\dist(a,b) = 1$ for all $a, b\in \loc$, and two time periods. %
Trip costs are $3$ per period, and exiting early costs $1$ per period: $\cost_{a,b,t} = 3\dist(a,b)$, and $\exitCost_{\Delta} = \Delta$. 
In each economy, at time $0$, there are $15$ and $10$ drivers that are already in the platform becoming available at locations $C$ and $B$.\footnote{With the assumption that all drivers are already in the platform, the results do not reflect the disadvantage of myopic for not making optimal entrance decisions.}
$20$ riders request trip $(C,B,0)$, and $10$ riders request trips $(B,C,0)$ and $(B,A,0)$ respectively. When the event ends, there are $N_{C,B,1}$ riders hoping to take a ride from $(C,1)$ to $(B,2)$. The values of all riders are independently drawn from the exponential distribution with mean $10$.

\newcommand{\nodeScaleSI}{0.8}
\newcommand{\hdistSI}{4}
\newcommand{\vdistSI}{1.5}	

\begin{figure}[t!]
\centering
\begin{tikzpicture}[->, >=stealth',shorten >=1pt, auto, node distance=2cm,semithick][font = \small]

\tikzstyle{vertex} = [fill=white,draw=black,text=black,scale=0.9]

\node[state]         (A0) [scale = \nodeScaleSI] {$A,0$};
\node[state]         (B0) [above of=A0, node distance = \vdistSI cm, scale=\nodeScaleSI] {$B,0$};
\node[state]         (A0) [scale = \nodeScaleSI] {$A,0$};
\node[state]         (C0) [above of=B0, node distance = \vdistSI cm, scale=\nodeScaleSI] {$C,0$};

\node[state]         (A1) [right of=A0, node distance = \hdistSI cm, scale=\nodeScaleSI] {$A,1$};
\node[state]         (B1) [right of=B0, node distance = \hdistSI cm, scale=\nodeScaleSI] {$B,1$};
\node[state]         (C1) [right of=C0, node distance = \hdistSI cm, scale=\nodeScaleSI] {$C,1$};

\node[state]         (A2) [right of=A1, node distance = \hdistSI cm, scale=\nodeScaleSI] {$A,2$}; 
\node[state]         (B2) [right of=B1, node distance = \hdistSI cm, scale=\nodeScaleSI] {$B,2$};
\node[state]         (C2) [right of=C1, node distance = \hdistSI cm, scale=\nodeScaleSI] {$C,2$};

\node[text width=3cm] at (-0.7, 2 * \vdistSI ) {15 Drivers};
\node[text width=3cm] at (-0.7, \vdistSI) {10 Drivers};

\path (C0) edge node[pos=0.25, sloped, above] {$20$ riders} (B1);
\path (B0) edge	node[pos=0.27, sloped, below] {$10$ riders} (C1);
\path (B0) edge	node[pos=0.5, sloped, below] {$10$ riders} (A1);
\path (C1) edge	node[pos=0.5, sloped, above] {$N_{C,B,1}$ riders} (B2);

\end{tikzpicture}
\caption{An example to illustrate the end of an event. \label{fig:ex1_end_of_event}}
\end{figure}

As we vary the number of riders $N_{C,B,1}$ requesting the trip $(C,B,1)$ from $0$ to $100$, we randomly generate $1,000$ economies, and compare the average welfare and driver's time efficiency %
in Figure~\ref{fig:ex1_welfare_efficiency}. 
Figure~\ref{fig:ex1_event_welfare} shows that the STP mechanism achieves a substantially higher social welfare than the myopic pricing mechanism, especially when there 
are a large number of drivers taking the trip $(C,B,1)$. 
Figure~\ref{fig:ex1_event_efficiency} shows that the STP mechanism becomes less time-efficient as the number of $(C,B,1)$ riders increases, as more of the $15$ drivers starting at $(C,0)$ stay in the same location until time $1$. The high time-efficiency achieved by myopic is because of the fact that undispatched drivers decided to exit immediately. The effective use of driver's time under myopic (total amount of time driver spend driving riders, divided by the total time a driver is willing to work) is in fact around 60 percent. 

\renewcommand{\imageHeight}{1.7}

\begin{figure}[t!]
\centering
\begin{subfigure}[t]{0.4\textwidth}
	\centering
    \includegraphics[height= \imageHeight in]{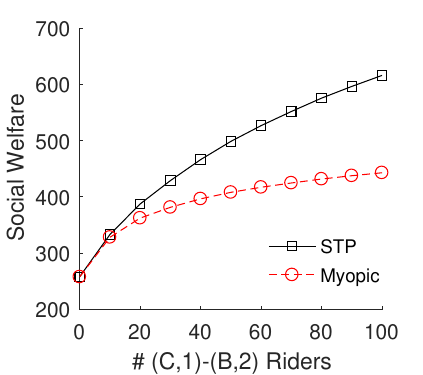}
    \caption{Social welfare.\label{fig:ex1_event_welfare}}
\end{subfigure}%
\begin{subfigure}[t]{0.4\textwidth}
	\centering
    \includegraphics[height=\imageHeight in]{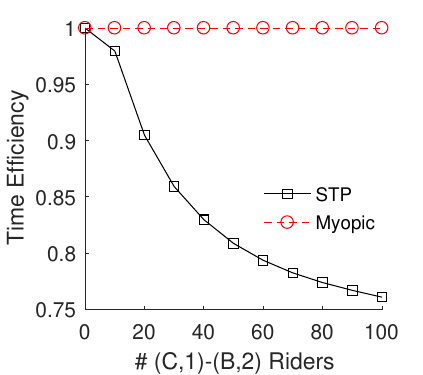}
    \caption{Driver time efficiency. \label{fig:ex1_event_efficiency}}
\end{subfigure}%
\caption{Comparison of social welfare and driver time efficiency for the end of an event.  \label{fig:ex1_welfare_efficiency} }    
\end{figure}

The average number of drivers taking each of the four trips of interest under the two mechanisms are shown in Figure~\ref{fig:ex1_drivers}. As $N_{C,B,1}$ increases, the STP mechanism dispatches more drivers to $(C,1)$ to pick-up higher-valued riders leaving $C$, while sending less drivers on trips $(C,B,0)$ and $(B,A,0)$. The myopic pricing mechanism, being oblivious to future demand, sends all drivers starting at $(C,0)$ to location $B$, and an average of only $5$ drivers to $(C,1)$ from $(B,1)$.

\begin{figure}[t!]
\centering
\begin{subfigure}[t]{0.4\textwidth}
	\centering
    \includegraphics[height= \imageHeight in]{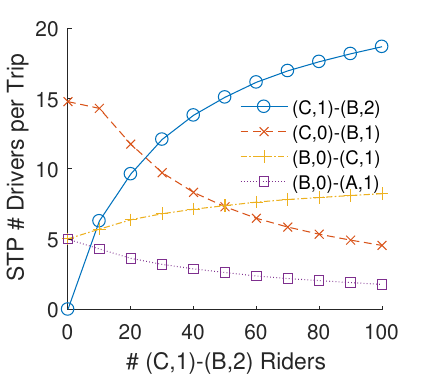}
    \caption{The STP mechanism. \label{fig:ex1_event_stp_drivers}}
\end{subfigure}%
\begin{subfigure}[t]{0.4\textwidth}
	\centering
    \includegraphics[height=\imageHeight in]{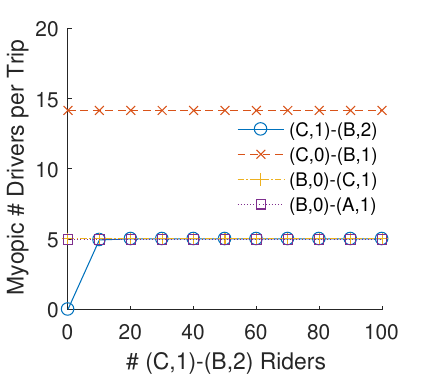}
    \caption{The myopic pricing mechanism.  \label{fig:ex1_event_myopic_drivers}}
\end{subfigure}%
\caption{Comparison of the number of drivers per trip for the end of an event. \label{fig:ex1_drivers} }    
\end{figure}

The average trip prices are 
plotted in Figure~\ref{fig:ex1_prices}. First of all, prices under STP are temporally ``smooth"--- trips leaving $C$ at times $0$ and $1$ have very similar prices. On average, $\price_{C,B,1}$ is higher than $\price_{C,B,0}$, since drivers taking the $(C,0)$-$(B,1)$ trip can exit at time $1$ and incur a smaller total cost. 
The price for the trip $(B,A,0)$ is the highest, so that a driver dispatched to $A$ does not envy those dispatched to take the trip $(B,C,0)$ and then $(C,B,1)$. 
In contrast, the price for the $(C,B)$ trip drastically increases from time $0$ to time $1$ under the myopic pricing mechanism, since there are few available driver at location $C$ at time $1$. 
The ``surge" for the trip $(C,B,1)$ is substantially higher under myopic pricing, implying that the platform is providing even less price reliability for the riders.

\begin{figure}[t!]
\centering
\begin{subfigure}[t]{0.4\textwidth}
	\centering
    \includegraphics[height= \imageHeight in]{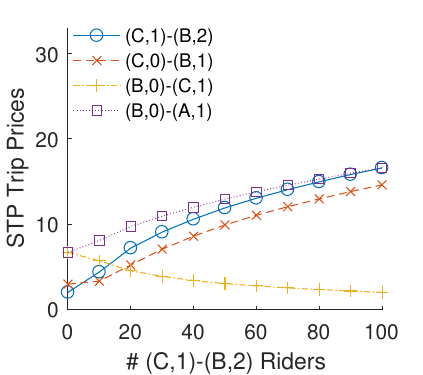}
    \caption{The STP mechanism.\label{fig:ex1_event_stp_prices}}
\end{subfigure}%
\begin{subfigure}[t]{0.4\textwidth}
	\centering
    \includegraphics[height=\imageHeight in]{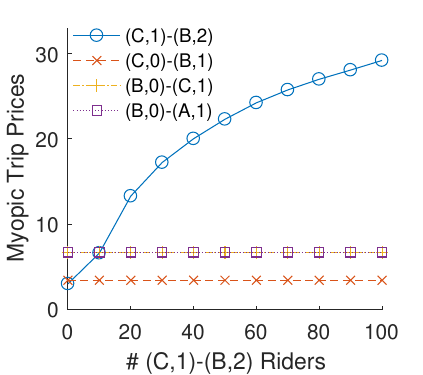}
    \caption{The myopic pricing mechanism. \label{fig:ex1_event_myopic_prices}}
\end{subfigure}%
\caption{Comparison of trip prices for the end of an event. \label{fig:ex1_prices} }    
\end{figure}

Figure~\ref{fig:ex1_regret_variance} illustrates the extent to which the myopic pricing mechanism failed to be incentive aligned or envy-free. With surging $\price_{C,B,1}$, drivers that are dispatched to  trips $(C,B,1)$ and $(B,A,1)$ may regret having not relocated to $C$ instead. %
Figure~\ref{fig:ex1_event_regret} shows that the average regret of the 25 drivers increases substantially as $N_{C,B,1}$ increases.
Among the 10 drivers who start at location $B$ at time $0$, the drivers taking the trip $(B,A,0)$ get a substantially smaller total payoff than those that take $(B,C,0)$ and subsequently $(C,B,1)$. Figure~\ref{fig:ex1_event_variance} shows the standard deviation (STD) of the total utilities of the drivers who start at $(B,0)$. The STP mechanism is incentive compatible and envy-free, thus the regret and earning variance are always zero.

\begin{figure}[t!]
\centering
\begin{subfigure}[t]{0.4\textwidth}
	\centering
    \includegraphics[height=\imageHeight in]{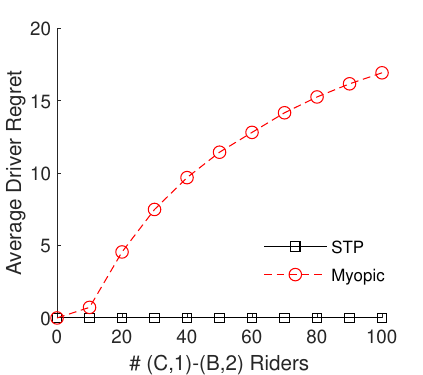}
    \caption{Average driver regret. \label{fig:ex1_event_regret}}
\end{subfigure}%
\begin{subfigure}[t]{0.4\textwidth}
	\centering
    \includegraphics[height=\imageHeight in]{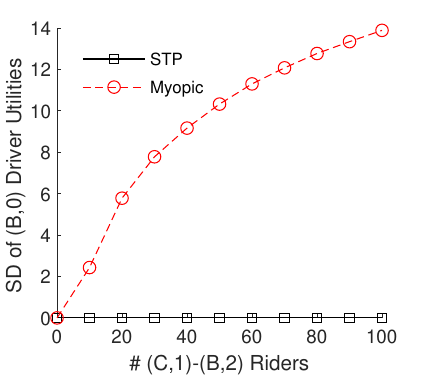}
    \caption{Standard deviation of $(B,0)$ drivers' utilities. \label{fig:ex1_event_variance}}
\end{subfigure}
\caption{Comparison of average driver regret and driver earning variance for the of an event. \label{fig:ex1_regret_variance} }    
\end{figure}

\subsection{Scenario Two: The Morning Rush Hour} \label{sec:sim_morning_rush}
We now compare the two mechanisms for the economy in
Figure~\ref{fig:ex2_morning_rush}, modeling the demand pattern of the morning rush hour. There are $\horizon = 20$ time periods and three locations $\loc = \{A, B, C\}$ with $\dist(a,b) = 1$ for all $a,b \in \loc$. Trip costs are $3$ per period, and exiting early costs $1$ per period: $\cost_{a,b,t} = 3\dist(a,b)$, and $\exitCost_{\Delta} = \Delta$. $C$ is a residential area, where there are a number of riders requesting rides to $B$, the downtown area, in every period. Location $A$ models some other area in the city.  %

\renewcommand{\nodeScaleSI}{0.8}
\renewcommand{\hdistSI}{2.7}
\renewcommand{\vdistSI}{1}	

\begin{figure}[t!]
\centering
\begin{tikzpicture}[->,>=stealth',shorten >=1pt, auto, node distance=2cm,semithick][font = \small]
\tikzstyle{vertex} = [fill=white,draw=black,text=black,scale=0.9]

\node[state]         (A0) [scale = \nodeScaleSI] {$A,0$};
\node[state]         (B0) [above of=A0, node distance = \vdistSI cm, scale=\nodeScaleSI] {$B,0$};
\node[state]         (A0) [scale = \nodeScaleSI] {$A,0$};
\node[state]         (C0) [above of=B0, node distance = \vdistSI cm, scale=\nodeScaleSI] {$C,0$};

\node[state]         (A1) [right of=A0, node distance = \hdistSI cm, scale=\nodeScaleSI] {$A,1$};
\node[state]         (B1) [right of=B0, node distance = \hdistSI cm, scale=\nodeScaleSI] {$B,1$};
\node[state]         (C1) [right of=C0, node distance = \hdistSI cm, scale=\nodeScaleSI] {$C,1$};

\node[state]         (A2) [right of=A1, node distance = \hdistSI cm, scale=\nodeScaleSI] {$A,2$}; 
\node[state]         (B2) [right of=B1, node distance = \hdistSI cm, scale=\nodeScaleSI] {$B,2$};
\node[state]         (C2) [right of=C1, node distance = \hdistSI cm, scale=\nodeScaleSI] {$C,2$};

\node[state]         (A9) [right of=A2, node distance = 1.2* \hdistSI cm, scale=\nodeScaleSI*0.93] {$A,19$}; 
\node[state]         (B9) [right of=B2, node distance = 1.2* \hdistSI cm, scale=\nodeScaleSI*0.93] {$B,19$};
\node[state]         (C9) [right of=C2, node distance = 1.2* \hdistSI cm, scale=\nodeScaleSI*0.93] {$C,19$};

\node[state]         (A10) [right of=A9, node distance = \hdistSI cm, scale = 0.7] {$A,20$}; 
\node[state]         (B10) [right of=B9, node distance = \hdistSI cm, scale = 0.7] {$B,20$};
\node[state]         (C10) [right of=C9, node distance = \hdistSI cm, scale = 0.7] {$C,20$};

\node[text width=3cm] at (8.2, 2 * \vdistSI){\dots};
\node[text width=3cm] at (8.2, \vdistSI){\dots};
\node[text width=3cm] at (8.2, 0 ) {\dots};

\node[text width=3cm] at (-0.6, 2 * \vdistSI ) {10 Drivers};
\node[text width=3cm] at (-0.6, \vdistSI) {10 Drivers};
\node[text width=3cm] at (-0.6, 0) {10 Drivers};

\path (C0) edge	node[pos=0.5, sloped, above] {$N_{C,B}$ riders} (B1);
\path (C1) edge	node[pos=0.5, sloped, above] {$N_{C,B}$ riders} (B2);
\path (C9) edge	node[pos=0.5, sloped, above] {$N_{C,B}$ riders} (B10);

\end{tikzpicture}
\caption{Morning rush hour. \label{fig:ex2_morning_rush}}
\end{figure}

In each economy, at time $t=0$, there are $10$ drivers starting in each of the three locations $A$, $B$ and $C$, who all stay until the end of the planning horizon. There are a total of $100$ riders with trip origins and destinations independently drawn at random from $\loc$ and trip starting times randomly drawn from $\activeTimeSet$. In addition, in each period there are $N_{C,B}$ commuters traveling from $C$ to $B$. We assume that the commuters' values for the rides are i.i.d. exponentially distributed with mean 20, whereas the random rides have values exponentially distributed with mean 10.

\begin{figure}[t!]
\centering
 \begin{subfigure}[t]{0.4\textwidth}
	\centering
    \includegraphics[height= \imageHeight in]{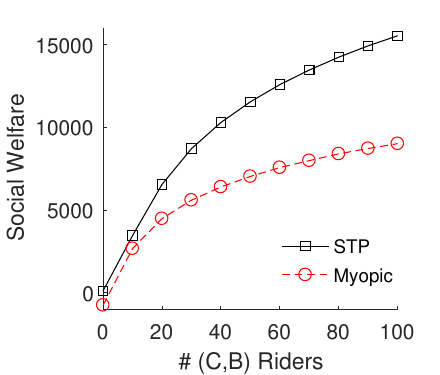}
    \caption{Social welfare.\label{fig:ex2_morning_welfare}}
\end{subfigure}%
 \hspace{2em}
\begin{subfigure}[t]{0.4\textwidth}
	\centering
    \includegraphics[height=\imageHeight in]{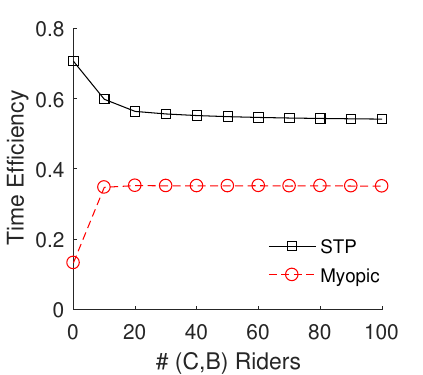}
    \caption{Driver time efficiency. \label{fig:ex2_morning_efficiency}}
\end{subfigure}%
\caption{Comparison of social welfare and driver time efficiency for the morning rush hours. \label{fig:ex2_welfare_efficiency}}    
\end{figure}

As we vary the $N_{C,B}$ from $0$ to $100$, the average social welfare achieved by the two mechanisms for $1,000$ randomly generated economies is as shown in Figure~\ref{fig:ex2_morning_welfare}. The STP mechanism achieves higher social welfare than the myopic pricing mechanism. 
Figure~\ref{fig:ex2_morning_efficiency} shows that the STP mechanism  achieves much higher driver time efficiency. The time efficiency of STP mechanism actually decreases as the number of $(C,B)$ riders per period increases above $10$, since the mechanism sends more empty cars to $C$ to pick up the higher value riders there. %

For the four origin-destination pairs, $(C,B)$, $(B,C)$, $(B,A)$ and $(A,B)$, Figures~\ref{fig:ex2_drivers} and~\ref{fig:ex2_prices} plot the average number of drivers getting dispatched 
to take these trips in each time period (including both trips with a rider, and repositioning without a rider), and the average trip prices. For each economy, the number of drivers for each origin-destination (OD) pair and the trip prices for this OD pair are averaged over the entire planning horizon.  Results on the other five trips, $(A,A)$, $(A,C)$, $(B,B)$, $(C,A)$ and $(C,C)$ can be interpreted similarly, therefore are deferred to Appendix~\ref{appx:additional_simulations}.

\begin{figure}[t!]
\centering
\begin{subfigure}[t]{0.4\textwidth}
	\centering
    \includegraphics[height= \imageHeight in]{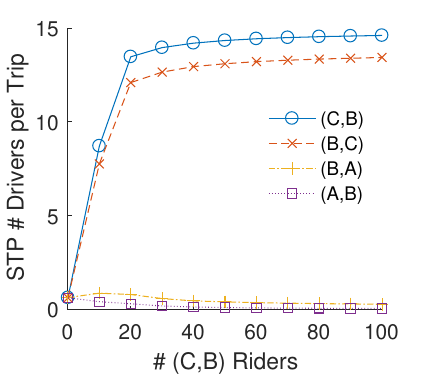}
    \caption{The STP mechanism. \label{fig:ex2_morning_stp_drivers}}
\end{subfigure}%
\begin{subfigure}[t]{0.4\textwidth}
	\centering
    \includegraphics[height=\imageHeight in]{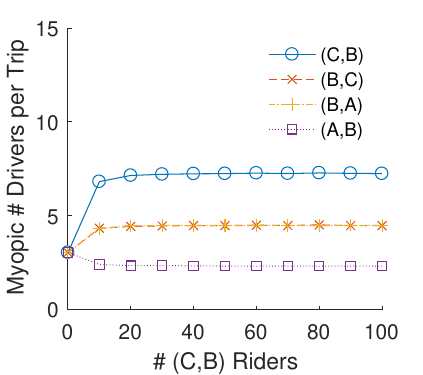}
    \caption{The myopic pricing mechanism.  \label{fig:ex2_morning_myopic_drivers}}
\end{subfigure}%
\caption{Comparison of the number of drivers per trip for the morning rush hour. \label{fig:ex2_drivers} }    
\end{figure}

Under the STP mechanism, given the large demand for trips from $C$ to
$B$ in each time period, there is a large number of drivers taking the
trip $(C,B)$, and also a large number of drivers relocating from $B$
to $C$ in order to pick up future riders from $C$ (see
Figure~\ref{fig:ex2_morning_stp_drivers}).
 A small number of drivers are dispatched from $B$ to $A$ due to the lack of future demand at $A$. Because of the abundance of supply at $B$ that are brought in by the $(C,B)$ trips, very few drivers are sent from $A$ to $B$. 
Under the myopic pricing mechanism, the number of drivers dispatched to take each trip, in contrast, does not contribute to the repositioning of drivers. See Figure~\ref{fig:ex2_morning_myopic_drivers}. There are an equal number of drivers traveling from $B$ to $A$ and $C$ despite the significant difference in the demand conditions at the two destinations. Moreover, a non-trivial number of drivers are traveling from $A$ to $B$ despite the fact that there are already too much of supply at location $B$.

\begin{figure}[t!]
\centering
\begin{subfigure}[t]{0.4\textwidth}
	\centering
    \includegraphics[height= \imageHeight in]{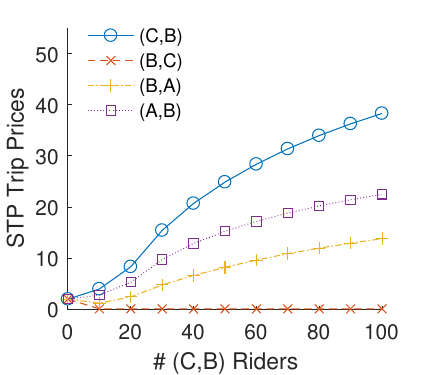}
    \caption{The STP mechanism.\label{fig:ex2_morning_stp_prices}}
\end{subfigure}%
\begin{subfigure}[t]{0.4\textwidth}
	\centering
    \includegraphics[height=\imageHeight in]{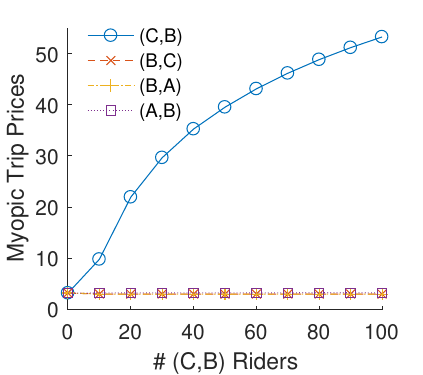}
    \caption{The myopic pricing mechanism. \label{fig:ex2_morning_myopic_prices}}
\end{subfigure}%
\caption{Comparison of trip prices for the morning rush hour. \label{fig:ex2_prices} }    
\end{figure}

Regarding the average prices under the STP mechanism plotted in Figure~\ref{fig:ex2_morning_stp_prices}, the morning commute route $(C,B)$ has a higher average price due to the large demand for this trip. The $(B,A)$ trip is less costly than the $(A,B)$ trip since there is plenty of supply of drivers that are brought to $B$ by the $(C,B)$ trips, so that the marginal value of supply at $B$ is low. The $(A,B)$ price is high so that not too many drivers are dispatched from $A$ to $B$. The $(B,C)$ trips are priced almost always at zero, despite the fact that the trip cost for the drivers is $3$, since it is beneficial for the economy for drivers to move to $C$ to pick up the commuters. 
Finally, Figure~\ref{fig:ex2_morning_myopic_prices} shows that the $(C,B)$ trip has a much higher average price under the myopic pricing mechanism than the STP mechanism, whereas the rest of the three trips are priced at the trip costs--- this is because without optimizing for the supply of drivers at each location,  $A$ and $B$ almost always have plenty of supply to pick up all riders starting from these locations, whereas there is far from enough drivers to satisfy the large demand at $C$.

\subsection{Scenario Three: Unbalanced Airport Trips} \label{sec:sim_airport}

In this scenario, we consider the imbalance between trips to and from the airport, as illustrated in Figure~\ref{fig:ex3_airport}. There are a total of $\horizon = 20$ time periods and two locations $\loc = \{A, D\}$, modeling the airport and the downtown area respectively. $\dist(A, A) = \dist(D,D) = 1$, whereas trips in between downtown and the airport are longer: $\dist(A,D) = \dist(D,A) = 2$. Trip costs are $3$ per period, and exiting early costs $1$ per period: $\cost_{a,b,t} = 3\dist(a,b)$, and $\exitCost_{\Delta} = \Delta$.

\renewcommand{\nodeScaleSI}{0.8}
\renewcommand{\hdistSI}{2.8}
\renewcommand{\vdistSI}{1.6}	

\begin{figure}[t!]
\centering
\begin{tikzpicture}[->,>=stealth',shorten >=1pt, auto, node distance=2cm,semithick][font = \small]
\tikzstyle{vertex} = [fill=white,draw=black,text=black,scale=0.9]

\node[state]         (A0) [scale = \nodeScaleSI] {$A,0$};
\node[state]         (B0) [above of=A0, node distance = \vdistSI cm, scale=\nodeScaleSI] {$D,0$};
\node[state]         (A0) [scale = \nodeScaleSI] {$A,0$};

\node[state]         (A1) [right of=A0, node distance = \hdistSI cm, scale=\nodeScaleSI] {$A,1$};
\node[state]         (B1) [right of=B0, node distance = \hdistSI cm, scale=\nodeScaleSI] {$D,1$};

\node[state]         (A2) [right of=A1, node distance = \hdistSI cm, scale=\nodeScaleSI] {$A,2$};
\node[state]         (B2) [right of=B1, node distance = \hdistSI cm, scale=\nodeScaleSI] {$D,2$};

\node[state]         (A3) [right of=A2, node distance = \hdistSI cm, scale=\nodeScaleSI] {$A,3$};
\node[state]         (B3) [right of=B2, node distance = \hdistSI cm, scale=\nodeScaleSI] {$D,3$};

\node[state]         (A10) [right of=A3, node distance = 1.2*\hdistSI cm, scale = 0.75] {$A,20$}; 
\node[state]         (B10) [right of=B3, node distance = 1.2*\hdistSI cm, scale = 0.75] {$D,20$};

\node[text width=1cm] at (\hdistSI * 3 + 2,  \vdistSI ) {\dots};
\node[text width=1cm] at (\hdistSI * 3 + 2, 0 ) {\dots};

\node[text width=3cm] at (-0.6, \vdistSI ) {10 Drivers};
\node[text width=3cm] at (-0.6, 0) {10 Drivers};

\path (B0) edge	node[pos=0.55, sloped, above] {$40$ riders} (B1);
\path (B1) edge	node[pos=0.55, sloped, above] {$40$ riders} (B2);
\path (B0) edge	node[pos=0.2, sloped, below] {$N_{D,A}$ } (A2);
\path (A0) edge	node[pos=0.25, sloped, below] {$40$ - $N_{D,A}$} (B2);
\path (B1) edge	node[pos=0.75, sloped, below] {$N_{D,A}$ } (A3);
\path (A1) edge	node[pos=0.8, sloped, below] {$40$ - $N_{D,A}$} (B3);

\end{tikzpicture}
\caption{Imbalance in trips to and from the airport. \label{fig:ex3_airport}}
\end{figure}

In each economy, there are $20$ available drivers at each of $A$ and $D$ at time $0$. Within the downtown area, there are $40$ riders requesting rides in each period, whereas in between the downtown area and the airport, there are a total of $40$ riders heading toward or leaving the airport, which may be unevenly distributed on the two directions. 
The value of each of the downtown riders are drawn i.i.d from the exponential distribution with mean $10$, and the value for each trip to or from the airport is drawn i.i.d from the exponential distribution with mean $40$. Since the airport trips are twice as long, we are modeling the scenario where the airport travelers are less price sensitive, and are willing to pay twice as much, in comparison to the downtown riders.

\begin{figure}[t!]
\centering
 \begin{subfigure}[t]{0.33\textwidth}
	\centering
    \includegraphics[height= \imageHeight in]{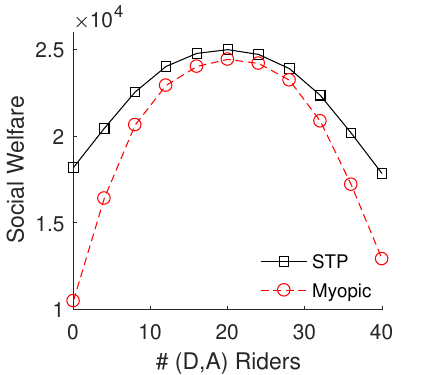}
    \caption{Social welfare.\label{fig:ex3_airport_welfare}}
\end{subfigure}%
\hspace{2em}
\begin{subfigure}[t]{0.33\textwidth}
	\centering
    \includegraphics[height=\imageHeight in]{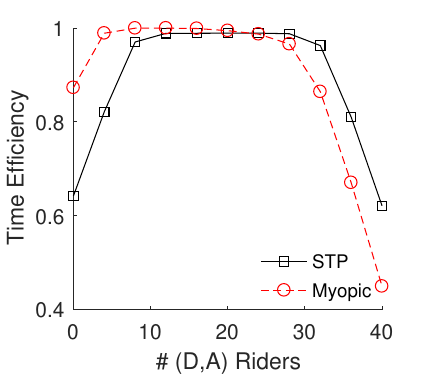}
    \caption{Driver time efficiency. \label{fig:ex3_airport_efficiency}}
\end{subfigure}%
\caption{Comparison of social welfare and driver time efficiency for the morning rush hours. \label{fig:ex3_welfare_efficiency}}    
\end{figure}

As we vary  $N_{D,A}$ from $0$ to $40$ (thus at the same time varying the number of $(A,D)$ riders from $40$ to $0$), the average social welfare and driver-time efficiency achieved by the two mechanisms over $1,000$ randomly generated economies are as shown in Figure~\ref{fig:ex3_welfare_efficiency}. We first observe that the more balanced the trip flow to and from the airport is (i.e. when $N_{D,A}$ is closer to $20$), the higher the social welfare and driver time efficiency achieved by STP. This is because when the trip flow is more balanced, it is more likely for a driver to pick up riders with high values for the trips both to and from the airport, whereas when the flow is unbalanced, drivers may relocate with an empty car or pick up riders with low values for one of the two directions. 
The myopic pricing mechanism achieves comparatively good welfare and driver time-efficiency when trip flows are reasonably balanced, however, the performance downgrades quickly as the flow becomes more unbalanced.

\begin{figure}[t!]
\centering
\begin{subfigure}[t]{0.4\textwidth}
	\centering
    \includegraphics[height= \imageHeight in]{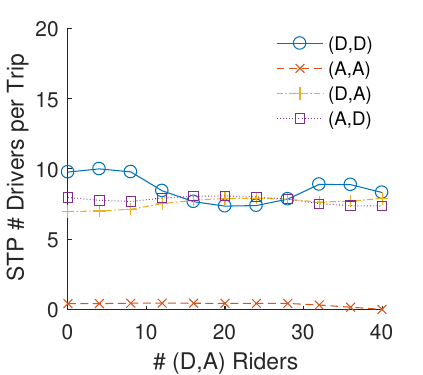}
    \caption{The STP mechanism. \label{fig:ex3_airport_stp_drivers}}
\end{subfigure}%
\begin{subfigure}[t]{0.4\textwidth}
	\centering
    \includegraphics[height=\imageHeight in]{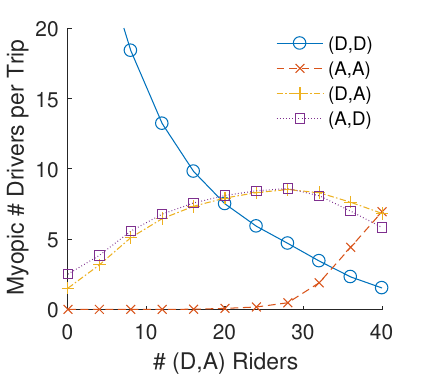}
    \caption{The myopic pricing mechanism.  \label{fig:ex3_airport_myopic_drivers}}
\end{subfigure}%
\caption{Comparison of driver numbers for the unbalanced trips to and from the airport. \label{fig:ex3_drivers}}    
\end{figure}

Figure~\ref{fig:ex3_airport_stp_drivers} shows that regardless of the flow imbalance, the number of drivers doing downtown or airport trips stay reasonably stable under the STP mechanism. %
However, Figure~\ref{fig:ex3_airport_myopic_drivers} shows that when $N_{D,A}$ is too small, too many drivers staying in downtown, forgoing the opportunity to pick up the large number of drivers hoping to return to downtown from the airport. Similarly, when $N_{D,A}$ is too large, the myopic pricing mechanism sends too many drivers to the airport, since they have higher per period surplus. As a result, a large number of drivers line up at the airport, and too few rider traveling within downtown are picked up, resulting in much lower efficiency.

\begin{figure}[t!]
\centering
\begin{subfigure}[t]{0.4\textwidth}
	\centering
    \includegraphics[height= \imageHeight in]{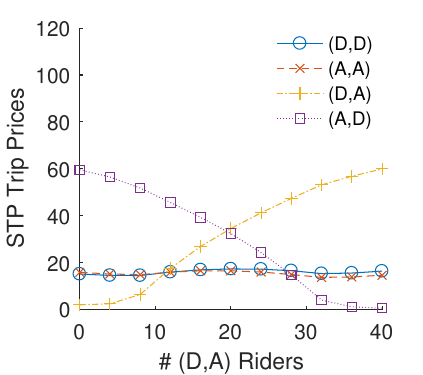}
    \caption{The STP mechanism.\label{fig:ex3_airport_stp_prices}}
\end{subfigure}%
\begin{subfigure}[t]{0.4\textwidth}
	\centering
    \includegraphics[height=\imageHeight in]{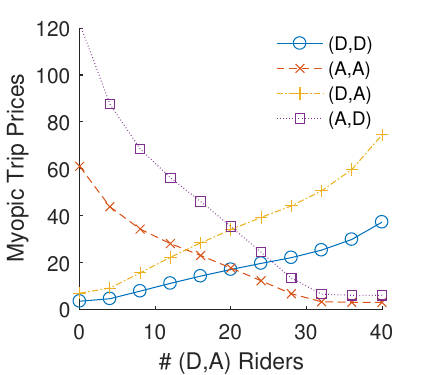}
    \caption{The myopic pricing mechanism. \label{fig:ex3_airport_myopic_prices}}
\end{subfigure}%
\caption{Comparison of trip prices for the unbalanced trips to and from the airport. \label{fig:ex3_prices} 
}
\end{figure}

The average prices for trips under the two mechanisms are shown in Figure~\ref{fig:ex3_prices}. Comparing the two mechanisms, we can see that the price for the $(D,D)$ downtown trip is almost constant, regardless of how unbalanced the airport trip flows are, however, the price for the downtown trip is seriously affected by the conditions at the airport. For the trips to and from the airport, we can see that both mechanisms increase the prices for the direction that is over-demanded, and lowers the price for other direction. The differences are (i) the price surges are much lower under STP than under myopic, providing riders more price stability, and (ii) the price for the under-demanded direction is close to zero under STP, reflecting the need to relocate cars even when there is no demand, however, myopic insists on setting a prices to at least cover the trip cost.

\section{Concluding Remarks} \label{sec:conclusion}

We study the problem of pricing and dispatching in ridesharing platforms in a way that drivers would always choose to accept the platform's dispatches. Under a complete information, discrete time, multi-period and multi-location model, we show that always accepting the dispatched trips forms a subgame-perfect equilibrium among the drivers under the STP mechanism. The STP mechanism is simple in design: computing a driver-pessimal competitive equilibrium plan at the beginning of the planning horizon, as well as after any deviation from this plan. %
The main operational insight from this paper, as discussed in the introduction, is the role of the marginal welfare contribution of adding an extra driver at the origin versus at the destination of a trip (as prescribed in \eqref{equ:stp_prices}), in determining prices that are appropriately smooth, and optimizing social welfare.

\paragraph{Welfare vs revenue optimization.}

We focus on welfare instead of revenue optimization in this paper. With the substantial network effect in ridesharing marketplaces, as well as fierce competition, major platforms such as Uber and Lyft have been prioritizing growth instead of profit. Welfare optimization is therefore aligned with this goal, %
and should also be the objective when a city or a non-profit organization (e.g. RideAustin) is operating a ridesharing platform.

In regard to revenue: the STP mechanism strictly balances budget. Alternatively, we many think about the platform taking a fixed fraction of the driver surplus. This does not affect the results presented in this paper, and the outcome under STP still resides in some $\epsilon$-core of the market, depending on the cut taken by the platform.\footnote{The (multiplicative) $\epsilon$-core of a game is the set of outcomes where there does not exist a coalition of participants who can form an alternative plan among themselves, such that all participants' utility increases by a fraction of $\epsilon$.%
}

\paragraph{Modeling assumptions.}

Throughout the paper, we assume that drivers have homogeneous costs and will all stay until the end of the planning horizon, and that riders are impatient. Examples~\ref{exmp:different_exit_times} and~\ref{exmp:patient_riders} in Appendix~\ref{appx:integrality} show that if either assumption is relaxed, the LP relaxation \eqref{equ:lp} for the optimal planning problem may no longer be integral since the reduction to the MCF problem fails. As a result, there may not exist anonymous, origin-destination CE prices. %

As the length of a time period becomes smaller, it is less likely for riders to be fully impatient, although prices that are appropriately smooth in time would reduce riders' incentives to strategically wait for a lower price. As a time period gets shorter, the envy-freeness
property %
also weakens--- the``equivalence of
rights'' depends on time and location, and there would be fewer riders
requesting the same trip, or drivers at the same location at the same time.

\paragraph{Complete information.} 

We assume complete information and a finite planning horizon throughout this paper. A natural next step is to generalize the model to settings where there is uncertainty about supply and demand, and where the planning horizon rolls forward as the uncertainty resolves over time. %
One challenge is that the class of $M^\natural$ concave functions is not closed under addition, thus the continuation value (i.e. the expected future welfare as a function of the positioning of drivers at the end of this immediate planning period) may no longer be $M^\natural$ concave. This may affect the integrality of the optimal planning problem and the existence of CE.
On the other hand, while we have made use of CE to establish envy-freeness and core-selecting properties, the use of CE prices is not necessary for incentive alignment, since a driver cannot simply pick a series of trips to complete, so this does not by itself prevent progress.
Building predictive models for supply and demand
in ridesharing platforms has proved challenging,
and it may be more promising and practical to directly estimate the welfare contributions of an extra driver using existing data from the current operations.

\paragraph{Information elicitation.} 

The focus of this paper is incentive alignment in a dynamic environment while maintaining driver
flexibility, and not on information asymmetry. This said, we do have some results in this regard.
On the driver's side, we may consider a scenario where the mechanism does not have the driver entrance information, but asks each driver to report at the beginning of the planning horizon %
the time and location at which she will enter.
We show in Appendix~\ref{appx:driver_entrance} that a driver with entrance location and time
$(\re_i, \te_i)$ has no incentive to report some
feasible entrance location and time $(\hat{\re_i}, \hat{\te_i})$ where
$\hat{\te_i} \geq \te_i + \dist(\re_i, \hat{\re_i})$, and then
actually enter the platform at $(\hat{\re_i}, \hat{\te_i})$.

On the rider side, Example~\ref{exmp:toy_econ_1_continued} in Appendix~\ref{appx:rider_incentives} shows that the STP mechanism is not truthful for riders.
This is a departure from the classical, unit-demand assignment problem~\citep{shapley1971assignment}, where the seller-pessimal outcome corresponds to the buyer-side VCG prices, and is truthful for buyers.
We show in Theorem~\ref{thm:minCE_riderVCG} that the rider-side VCG
payment for a rider is equal to the minimum price for her trip among
all CE outcomes, but such trip prices, when adopted for all
riders, may not form a CE (i.e., the anonymous trip prices among all CE outcomes does not form a lattice). This implies that no budget-balanced core-selecting mechanism is truthful for riders, and we also show that no budget-balanced, optimal and subgame-perfect incentive compatible mechanism is truthful for riders.

\bibliography{uber-arxiv-refs}

\newpage

\appendix

\noindent{}\textbf{\huge{Appendix}}

\bigskip

\noindent{}%
Appendix~\ref{appx:cont_time} provides a continuous-time interpretation of the discrete time model that we adopted. Appendix~\ref{appx:proofs} includes proofs that are omitted from the body of
the paper. 
Appendix~\ref{appx:examples} provides examples and discussions on %
the integrality of the LP relaxation and existence of CE, incentives of riders, entrance as drivers' private information, and the naive recomputation of optimal CE plans.
Appendix~\ref{appx:literature} discusses the relationship
with the literature on trading networks and the dynamic VCG mechanism, and why they do not solve the ridesharing problem. Finally, additional simulation results are presented in
Appendix~\ref{appx:additional_simulations}. %

\section{Continuous Time Interpretation} \label{appx:cont_time}

Under the discrete time model introduced in Section~\ref{sec:preliminaries}, trips within the same location takes $\delta(a,a) = 1$ unit of time for all $a \in \loc$, however, we also assume that a driver can drop-off a rider and pick-up a new rider in the same location at the same time point. 

\begin{figure}[htpb!]
\centering
\begin{tikzpicture}[scale = 0.98] %
\draw[-latex, line width=0.25mm] 	(-0.5,0) -- (15.5,0)node[anchor=north] {time};

\draw[-, , line width=0.5mm] 	(0,0.1) -- (0,-0.1);
\draw[-, , line width=0.5mm] 	(5,0.1) -- (5,-0.1);
\draw[-, , line width=0.5mm] 	(10,0.1) -- (10,-0.1);
\draw[-, , line width=0.5mm] 	(15,0.1) -- (15,-0.1);

\draw(0, 0.1) node[anchor=south] {$t=0$};
\draw(5, 0.1) node[anchor=south] {$t=1$};
\draw(10, -0.1) node[anchor=north] {$t=2$};
\draw(15, 0.1) node[anchor=south] {$t=3$};
    
\node[text width=2.3cm] at (1, -0.8) 
    {Pick-up at $A$};

\node[text width=2.3cm] at (4,-0.8) 
    {Drop-off at $A$};
        
\node[text width=2.3cm] at (7,-0.8) 
    {Pick-up at $A$};
            
\node[text width = 2.5cm] at (13.5,-0.8) 
    {Drop-off at $B$};

\draw[->](1, -0.5) -- (1, -0.1);
\draw[->](4, -0.5) -- (4, -0.1);
\draw[->](7, -0.5) -- (7, -0.1);
\draw[->](13.5, -0.5) -- (13.5, -0.1);
                    
\draw [decorate, decoration={brace, amplitude=10pt},xshift=0.4pt,yshift=-0.4pt] (1,0) -- (4,0) node[black,midway,yshift=0.6cm] { $A$-$A$ trip};

\draw [decorate, decoration={brace, amplitude=10pt},xshift=0.4pt,yshift=-0.4pt] (7,0) -- (13.5,0) node[black,midway,yshift=0.6cm] { $A$-$B$ trip};
                    
\end{tikzpicture}

\caption{Time-line for a within-location trip $A \rightarrow A$ which takes $\dist(A,A) = 1$ period of time, and a between-location trip $A \rightarrow B$ which takes $\dist(A, B) = 2$ periods of time.
\label{fig:trip_timeline}} 
\end{figure}

Figure~\ref{fig:trip_timeline} illustrates the continuous-time interpretation of this discrete-time model. There are two trips: an $A$ to $A$ trip at time $0$ which takes $\dist(A,A) = 1$ unit of time, and an $A$ to $B$ trip at time $1$, which takes $\dist(A,B) = 2$ units of time and ends at time $t  = 3$. The time after the drop-off of the first rider at $A$ and the time for the second pick-up at $A$ is the time the driver takes to travel within location $A$ to pick up the second rider. In the discrete time model, both the drop-off and the pick-up are both considered to happen at time $t=1$.

\section{Proofs} \label{appx:proofs}

We provide in this section proofs that are omitted from the body of the paper.

\subsection{Proof of Lemma~\ref{lem:CE_plan_properties}} \label{appx:proof_lem_CE_plan}

\lemCEPlan*

\begin{proof} %
Assume that there exists $(a,b,t) \in \trips$ s.t. $\exists j \in \riders$ s.t. $(\origin_j, \dest_j, \tr_j) = (a,b,t)$ and $\price_{a,b,t} < 0$. Rider best response implies $x_j = 1$, thus there exists $i \in \driverSet$ s.t. $(\origin_j, \dest_j, \tr_j, j) \in \actpath_i$, and this is paid $\price_{a,b,t} < 0$ at time $t$. This violates driver best response, since keeping the rest of the action path unchanged, but choosing not to get paid for this trip, the driver would get a higher total payment.

This implies that $\price_{\origin_j, \dest_j,\tr_j} = \price_{\origin_j, \dest_j,\tr_j}^+$ holds for all $j \in \riders$, therefore given the dispatching $(x,\actpath)$, rider best response under prices $\price$ implies rider best response under prices $\price^+$. 
Denote the total payment to each driver and the utility of each driver given $(x, \actpath, \price)$ as $\dPayment_i^+$ and $\pi_i^+$.
We know that for each driver $i \in \driverSet$, $\dPayment_i^+ = \sum_{j \in \riders} \one{(\origin_j,\dest_j,\tr_j,j) \in \actpath_i} \price_{\origin_j,\dest_j,\tr_j}^+
= \sum_{j \in \riders} \one{(\origin_j,\dest_j,\tr_j,j) \in \actpath_i} \price_{\origin_j,\dest_j,\tr_j} = \dPayment_i$, thus $\pi_i^+ = \pi_i$. This implies driver best response: $\pi_i^+ = \pi_i = 
\max_{k = 0, \dots, |\pathSet_i|} \{ \sum_{(a,b,t) \in  \path_{i,k} } \max \{ \price_{a,b,t} , 0 \} - \pathCost_{i,k} \}
=\max_{k = 0, \dots, |\pathSet_i|}  \{ \sum_{(a,b,t) \in  \path_{i,k} } \max \{ \price_{a,b,t}^+ , 0 \} - \pathCost_{i,k} \}$,
This completes the proof that $(x, \actpath, \price)$ also forms a CE, and that the driver and rider payments and utilities under the two plans are identical. 
\end{proof}

\subsection{Proof of Lemma~\ref{thm:lp_integrality}} \label{appx:proof_thm_lp_integrality}

\subsubsection{Minimum Cost Flow Problems} \label{appx:mcf}

We first provide the formulation of the minimum cost flow (MCF) problem, and the reduction of an optimal dispatching problem to an MCF problem where drivers flow through a network with nodes corresponding to  the initial states of drivers and the (location, time) pairs, and the edge costs equal to the trip costs for drivers minus the rider values. 

\medskip

Let $\graph = (\nodeSet, \edgeSet)$ be a directed graph with a node set $\nodeSet$ and an edge set $\edgeSet$.
Let $\minCap:\edgeSet \rightarrow \setZ \cup \{-\infty \}$ be the lower capacity function, $\maxCap:\edgeSet \rightarrow \setZ \cup \{ +\infty \}$ be the upper capacity function, and let $\gamma: \edgeSet \rightarrow \setR$ be the cost function.
For each edge $\edge \in \edgeSet$, denote $\tail \edge \in \nodeSet$ as the initial (tail) node of $\edge$ and $\head \edge \in \nodeSet$ as the terminal (head) node of $\edge$. That is, $\tail \edge = \node_1$ and $\head \edge = \node_2$ for the edge $\edge = (\node_1, \node_2)$.

A feasible flow $\flow$ is a function $\flow: \edgeSet \rightarrow \setR$ such that $\minCap(\edge) \leq \flow(\edge) \leq \maxCap(\edge)$ for each $\edge \in \edgeSet$. Its boundary $\partial \flow: \nodeSet \rightarrow \setR$ is defined as
\begin{align}
	\partial \flow(\node) =  \sum \{ \flow(\edge) ~|~ \edge \in \edgeSet,~\tail \edge = \node\} -  \sum \{ \flow(\edge) ~|~ \edge \in \edgeSet,~\head \edge = \node\}.
\end{align}
A node $\node$ for which $\partial \flow(\node) > 0$ is a source of the flow, and a node $\node$ is a sink if $\partial \flow(\node) < 0$.
Let $\xi$ be a vector in $\setR^{|\nodeSet|}$, the minimum cost for any flow with boundary condition $\xi$ is:
\begin{align}
	\omega(\xi) = \inf_{\flow} \left\lbrace \left. \sum_{\edge \in \edgeSet}  \cost(\edge)\flow(\edge) \right|~ \flow: \text{ feasible flow with } \partial \flow = \xi \right\rbrace. \label{equ:mcf}
\end{align}

\subsubsection{Reducing Optimal Dispatching to MCF} \label{appx:planning_to_mcf}

Given an instance of the optimal dispatching problem with planning horizon $\horizon$, locations $\loc$, distances $\dist$, costs $\{ \cost_{a,b,t} \}_{(a,b,t) \in \trips}$ and $\{\exitCost_{\Delta}\}_{\Delta = 1,\dots, \horizon}$, riders $\riders$ and drivers $\driverSet$, we construct a corresponding MCF problem. Let $\graph = (\nodeSet, \edgeSet)$ be the graph, where the nodes $\nodeSet$ consists of (location, time) pairs, the nodes $\{\driverNode_i\}_{i \in \driverSet}$ modeling the initial states of drivers, and an additional ``sink" node $\sink$ representing the end of time:
\begin{align*}
	\nodeSet = \set{(a,t)}{a \in \loc,~ t \in \timeSet}\cup \set{\driverNode_i}{i \in \driverSet} \cup \{ \sink \}.
\end{align*}
The set of edges $\edgeSet = \edgeSet_1 \cup \edgeSet_2 \cup \edgeSet_3 \cup \edgeSet_4$ consists of the following parts:
\begin{enumerate}[$\bullet$]
	\item  $\edgeSet_1 = \set{\riderEdge_j}{ j \in \riders}$ corresponds to rider trips, %
where the edge $\riderEdge_j = ((\origin_j, \tr_j), ~(\dest_j, \tr_j + \dist(\origin_j, \dest_j)))$ corresponds to the trip requested by rider $j$ and has minimum capacity $\minCap(\riderEdge_j) = 0$, maximum capacity $\maxCap(\riderEdge_j)= 1$ and cost $\edgeCost(\riderEdge_j) = -\val_j + \cost_{\origin_j, \dest_j, \tr_j}$. Intuitively, if a unit of driver flows through the edge corresponding to rider $j$ (i.e. rider $j$ is picked up by a driver), we incur a cost of $\cost_{\origin_j, \dest_j, \tr_j}$, and  gain value $\val_j$.
	\item $\edgeSet_2$ consists of edges that are feasible relocating trips without riders:
\begin{align*}
	\edgeSet_2 = \set{((a, t),~(b, t+\dist(a,b)) }{(a,b,t) \in \trips}.
\end{align*}
Recall that $\trips = \{(a,b,t) ~|~ a,b \in\loc, t + \dist(a,b) \leq \horizon \}$ denotes the set of feasible trips within the planning horizon.
	There is no upper bound on capacities of these edges: $\forall \edge \in \edgeSet_2$, $\minCap(\edge) = 0$, $\maxCap(\edge) = +\infty$. The edge costs are $\edgeCost(\edge) = \cost_{a,b,t}$, if $\edge = ((a, t),~(b, t+\dist(a,b))$.
	\item $\edgeSet_3$ consists of edges that connect all nodes $(a,t)$ to the sink $\sink$, representing the exit of a driver from location $a$ and time $t$:
\begin{align*}
	\edgeSet_3 = \set{((a, t),~\sink)}{a \in \loc,~t \in \timeSet}.
\end{align*}
	Similar to $\edgeSet_2$, there is no upper capacity constraint: $\forall \edge \in \edgeSet_3$, $\minCap(\edge) = 0$, $\maxCap(\edge) = +\infty$. A driver exiting at time $t$ incurs an early exiting opportunity cost of $\exitCost_{\horizon - t}$, therefore $\edgeCost(\edge) = \exitCost_{\horizon - t}$ for all $\edge = ((a, t),~\sink) \in \edgeSet_3$.
	\item $\edgeSet_4$ consists of two sets of edges. The first set of edges allow each driver to enter the platform at their entrance location and time $(\re_i, \te_i)$, and the second set of edges allow drivers with $\driverEntrance_i = 0$ to not enter the platform at all:
	\begin{align*}
		\edgeSet_4 = \set{(\driverNode_i,~(\re_i, \te_i))}{i \in \driverSet} \cup \set{(\driverNode_i,~\sink)}{i \in \driverSet \txtst \driverEntrance_i = 0}.
	\end{align*}		
	There is no upper capacity constraint, and no additional cost for these edges: $\forall \edge \in \edgeSet_4$, $\minCap(\edge) = 0$, $\maxCap(\edge) = +\infty$, and $\edgeCost(\edge) = 0$.
\end{enumerate}

The boundary condition of the MCF problem is given by:
\begin{align*}
	\xi_{\driverNode_i} =&~ 1,~\forall i \in \driverSet, \\
	\xi_\node =&~ 0,~ \txtif \node = (a,t) \text{ for some } a \in \loc \txtand t \in \timeSet,\\
	\xi_\sink =&~ -\nd.
\end{align*}

\paragraph{Flow LP}

Given this construction, there are non-zero edge costs and upper flow capacity constraints only for edges in $\edgeSet_1$. The minimum cost flow problem \eqref{equ:mcf} can therefore be simplified and rewritten in the following form:
\begin{align}
\min_{\flow} ~ &  \sum_{j \in \riders} (\cost_{\origin_j,\dest_j,\tr_j}-\val_j) \flow(\riderEdge_j)  + \sum_{(a,b,t)\in \trips} \cost_{a,b,t} \flow(((a,t),~(b,t+\dist(a,b)))) \notag \\
&  + \sum_{a \in \loc, t \in \timeSet} \exitCost_{\horizon - t} \flow(((a,t),~\sink))
	 \label{equ:flow_LP}\\
	\txtst &  \sum_{\edge \in \edgeSet, ~\tail \edge = (a,t)} \flow(\edge) -   \sum_{\edge \in \edgeSet, ~\head \edge = (a,t)} \flow(\edge) = 0,  & \forall a \in \loc, ~\forall t \in \timeSet 	
	 \label{equ:flow_balance_constraint} \\
	& \sum_{\edge \in \edgeSet, ~\tail \edge = \driverNode_i} \flow(\edge) = 1,  & \forall i \in \driverSet  \label{equ:flow_driver_cnst}	 \\
	& \flow(\riderEdge_j) \leq 1, & \forall j \in \riders  \label{equ:flow_capacity_constraint} \\
	& \flow(\edge) \geq 0, & \forall \edge \in \edgeSet \label{equ:flow_non_neg_constraint}
\end{align}

Note that given \eqref{equ:flow_balance_constraint} and \eqref{equ:flow_driver_cnst}, the flow balance constraint at the sink, $\sum_{\edge \in \edgeSet,~\head \edge = \sink} \flow(\edge) = \nd$, is redundant, and therefore omitted from the above formulation in order to achieve better interpretability of the dual variables. 
Observing that minimizing the negation of the total value of riders that are picked up is equivalent to maximizing the total value of riders that are picked up, we can rewrite \eqref{equ:flow_LP} in the following form:
\begin{align}
\max_{\flow} ~ &  \sum_{j \in \riders} (\val_j - \cost_{\origin_j,\dest_j,\tr_j}) \flow(\riderEdge_j)  - \sum_{(a,b,t)\in \trips} \cost_{a,b,t} \flow(((a,t),~(b,t+\dist(a,b)))) \notag \\
&  - \sum_{a \in \loc, t \in \timeSet} \exitCost_{\horizon - t} \flow(((a,t),~\sink)) \label{equ:flow_LP_simp} \\
	\txtst &  \sum_{\edge \in \edgeSet, ~\tail \edge = (a,t)} \flow(\edge) -   \sum_{\edge \in \edgeSet, ~\head \edge = (a,t)} \flow(\edge) = 0,  & \forall a \in \loc, ~\forall t \in \timeSet \label{equ:flow_balance_cnst_simp} \\
	& \sum_{\edge \in \edgeSet, ~\tail \edge = \driverNode_i} \flow(\edge) = 1,  & \forall i \in \driverSet  \label{equ:flow_driver_cnst_simp}	 \\	
	& \flow(\riderEdge_j) \leq 1, & \forall j \in \riders \label{equ:flow_capacity_cnst_simp} \\
	& \flow(\edge) \geq 0, & \forall \edge \in \edgeSet \label{equ:flow_nonneg_cnst_simp}
\end{align}

We refer to \eqref{equ:flow_LP_simp} as the flow LP.

\subsubsection{Proof of Lemma~\ref{thm:lp_integrality}}

\thmLPIntegrality*

\begin{proof}

It is known that the MCF problems with certain structure have integral optimal solutions~\cite{murota2003discrete}.
The flow LP \eqref{equ:flow_LP_simp} is integral since (I) the flow balance constraints \eqref{equ:flow_balance_cnst_simp}  and \eqref{equ:flow_driver_cnst_simp} can be written in matrix form $\flowMatrix \flow = \xi$ where $\flowMatrix$ is total unimodular and $\xi$ has only integer entries and (II) the edge capacity constraints \eqref{equ:flow_capacity_cnst_simp} are all integral. See Section III.1.2 in~\cite{WolseyLaurenceA1999IaCO} for details on total unimodularity and the integrality of polyhedron.

\medskip

To prove the integrality of the original LP \eqref{equ:lp}, we show that 
\begin{enumerate}[(i)]
	\item for each feasible solution to the LP \eqref{equ:lp}, there exists a feasible solution to the flow LP \eqref{equ:flow_LP_simp} with the same objective, and
	\item for each integral feasible solution to the flow LP  \eqref{equ:flow_LP_simp}, there exists a corresponding integral feasible solution to the LP \eqref{equ:lp} with the same objective.
\end{enumerate}

The integrality of MCF then implies that there exists an integral optimal solution of \eqref{equ:lp}, since the optimal objective of the LP \eqref{equ:lp} cannot exceed the optimal objective of the MCF, which is achieved at some integral feasible solution of the MCF, and therefore also at some integral feasible solution of \eqref{equ:lp}.
We now prove (i) and (ii).

\bigskip

\noindent{}\textit{Part (i).} 
Let $(x,y)$ be a feasible solution to the LP \eqref{equ:lp}.  
A solution $\flow$ to the flow LP \eqref{equ:flow_LP_simp} can be constructed as follows:
\begin{enumerate}[$\bullet$]
	\item For each $j\in \riders$, let $\flow(\riderEdge_j)= x_j$. We know that $0 \leq \flow(\riderEdge_j) \leq 1$ for all $j \in \riders$. 
	\item For each $\edge =((a,t),~ (b, t+\dist(a,b)) \in \edgeSet_2$ corresponding to the relocation trip $(a,b,t) \in \trips$, let $\flow(\edge) = \sum_{i \in \driverSet} \sum_{k = 0}^{|\pathSet_i|} y_{i,k} \one{ (a,b,t) \in  \path_{i, k}} - \sum_{j \in \riders} x_j \one{(\origin_j, \dest_j, \tr_j) = (a,b,t)}.$
Constraint \eqref{equ:lp_trip_capacity_constraint} guarantees that $\flow(\edge) \geq 0$. 
	\item For each driver $i \in \driverSet$ s.t. $\driverEntrance_i = 1$, let $\flow((\driverNode_i,~(\re_i, \te_i))) = \sum_{k = 0}^{|\pathSet_i|}y_{i,k}$.
	For each $i \in \driverSet$ s.t. $\driverEntrance_i = 0$, let $\flow((\driverNode_i,~(\re_i, \te_i))) = \sum_{k = 1}^{|\pathSet_i|}y_{i,k}$, %
and $\flow((\driverNode_i,~\sink)) = y_{i,0}$. %
	\item For each $\edge = ((a, \horizon), \sink) \in \edgeSet_3$, let  $\flow(\edge) = \sum_{\edge' \in \edgeSet, ~\head \edge' = (a, \horizon)} \flow(\edge')$ to balance the flow in and out of $(a, \horizon)$. This is the total number of drivers that existed the platform from $(a,t)$. 
\end{enumerate}

The edge capacity constraints \eqref{equ:flow_capacity_cnst_simp} and \eqref{equ:flow_nonneg_cnst_simp} are satisfied by construction. Given constraint \eqref{equ:lp_driver_constraint} %
and the fact that each $\path_{i, k}$ is a feasible path, %
constraints \eqref{equ:flow_balance_cnst_simp} and \eqref{equ:flow_driver_cnst_simp} are satisfied. Moreover, it is obvious that the objective of the two linear programs coincide, thus $\flow$ is a feasible solution to the flow LP \eqref{equ:lp} with the same objective.

\bigskip

\noindent{}\textit{Part (ii).}
Given a feasible, integral solution $\flow$ to the flow LP \eqref{equ:flow_LP_simp}, we construct an integral feasible solution to the original LP. %
For the riders, let $x_j = \flow(\riderEdge_j)$ for all $j \in \riders$. For the drivers, from the standard flow decomposition arguments~\cite{bertsimas1997introduction}, the $\nd$ units of flow in $\flow$ that all converge in $\sink$ can be decomposed into $\nd$ paths of single units of flow, that correspond to each driver's feasible path taken over the entire planning horizon. This gives us a feasible solution to the original LP, and it is easy to see that the objectives are the same. This completes the proof of the lemma. 
\end{proof}

The reduction to MCF can also be used to solve the original LP efficiently. In the optimal dispatching problem, the number of feasible paths for each driver is exponential in $|\loc|$ and $\horizon$, thus there are exponentially many decision variables in the LP \eqref{equ:lp}. The numbers of decision variables and constraints of the flow LP are, in contrast, polynomial in $|\riders|$, $|\loc|$ and $\horizon$, and there are efficient algorithms for solving network flow problems (see~\cite{ahuja1993network}).

\subsection{Proof of Lemma~\ref{thm:optimal_plans}} \label{appx:proof_thm_opt_plan}

Before proving the lemma, we first state the complementary slackness (CS) conditions~\cite{bertsimas1997introduction}. 
Given a feasible solution $(x,y)$ to the primal LP \eqref{equ:lp}, and a feasible solution $(\price,\pi,u)$ to the dual LP \eqref{equ:dual}, both solutions are optimal if and only if the following conditions hold:

\begin{enumerate}[(\CS-1)]
	\setlength\itemsep{0.0em}
	\item \label{item:cs1} for all $j \in \riders$, $x_j > 0 \Rightarrow u_j = \val_j - \price_{\origin_j, \dest_j, \tr_j}$,
	\item \label{item:cs2} for all $j \in \riders$, $u_j > 0 \Rightarrow x_j = 1$,
	\item \label{item:cs4} for all $i \in \driverSet$ and all $k = 1, \dots, |\pathSet_i|$, $y_{i,k} > 0 \Rightarrow \pi_i = \sum_{(a,b,t) \in \path_{i,k}} \price_{a,b,t} - \pathCost_{i,k}$,
	\item \label{item:cs5} for all $(a,b,t) \in \trips$, \vspace{-0.8em}
\begin{align*}
	\price_{a,b,t}> 0 \Rightarrow \sum_{j \in \riders} x_j  \one{ (\origin_j, \dest_j, \tr_j) = (a,b,t) } = \sum_{i \in \driverSet} \sum_{k = 0}^{|\pathSet_i|}  y_{i,k} \one{(a,b,t) \in \path_{i, k}}.
\end{align*}	
\end{enumerate}

We also provide this following lemma, showing that given any CE outcome, %
trips with excess driver supply have non-positive prices.

\begin{lemma}\label{lem:excessive_supply_zero_price} Given any plan with anonymous trip prices $(x, \actpath, \price)$ that forms a CE, for any $(a,b,t) \in \trips$, if  there exists a driver $ i \in \driverSet$ s.t. $(a,b,t) \in \actpath_i$, then $\price_{a,b,t} \leq 0$.
\end{lemma}

The proof is straightforward. If there exists any trip $(a,b,t)$ with a positive price, and a driver that takes this trip as relocation without a rider, the driver is not getting paid for this trip. This violates driver best response, since getting paid for this trip improves total payment to the driver.

\bigskip

We are now ready to prove Lemma~\ref{thm:optimal_plans}.

\thmOptPlans*

\begin{proof} Given any welfare optimal dispatching $(x, \actpath)$, we can construct an integral optimal solution $(x,y)$ to the LP \eqref{equ:lp}, where for all $i \in \driverSet$, $y_{i,k} = 1$ if the movement of driver $i$ in space and time according to the action path $\actpath_i$ is consistent with the path $\path_{i,k}$. %
Given any CE plan $(x, \actpath, \price)$, Lemma~\ref{lem:CE_plan_properties} implies that $(x, \actpath, \price^+)$ also forms a CE, where $\price^+$ is defined s.t. $\price_{a,b,t}^+ = \max\{\price_{a,b,t}, 0\}$.
We prove the lemma in two steps:
\begin{enumerate}[$\bullet$]
	\item Step 1. Given any optimal dispatching $(x, \actpath)$, and any optimal solution $(\price, \pi, u)$ to the dual LP \eqref{equ:dual}, the CS conditions imply that $\pi$ and $u$ can be interpreted as drivers' and rider's utilities, if the anonymous trip prices is given by $\price$. Optimal dual conditions guarantee driver and rider best responses, thus the plan $(x, \actpath, \price)$ forms a CE.
	\item Step 2. Given a CE plan $(x, \actpath, \price)$, let $(x,y)$ be the corresponding primal solution, and construct a dual solution $(\price^+, \pi, u)$, where $\pi$ and $u$ are the corresponding driver and rider utilities. CS conditions are satisfied between $(x,y)$ and $(\price^+, \pi, u)$, thus $(x, \actpath, \price)$ is welfare optimal.
\end{enumerate}

This proves the correspondence between CE and optimal plans and the existence of CE. 
Lemma~\ref{lem:dual_payment_correspondence} in Appendix~\ref{appx:proof_lemma_opt_pes_plans} implies that the CE prices can be efficiently computed from solving the dual of the flow LP. 
Regarding the properties: rider IR and envy-freeness is guaranteed by anonymous trip prices and CE; strict budget balance is guaranteed by the definition of anonymous trip prices; for driver envy-freeness, given any two drivers with the same initial (starting location/time, whether the driver had entered the platform or not), they have the same set of feasible paths, therefore both get the same highest total utility among those paths. 

We now prove the above two steps.

\bigskip

\noindent{}\emph{Step 1: Optimal primal and dual solutions $\Rightarrow$ CE.} 

Given an optimal dispatch $(x, \actpath, \price)$, let $(x,y)$ be the corresponding optimal integral solution to the primal LP \eqref{equ:lp}, and let $(\price, \pi, u)$ be any optimal solution to the dual LP \eqref{equ:dual}.  
We first show that if the anonymous trip prices are given by $\price$, then the dual variables $\pi$ and $u$ correspond to drivers' and riders' utilities, respectively:
\begin{enumerate}[1.]
	\setlength\itemsep{0.0em}
	\item $x_j > 0 \Rightarrow u_j = \val_j - \price_{\origin_j, \dest_j, \tr_j}$ from \cs{1}, thus for riders that are picked up, %
$u_j$ represent the utilities of the rider, which is her value minus the price for her trip.  %
	\item $u_j > 0 \Rightarrow x_j = 1$ from \cs{2}, i.e. in order for a rider to have positive utility, the rider must be picked up. This implies that $x_j = 0 \Rightarrow u_j = 0$, i.e. riders that are not picked up have zero utilities. Combining 1. and 2., we know that $u_j$ correspond to the rider's utilities.
	\item $y_{i,k} > 0 \Rightarrow \pi_i = \sum_{(a,b,t) \in \path_{i,k}} \price_{a,b,t} - \pathCost_{i,k}$ from \cs{4}, i.e. if driver $i$ takes her $k\th$ feasible path, then $\pi_i$ equals the sum of the prices of each trip covered by this path minus the total cost of this path. $\sum_{(a,b,t) \in \path_{i,k}} \price_{a,b,t} $ is equal to the driver's total payment since (I) for any rider trip, i.e. $(a,b,t) \in \path_{i,k}$ s.t. $\exists j \in \riders$ s.t. $(a,b,t,j) \in \actpath_i$, the driver is paid $\price_{a,b,t}$, and (II), for $(a,b,t)$ where the driver relocates without a rider, \cs{5} implies that $\price_{a,b,t} = 0$, therefore $\price_{a,b,t}$ is also the driver's payment. As a result, $\pi_i$ coincides with the total utility of driver $i$.
\end{enumerate}

We now show that this outcome forms a CE. For rider best response: constraint $u_j \geq 0$ guarantees IR for riders, thus riders that are picked-up can afford the price; $\val_j - \price_{\origin_j, \dest_j, \tr_j} > 0 \Rightarrow u_j > 0 \Rightarrow x_j = 1$ implies that all riders that strictly prefer getting pickup up are dispatched to some driver.
For driver best response, the dual constraints \eqref{equ:dual_cnst_driver} and \eqref{equ:dual_price_nonneg} guarantee that for all $i \in \driverSet$, $ \pi_i \geq \max_{k = 0, \dots, |\pathSet_i|} \left\lbrace \sum_{(a,b,t) \in \path_{i,k}} \price_{a,b,t} - \pathCost_{i,k} \right\rbrace = \max_{k = 0, \dots, |\pathSet_i|} \left\lbrace \sum_{(a,b,t) \in \path_{i,k}} \max \{ \price_{a,b,t}, ~0 \} - \pathCost_{i,k} \right\rbrace $.

\bigskip

\noindent{}\emph{Step 2: CE $\Rightarrow$ Optimal primal and dual solutions.}

Let $(x, \actpath, \price)$ be a CE plan with anonymous trip prices, and let $(u, \pi)$ be riders' and drivers' utilities under this plan. The plan being feasible implies that corresponding $(x, y)$ is a feasible and integral primal solution. %
We show that  $(\price^+, \pi, u)$ is a feasible solution to the dual LP \eqref{equ:dual}: \eqref{equ:dual_price_nonneg} holds by definition of $\price^+$; Lemma~\ref{lem:CE_plan_properties} and rider best response implies \eqref{equ:dual_cnst_rider} and \eqref{equ:dual_util_nonneg}; driver best response %
implies \eqref{equ:dual_cnst_driver}.

We now prove that $(x,y)$ and $(\price^+, \pi, u)$ must both be optimal, by checking the CS conditions:
\begin{enumerate}[1.]
	\setlength\itemsep{0.0em}
	\item For \cs{1}: given any $j \in \riders$ s.t. $x_j > 0$, we know that the rider is picked up, pays $\price_{\origin_j, \dest_j, \tr_j}$ and gets utility $u_j = \val_j - \price_{\origin_j, \dest_j, \tr_j}$.
	\item For \cs{2}: for any rider $ j \in \riders$ that gets utility $u_j > 0$, she must be picked up since otherwise her utility would be zero, therefore $x_j = 1$.
	\item For \cs{4}, for each driver $i \in \driverSet$, $y_{i,k} > 0$ implies that driver $i$ takes her $k\th$ feasible path. For each trip $(a,b,t)$ on the path, the driver gets paid $\price_{a,b,t}$ regardless of whether she picks up a rider (see Lemma~\ref{lem:excessive_supply_zero_price}), which implies that trips with excess driver supply have zero prices. 
Therefore, her total utility is the sum of the prices minus her cost, implying $\pi_i = \sum_{(a,b,t) \in \path_{i,k}} \price_{a,b,t} - \pathCost_{i,k}$. 
	\item For \cs{5}, $\price_{a,b,t} > 0 \Rightarrow \sum_{j \in \riders} x_j \one{(\origin_j, \dest_j, \tr_j) = (a,b,t)} = \sum_{i \in \driverSet} \sum_{k = 0}^{|\path_i|}  y_{i,k} \one{(a,b,t) \in \path_{i, k}} $ is implied by Lemma~\ref{lem:excessive_supply_zero_price}--- otherwise, there is excess supply for trip $(a,b,t)$, implying $\price_{a,b,t} = 0$. 
\end{enumerate}
This completes the proof of the lemma.
\end{proof}

\subsection{Proof of Lemma~\ref{thm:opt_pes_plans}} \label{appx:proof_lemma_opt_pes_plans}

\subsubsection{Dual of the Flow LP}

Before proving  Lemma~\ref{thm:opt_pes_plans}, we first discuss the  dual of the flow LP, and its correspondence to the dual LP \eqref{equ:dual}.
Let $\varphi_{a,t}$, $\varphi_{\driverNode_i}$, and $\mu_j$ be the dual variables corresponding to constraints \eqref{equ:flow_balance_cnst_simp}, \eqref{equ:flow_driver_cnst_simp} and \eqref{equ:flow_capacity_cnst_simp}, respectively. The dual LP of \eqref{equ:flow_LP_simp} can be written as:
\begin{align}
	\min ~ &  \sum_{i \in \driverSet} \varphi_{\driverNode_i} + \sum_{j \in \riders}  \mu_j  \label{equ:flow_dual} \\
	\txtst & \varphi_{\origin_j, \tr_j} - \varphi_{\dest_j, \tr_j + \dist(\origin_j, \dest_j)}  + \mu_j \geq \val_j - \cost_{\origin_j, \dest_j, \tr_j} , & \forall j \in \riders \label{equ:flow_dual_rider_cnst}  \\ 
	 	& \varphi_{a,t} - \varphi_{b, t + \dist(a,b)} \geq -\cost_{a,b,t},  & \forall (a,b,t) \in \trips \label{equ:flow_dual_nonneg_1} \\
	 	& \varphi_{a, t} \geq - \exitCost_{\horizon - t},  & \forall a \in \loc, ~\forall t \in \timeSet  \label{equ:flow_dual_nonneg_2} \\
		& \varphi_{\driverNode_i} \geq \varphi_{\re_i, \te_i}, & \forall i \in \driverSet  \label{equ:flow_dual_driver_cnst_1}\\
		& \varphi_{\driverNode_i} \geq 0, & \forall i \in \driverSet \txtst \driverEntrance_i = 0, \label{equ:flow_dual_driver_cnst_2} \\	 	
	 	& \mu_j \geq 0, & \forall j \in \riders \label{equ:flow_dual_nonneg_3}
\end{align}

Given a solution $(\varphi, \mu)$ of \eqref{equ:flow_dual}, the $\varphi$ variables corresponding to the flow balance constraints are usually referred to as the \emph{potential} of the nodes, and we call $\varphi$ an \emph{optimal potential} of the MCF problem if there exist $\mu \in \setR^{|\riders|}$ s.t. $(\varphi, \mu)$ is an optimal solution of \eqref{equ:flow_dual}. 
The potential for each node can be interpreted as how ``useful" it is to have an additional unit of flow originating from this node, and $\mu_j$ for each $j$ can be interpreted as the utility of the rider $j$.

\paragraph{Complementary Slackness Conditions} Given a feasible solution $\flow$ to the flow primal LP~\eqref{equ:flow_LP_simp} and a feasible solution $(\varphi, \mu)$ to the flow dual LP~\eqref{equ:flow_dual}, both solutions are optimal if and only if the following complementary slackness conditions~\cite{bertsekas1990auction} are satisfied.

\begin{enumerate}[(\FCS-1)]
	\setlength\itemsep{0.0em}
	\item \label{item:fcs1} for all $j \in \riders$, $\flow(\riderEdge_j) > 0 \Rightarrow \varphi_{\origin_j, \tr_j} - \varphi_{\dest_j, \tr_j + \dist(\origin_j, \dest_j)}  + \mu_j = \val_j - \cost_{\origin_j,\dest_j,\tr_j}$,
	\item \label{item:fcs2} for all $j \in \riders$, $\mu_j > 0 \Rightarrow \flow(\riderEdge_j) = 1$,
	\item \label{item:fcs3} for all $(a,b,t) \in \trips$, $\flow(((a,t),(b,t+\dist(a,b))) > 0 \Rightarrow \varphi_{a,t} - \varphi_{b,t+\dist(a,b)} = -\cost_{a,b,t}$,
	\item \label{item:fcs4} for all $a \in \loc$ and $t \in \timeSet$, $\flow(((a, t), \sink)) > 0 \Rightarrow \varphi_{a, t} = -\exitCost_{\horizon - t}$.
	\item \label{item:fcs5} for all $i \in \driverSet$, $\flow((\driverNode_i,~(\re_i,\te_i))) > 0 \Rightarrow \varphi_{\driverNode_i} = \varphi_{\re_i,\te_i}$.
	\item \label{item:fcs6} for all $i \in \driverSet$ s.t. $\driverEntrance_i = 0$, $\flow((\driverNode_i,~\sink)) > 0 \Rightarrow \varphi_{\driverNode_i} = 0$.
\end{enumerate}

The following lemma establishes a one-to-one correspondence between the $\pi_i$ variables in optimal solutions to the dual LP \eqref{equ:dual}, and the $\varphi_{\driverNode_i}$ variables in optimal solutions to \eqref{equ:flow_dual}.

\begin{lemma} \label{lem:dual_payment_correspondence}  %
For any optimal solution $(\price, \pi, u)$ to the dual LP \eqref{equ:dual}, there exists an optimal solution $(\varphi, \mu)$ to the dual of the flow LP \eqref{equ:flow_dual} such that $\varphi_{\driverNode_i} = \pi_i$ for all $i \in \driverSet$, $u_j = \mu_j$ for all $j \in \riders$, and vise versa.
\end{lemma}

\begin{proof} 

We prove the following two directions by construction:

\begin{enumerate}[(i)]
	\item Given an optimal solution $(\price, \pi, u)$ to  \eqref{equ:dual}, there exists an optimal solution $(\varphi, \mu)$ to \eqref{equ:flow_dual} s.t. $\varphi_{\driverNode_i} = \pi_i$ for all $i \in \driverSet$ and that $\mu_j = u_j$ for all $ j \in \riders$.
	\item Given an optimal solution $(\varphi, \mu)$ to \eqref{equ:flow_dual}, there exists an optimal solution $(\price, \pi, u)$ to \eqref{equ:dual} s.t. $\pi_i = \varphi_{\driverNode_i} $ for all $i \in \driverSet$ and that $u_j = \mu_j$ for all $ j \in \riders$.	
\end{enumerate}

\noindent{}\textit{Part (i).}
Given any optimal solution $(\price, \pi, u)$ to \eqref{equ:dual}, we construct a solution $(\varphi, \mu)$ to \eqref{equ:flow_dual} from the prices $\price$ as follows, where $\varphi_{a,t}$ represents the highest continuation payoff for any driver from location $a$ and time $t$ onward, $\varphi_{\driverNode_i}$ represents the highest achievable payoff of driver $i$, and $\mu_j$ represents the highest achievable utility of rider $j$:
\begin{enumerate}[$\bullet$]
	\setlength\itemsep{0.0em}
	\item For all $a \in \loc$, let $\varphi_{a, \horizon} = - \exitCost_{\horizon - T} = 0$.
	\item For all $a \in \loc$ and all $t = \horizon-1, \horizon-2, \dots, 0$, let 
	\begin{align}
		\varphi_{a,t} =  \max\left\lbrace  \max_{b \in \loc \txtst t + \dist(a,b) \leq \horizon} \left\lbrace  \varphi_{b, t + \dist(a,b)} + \price_{a,b,t} - \cost_{a,b,t} \right\rbrace, ~ - \exitCost_{T - t} \right\rbrace. \label{equ:potential_construction}
	\end{align}
	\item For all $i \in \driverSet$, s.t. $\driverEntrance_i = 1$, let $\varphi_{\driverNode_i} = \varphi_{\re_i, \te_i}$; for $i \in \driverSet$ s.t. $\driverEntrance_i = 0$, let $\varphi_{\driverNode_i}= \max\{ \varphi_{\re_i, \te_i},~0\}$.
	\item For all $ j \in \riders$, let $\mu_j  = \max \{ \val_j - \price_{\origin_j, \dest_j, \tr_j}', ~0 \}$, where 
\begin{align}
	\price_{a,b,t}' \triangleq \varphi_{a,t} - \varphi_{b, t + \dist(a,b)} + \cost_{a,b,t}, ~\forall (a,b,t) \in \trips. \label{equ:price_construction}
\end{align}	
\end{enumerate}

Note that $\price_{a,b,t}' \geq \price_{a,b,t}$ holds for all $(a,b,t) \in \trips$, since $\varphi_{a,t} \geq \varphi_{b, t+ \dist(a,b)} + \price_{a,b,t} - \cost_{a,b,t}$ from \eqref{equ:potential_construction}. 
Moreover, we claim that for any optimal solution $(x,y)$ to the LP \eqref{equ:lp},
\begin{align}
	\price_{a,b,t} = \price'_{a,b,t} \text{ for all } (a,b,t) \in \trips \txtst \sum_{i \in \driverSet}  \sum_{k = 1}^{|\pathSet_i|}y_{i,k} \one{(a,b,t) \in \path_{i,k} } > 0, \label{equ:alt_price_identity}
\end{align}
i.e. the prices must coincide for any trips that is taken by at least one driver.

To prove \eqref{equ:alt_price_identity}, first observe that $\forall (a,t) \in \loc \mytimes \timeSet$, $\varphi_{a,t}$ as in \eqref{equ:potential_construction} is equal to the highest total utility among all possible paths starting from $(a,t)$ to the end of time, given the prices $\price$--- this is obvious for $t = \horizon$ (since $\varphi_{a,\horizon} = 0$ for al $a$) and also for $t < \horizon$ by induction.
Now consider any trip $(a,b,t)$ taken by some driver, say driver $1$. 
Since the outcome forms a CE (since $(x,y)$ and $(\price, \pi, u)$ are optimal primal and dual solutions), we know that the total utility to driver $1$ from time $t$ onward must be $\varphi_{a,t}$, the highest total utility among all possible paths starting from $(a,t)$. Similarly, the total utility to driver $1$ from location $b$ and time $t + \dist(a,b)$ onward is $\varphi_{b,{t + \dist(a,b)}}$. Since the $(a,b,t)$ trip pays the driver $\price_{a,b,t}$ and costs $\cost_{a,b,t}$, we know that $\varphi_{a,t} = \price_{a,b,t} + \varphi_{b,t+\dist(a,b)} - \cost_{a,b,t}$ must hold, which gives us $\price_{a,b,t}' = \varphi_{a,t} - \varphi_{b,t+\dist(a,b)} = \price_{a,b,t}$.

\smallskip

We now show that $(\varphi, \mu)$ forms an optimal solution to \eqref{equ:flow_dual}. 
Given the non-negativity of $price$, constraints \eqref{equ:flow_dual_rider_cnst} to \eqref{equ:flow_dual_nonneg_3} are satisfied by construction, thus what is left to prove is optimality. Let $(x,y)$ be some optimal integral solution to \eqref{equ:lp}. We know that $(x,y)$ and $(\price, \pi, u)$ satisfy the CS conditions \cs{1}-\cs{5}, and form a CE. We construct an optimal integral solution $\flow$ to \eqref{equ:flow_LP_simp} from $(x,y)$ in the same way as in the proof of Lemma~\ref{thm:lp_integrality}, and it is sufficient for the optimality to prove that \fcs{1}-\fcs{6} hold between $\flow$ and $(\varphi, \mu)$:

\begin{enumerate}[1.]
	\setlength\itemsep{0.0em}
	\item %
	To show \fcs{1}, first observe that for all $j \in \riders$, $\flow(\edge_j) > 0 \Rightarrow x_j > 0$ implies that rider $j$ is picked up, thus the trip $(\origin_j, \dest_j, \tr_j)$ is taken by some driver, thus $\price_{\origin_j, \dest_j, \tr_j} = \price_{\origin_j, \dest_j, \tr_j} '$ by \eqref{equ:alt_price_identity}. Moreover, $u_j = \val_j - \price_{\origin_j, \dest_j, \tr_j}$ from \cs{1}, implying $\val_j - \price_{\origin_j, \dest_j, \tr_j}' \geq 0$.
	 This gives us: $\flow(\edge_j) > 0 %
	 \Rightarrow \mu_j = \val_j - \price_{\origin_j, \dest_j, \tr_j}' \Rightarrow   \varphi_{\origin_j, \tr_j} - \varphi_{\dest_j, \tr_j + \dist(\origin_j, \dest_j)}  + \mu_j= \val_j - \cost_{\origin_j, \dest_j,\tr_j}$.  
	\item To show \fcs{2}, %
	recall that $\price'_{a,b,t} \geq \price_{a,b,t}$ for all $(a,b,t) \in \trips$. Therefore, given \eqref{equ:dual_cnst_rider} and \cs{2}, we know that for all $j \in \riders$, $\mu_j > 0 \Rightarrow \val_j - \price_{\origin_j, \dest_j, \tr_j}' > 0 \Rightarrow \val_j - \price_{\origin_j, \dest_j, \tr_j} > 0  \Rightarrow u_j > 0 \Rightarrow x_j = 1 \Rightarrow \flow(\riderEdge_j) = 1$.
	\item \fcs{3} %
	holds since $\flow(\edge) > 0$ only when there is excess supply in the dispatching $(x,y)$ for the trip $(a,b,t)$, therefore $\price_{a,b,t} = 0$ given \cs{5}. \eqref{equ:alt_price_identity} then implies that $ \price_{a,b,t}' = 0$, thus 
	$\varphi_{a,t} - \varphi_{b, t + \dist(a,b)} = \price_{a,b,t}'  - \cost_{a,b,t} = - \cost_{a,b,t}$.
	\item \fcs{4} holds, %
	since $\flow(((a,t), \sink)) > 0$ implies that given dispatch $(x,y)$, there exists at least one driver that exited the platform from $(a,t)$, therefore gets utility $- \exitCost_{T - t}$ from time $t$ onward. As a result, the highest utility for all paths from $(a,t)$ onward must be $\varphi_{a,t} =- \exitCost_{T - t}$ without violating CE. 
	\item \fcs{5} holds by construction for $i$ s.t. $\driverEntrance_i = 1$. For $i$ s.t. $\driverEntrance_0 = 0$, since when $\flow((\driverNode_i, ~ (\re_i, \te_i))) > 0$, driver $i$ entered the platform (instead of not entering and getting zero utility), thus her utility $\varphi_{\re_i, \te_i}$ from $(\re_i, \te_i)$ onward must not be negative. and as a result, $\varphi_{\driverNode_i} = \max\{\varphi_{\re_i, \te_i},~0\} = \varphi_{\re_i, \te_i}$. 	
	\item \fcs{6} holds, since when $\flow((\driverNode_i,~\sink)) > 0$, driver $i$ did not enter the platform at all according to the dispatch $(x,y)$. As a result, CE implies that entering must not give her positive utility, therefore $\varphi_{\re_i, \te_i} \leq 0$, and her utility $\varphi_{\driverNode_i} = \max \{ \varphi_{\re_i, \te_i}, 0 \} = 0$.
\end{enumerate}

What is left to show is that $\varphi_{\driverNode_i} = \pi_i$ for all $i \in \driverSet$ and that $\mu_j = u_j$ for all $j \in \driverSet$. $\varphi_{\driverNode_i} = \pi_i$ holds, since $\varphi_{\driverNode_i}$ as constructed is the highest achievable utility for driver $i$ among all feasible paths, and $\pi_i$ must take this value given CE. For $j \in \riders$ s.t. $x_j= 0$, we know that $\flow(\riderEdge_j) = 0$, and $u_j = \mu_j = 0$ must hold. For $j \in \riders$ s.t. $x_j = 1$, $\price_{\origin_j, \dest_j, \tr_j} = \price_{\origin_j, \dest_j, \tr_j}'$ holds given \eqref{equ:alt_price_identity}, therefore $\flow(\riderEdge_j) \Rightarrow \mu_j = \val_j - \price_{\origin_j, \dest_j, \tr_j}' = \val_j - \price_{\origin_j, \dest_j, \tr_j} = u_j$.

\bigskip

\noindent{}\textit{Part (ii).}  
Let $(\varphi, \mu)$ be an optimal solution to \eqref{equ:flow_dual}.
We now construct a solution $(\price, \pi, u)$ to \eqref{equ:dual}, where the price $\price_{a,b,t}$ is the loss of potential between the origin node $(a,t)$ and the destination node $(b,t+\dist(a,b))$ plus the trip cost $\cost_{a,b,t}$, and the driver and rider utilities are given by $\varphi_{\driverNode_i}$ and $\mu_i$:
\begin{align*}
	u_j & = \mu_j, & \forall j \in \riders  \\
	\pi_i & = \varphi_{\driverNode_i}, & \forall i \in \driverSet \\
	\price_{a,b,t} &= \varphi_{a,t} - \varphi_{b, t + \dist(a,b)} + \cost_{a,b,t}, & \forall (a,b,t) \in \trips
\end{align*}
We first show that $(\price, \pi, u)$ is a feasible solution to \eqref{equ:dual}.
\begin{enumerate}[1.]
	\setlength\itemsep{0.0em}
	\item From telescoping sum, for any feasible path $\path_{i,k}$ of driver $i$, starting at $(\re_i, \te_i)$ and ending at $(a', t)$ for some $a' \in \loc$ and some $t' \in \timeSet$, the total utility from taking the path is  $\sum_{(a,b,t) \in \path_{i,k}} \left( \price_{a,b,t} - \cost_{a,b,t} \right)  - \exitCost_{\horizon - t'} = \varphi_{\re_i, \te_i} - \varphi_{a', t'} - \exitCost_{\horizon - t'} $, which is at most $  \varphi_{\re_i, \te_i} $ (since $\varphi_{a', t'} \geq - \exitCost_{\horizon - t'}$ for all $a' \in \loc$ and $t' \in \timeSet$, guaranteed by \eqref{equ:flow_dual_nonneg_2}). This implies that the utility  $\pi_i = \varphi_{\driverNode_i} \geq \varphi_{\re_i, \te_i} \geq  \sum_{(a,b,t) \in \path_{i,k}} \left( \price_{a,b,t} - \cost_{a,b,t} \right)  - \exitCost_{\horizon - t'}  = \sum_{(a,b,t) \in \path_{i,k}}  \price_{a,b,t} - \pathCost_{i,k} $ for any $k \in \{ 1, \dots, |\pathSet_i| \}$, therefore \eqref{equ:dual_cnst_driver} holds.
	\item  \eqref{equ:flow_dual_rider_cnst} implies $u_j = \mu_j  \geq \val_j - (\varphi_{\origin_j, \tr_j} - \varphi_{\dest_j, \tr_j + \dist(\origin_j, \dest_j)} + \cost_{\origin_j,\dest_j,\tr_j}) = \val_j - \price_{\origin_j, \dest_j, \tr_j}$ thus \eqref{equ:dual_cnst_rider} holds.
	\item \eqref{equ:flow_dual_nonneg_1} implies that $\price_{a,b,t} \geq 0$ thus \eqref{equ:dual_price_nonneg} holds.
	\item Lastly, \eqref{equ:flow_dual_nonneg_3} implies $u_j = \mu_j \geq 0$, which is \eqref{equ:dual_util_nonneg}.
\end{enumerate}

Therefore, $(\price, \pi, u)$ is a feasible solution to \eqref{equ:dual}. 
Regarding the optimality of $(\price, \pi, u)$, we know by construction that the objective of \eqref{equ:dual} is equal to that of \eqref{equ:flow_dual}. 
Recall the correspondence of optimal solutions that we established in Appendix~\ref{appx:proof_thm_lp_integrality}, that the optimal objective of the flow LP \eqref{equ:flow_LP_simp} is equal to that of the original LP \eqref{equ:lp}. This implies that the optimal objective of the dual \eqref{equ:dual} is equal to the optimal objective of the LP \eqref{equ:lp}, therefore $(\price, \pi, u)$ is an optimal solution such that $\pi_i = \varphi_{\driverNode_i}$ for all $i \in \driverSet$ and that $u_j = \mu_j$ for all $ j \in \riders$.

This completes the proof of the lemma.
\end{proof}

\subsubsection{Proof of Lemma~\ref{thm:opt_pes_plans}}

\lemOptPesPlans*

\begin{proof}

Step~2 of the proof of Lemma~\ref{thm:optimal_plans} established that the set of possible driver utilities among all CE outcomes correspond to the $\pi$ variables among the set of optimal solutions $(\price, \pi, u)$ to the dual LP \eqref{equ:dual}. Since Lemma~\ref{lem:dual_payment_correspondence} established the correspondence between the $\pi$ variables and the $\varphi_{\driverNode_i}$ variables in optimal solutions to \eqref{equ:dual} and \eqref{equ:flow_dual}, what we need to show is the lattice structure of $\varphi$ in optimal solutions of \eqref{equ:flow_dual}, and that $\Phi$ and $\Psi$ reside on the bottom and the top of the lattice. %

\medskip

\noindent{}{\emph{Step 1. Proof of the Lattice Structure}}   

We first prove the lattice structure. Let $(\varphi, \mu)$ and $(\varphi', \mu')$ be two optimal solutions of \eqref{equ:flow_dual}. We prove that the join and the meet of $\varphi$ and $\varphi'$ are both optimal potentials. Let the join and the meet be defined as:
For all $(a,t) \in \loc \mytimes \timeSet$, let the join and the meet of the potentials be
\begin{align*}
	\barphi_{a,t} &\triangleq \max\left\lbrace \varphi_{a,t}, ~ \varphi_{a,t}'\right\rbrace,~\forall (a,t) \in \loc \mytimes \timeSet, ~ \barphi_{\driverNode_i} \triangleq \max\{\varphi_{\driverNode_i}, ~\varphi_{\driverNode_i}' \},~\forall i \in \driverSet, \\
	\undphi_{a,t} &\triangleq \min\left\lbrace \varphi_{a,t}, ~ \varphi_{a,t}'\right\rbrace, ~\forall (a,t) \in \loc \mytimes \timeSet, ~\undphi_{\driverNode_i} \triangleq \min \{\varphi_{\driverNode_i}, ~\varphi_{\driverNode_i}' \},~\forall i \in \driverSet, 
\end{align*}
For convenience of notation, for all $(a,b,t) \in \trips$, denote
\begin{align*}
	\price_{a,b,t} & \triangleq \varphi_{a,t} - \varphi_{b, t +\dist(a,b)} + \cost_{a,b,t}, \\
	\price_{a,b,t}' &\triangleq \varphi_{a,t}' - \varphi_{b, t +\dist(a,b)}' + \cost_{a,b,t},
\end{align*}
and let $\barp$ and $\undp$ be the prices constructed from the join and the meet of the potentials:
\begin{align*}
	\barp_{a,b,t} &\triangleq \barphi_{a,t} - \barphi_{b, t +\dist(a,b)} + \cost_{a,b,t}, \\
	\undp_{a,b,t} &\triangleq \undphi_{a,t} - \undphi_{b, t +\dist(a,b)} + \cost_{a,b,t}.
\end{align*}
Finally, for all $j \in \riders$, let
\begin{align*}
	\barmu_j &\triangleq \max\{ \val_j - \barp_{\origin_j, \dest_j, \tr_j}, ~0 \} , \\
	\undmu_j &\triangleq \max\{ \val_j - \undp_{\origin_j, \dest_j, \tr_j}, ~0 \}.
\end{align*}

We first prove that both $(\barphi, \barmu)$ and $(\undphi, \undmu)$ are feasible solutions to \eqref{equ:flow_dual}. Constraints \eqref{equ:flow_dual_rider_cnst},  \eqref{equ:flow_dual_nonneg_2}, \eqref{equ:flow_dual_driver_cnst_2} and \eqref{equ:flow_dual_nonneg_3} hold by construction. For constraint \eqref{equ:flow_dual_nonneg_1}, we first show that for all $(a,b,t) \in \trips$,
\begin{align}
	\barp_{a,b,t} & \in [\min\{\price_{a,b,t}, ~ \price_{a,b,t}' \}, ~\max\{\price_{a,b,t}, ~ \price_{a,b,t}' \} ], \label{equ:price_ineq_bar} \\
	\undp_{a,b,t} & \in [\min\{\price_{a,b,t}, ~ \price_{a,b,t}' \}, ~\max\{\price_{a,b,t}, ~ \price_{a,b,t}' \} ]. \label{equ:price_ineq_und}
\end{align}
We only prove $\barp_{a,b,t} \geq \min\{\price_{a,b,t}, ~ \price_{a,b,t}' \}$ here. The proof for the other three inequalities are very similar. Assume w.l.o.g. that $\varphi_{a,t} \geq \varphi_{a,t}'$. This implies $\barphi_{a,t} = \varphi_{a,t}$. 
Consider two scenarios: (I) $\varphi_{b,t+\dist(a,b)} \geq \varphi'_{b, t+\dist(a,b)}$ and (II) $\varphi_{b,t+\dist(a,b)} < \varphi'_{b, t+\dist(a,b)}$. 
For (I), $\barphi_{b, t+\dist(a,b)}  = \varphi_{b,t+\dist(a,b)}$ thus $\barp_{a,b,t} =  \barphi_{a,t} - \barphi_{b, t +\dist(a,b)} + \cost_{a,b,t} = \varphi_{a,t} - \varphi_{b, t +\dist(a,b)} + \cost_{a,b,t} = \price_{a,b,t} \geq \min\{\price_{a,b,t}, ~ \price_{a,b,t}' \}$. 
For (II), we know that $\barphi_{b,t+\dist(a,b)}  = \max \{\varphi_{b, t+\dist(a,b)}, ~\varphi'_{b, t+\dist(a,b)} \} = \varphi'_{b, t+\dist(a,b)}$, thus $\barp_{a,b,t} =  \barphi_{a,t} - \barphi_{b, t +\dist(a,b)} + \cost_{a,b,t} = \varphi_{a,t} - \varphi'_{b, t +\dist(a,b)} + \cost_{a,b,t} \geq \varphi'_{a,t} - \varphi'_{b, t +\dist(a,b)} + \cost_{a,b,t} = \price'_{a,b,t} \geq \min\{\price_{a,b,t}, ~ \price_{a,b,t}' \}$. %

$\varphi$ and $\varphi'$ satisfying \eqref{equ:flow_dual_nonneg_1} implies $\price_{a,b,t} \geq 0$ and $\price_{a,b,t}' \geq 0$, which means that $\min \{ \price_{a,b,t},~\price_{a,b,t}'\} \geq 0$ for all $(a,b,t) \in \trips$. Therefore $\barp_{a,b,t} \geq 0$ and $\undp_{a,b,t} \geq 0$ both hold. This proves $\barphi$ and $\undphi$ satisfy \eqref{equ:flow_dual_nonneg_1}. 
Constraint \eqref{equ:flow_dual_driver_cnst_1} also holds, since for all $i \in \driverSet$, $\barphi_{\driverNode_i} = \max\{ \varphi_{\driverNode_i}, \varphi_{\driverNode_i}' \} \geq \max\{ \varphi_{\re_i,\te_i}, \varphi_{\re_i, \te_i}'\} = \barphi_{\re_i, \te_i}'$, and $\undphi_{\driverNode_i} =  \min\{  \varphi_{\driverNode_i}, \varphi_{\driverNode_i}' \} \geq  \min\{ \varphi_{\re_i,\te_i}, \varphi_{\re_i, \te_i}'\}  = \undphi_{\re_i, \te_i}' $. Thus $(\barphi, \barmu)$ and $(\undphi, \undmu)$ are both feasible.

\medskip

Let $\flow$ be an integral optimal solution to the flow LP \eqref{equ:flow_LP_simp}.
We prove that $(\barphi, \barmu)$ and $(\undphi, \undmu)$ are both optimal solutions to \eqref{equ:flow_dual} by showing that the CS conditions \fcs{1}-\fcs{6} hold between $\flow$ and $(\barphi,  \barmu)$, and also between $\flow$ and $(\undphi, \undmu)$. First note that \fcs{1}-\fcs{6} hold in between $\flow$ and $(\varphi, \mu)$ and between $\flow$ and $(\varphi', \mu')$.

\begin{enumerate}[1.]
	\setlength\itemsep{0.0em}
	\item To show \fcs{1}, %
	note that \fcs{1} between $\flow$ and $(\varphi, \mu)$, $(\varphi',\mu')$ imply that if $\flow(\riderEdge_j) > 0$, $\mu_j = \val_j - \price_{\origin_j, \dest_j, \tr_j} \geq 0$ and $\mu'_j = \val_j - \price_{\origin_j, \dest_j, \tr_j}' \geq 0$. Applying \eqref{equ:price_ineq_bar} and \eqref{equ:price_ineq_und}, we get $\val_j - \barp_{\origin_j,\dest_j,\tr_j} \geq \val_j - \max\{\price_{\origin_j,\dest_j,\tr_j} ,~ \price_{\origin_j,\dest_j,\tr_j} '\} = \min\{\val_j - \price_{\origin_j, \dest_j, \tr_j},~ \val_j - \price_{\origin_j, \dest_j, \tr_j}' \} \geq 0$ and $ \val_j - \undp_{\origin_j,\dest_j,\tr_j}  \geq \val_j - \max\{\price_{\origin_j,\dest_j,\tr_j} ,~ \price_{\origin_j,\dest_j,\tr_j} '\} = \min\{\val_j - \price_{\origin_j, \dest_j, \tr_j},~ \val_j - \price_{\origin_j, \dest_j, \tr_j}' \} \geq 0$. The definitions of $\barmu_j$ and $\undmu_j$ then imply $\barmu_j = \val_j - \barp_{\origin_j,\dest_j,\tr_j}$ and $\undmu_j = \val_j - \undp_{\origin_j,\dest_j,\tr_j}$ hold.
	\item To show \fcs{2}, observe that for all $j \in \riders$, $\barmu_j > 0 \Rightarrow \val_j - \barp_{\origin_j, \dest_j, \tr_j} > 0  \Rightarrow \val_j - \min\{\price_{\origin_j, \dest_j, \tr_j},~\price_{\origin_j, \dest_j, \tr_j}' \} > 0 \Rightarrow \max\{  \val_j - \price_{\origin_j, \dest_j, \tr_j}, ~\val_j - \price'_{\origin_j, \dest_j, \tr_j} \} > 0 \Rightarrow  \max \{\mu_j, \mu_j'\} > 0 \Rightarrow \flow(\riderEdge_j) = 1$. Similarly, $\undmu_j > 0 \Rightarrow \flow(\riderEdge_j) = 1$. 
	\item We now consider \fcs{3}. For any $\edge = ((a,t),(b, t+ \dist(a,b))) \in \edgeSet_2$, $\flow(\edge) > 0 \Rightarrow \varphi_{a,t} - \varphi_{b, t + \dist(a,b)} = \varphi_{a,t}' - \varphi_{b, t + \dist(a,b)}' = -\cost_{a,b,t}$. This implies that $\price_{a,b,t} =\price_{a,b,t}' = 0$, and as a result, $ 0 \leq \barp_{a,b,t},~ \undp_{a,b,t} \leq 0$, therefore $\barphi_{a,t} - \barphi_{b, t + \dist(a,b)} = -\cost_{a,b,t}$ and $\undphi_{a,t} - \undphi_{b, t + \dist(a,b)} = -\cost_{a,b,t}$ both hold. 
	\item \fcs{4} %
	holds since for each $ \edge = ((a, t), \sink) \in \edgeSet$ for some $a\in\loc$ and $t\in \timeSet$, $\flow(\edge) > 0 \Rightarrow \varphi_{a, \horizon} = -\exitCost_{\horizon - t}$ and $\varphi_{a, \horizon}' = -\exitCost_{\horizon - t}$. Thus, $\barphi_{a, \horizon} = \undphi_{a, \horizon} = -\exitCost_{\horizon - t}$. 
	\item For \fcs{5}: $\flow((\driverNode_i,(\re_i,\te_i))) > 0$ implies $\varphi_{\driverNode_i} = \varphi_{\re_i,\te_i}$ and $\varphi_{\driverNode_i}' = \varphi_{\re_i,\te_i}'$, thus $\barphi_{\driverNode_i} = \max\{\varphi_{\driverNode_i},\varphi_{\driverNode_i}' \} = \max\{\varphi_{\re_i,\te_i},\varphi_{\re_i,\te_i}'  \} = \barphi_{\re_i,\te_i}$ and $\undphi_{\driverNode_i} = \min\{ \varphi_{\driverNode_i},\varphi_{\driverNode_i}'\} = \min \{\varphi_{\re_i,\te_i},\varphi_{\re_i,\te_i}' \} = \undphi_{\re_i,\te_i}$.
	\item For \fcs{6}, for $i \in \driverSet$ s.t. $\driverEntrance_i = 0$ and $\flow((\driverNode_i,\sink)) > 0$, $\varphi_{\driverNode_i} = \varphi_{\driverNode_i}' = 0$, thus $\barphi_{\driverNode_i} = \undphi_{\driverNode_i} = 0$.
\end{enumerate}

This completes the proof of the lattice structure of drivers' total utilities. 

\bigskip

\noindent{}\emph{Step 2. Driver Optimal and Pessimal Plans.}

We now prove the correspondence between the welfare changes and the top and bottom of the lattice.
Recall that $\Phi_{\driverNode_i}$ and $\Psi_{\driverNode_i}$ are the welfare gain/loss from replicating/losing driver $i$, respectively, and  $\Phi_{a,t}$  is the welfare gain from adding another driver (that has entered) to location $a$ and time $t$.
Here we define the welfare loss from losing one driver from location $a$ at time $t$ as:
\begin{align}
	\Psi_{a,t} \triangleq \sw(\driverSet, \riders)- \sw(\driverSet \backslash \{(1, a, t, \horizon)\}, \riders), \label{equ:welfare_loss_general_appx}
\end{align}
where $\sw(\driverSet \backslash \{(1, a, t, \horizon)\})$ is the highest achievable social welfare, if one of the drivers in $\driverSet$ who was supposed to exit the platform at time $\horizon$ now needs to exit the platform at location $a$ at time $t$. Note that this does not specify which particular driver exits, and can be considered as the objective of the flow LP where we simply subtract 1 from the boundary condition $\xi_{a,t}$ at the node $(a,t)$. %

We first show via standard arguments with the residual graph that $\Phi$ and $\Psi$ as we defined in \eqref{equ:welfare_gain} and \eqref{equ:welfare_loss_general_appx} are optimal potentials for the flow LP. We then show via subgradient arguments that $\Phi$ and $\Psi$ are the bottom and the top of the lattice of the potentials, respectively. %
Given Lemma~\ref{lem:dual_payment_correspondence}, and the fact that driver payments among CE outcomes correspond to the optimal solutions of the dual LP \eqref{equ:dual}, we know $\Phi_{\driverNode_i}$ and $\Psi_{\driverNode_i}$ correspond to the bottom and the top of the lattice of driver's total payments among all CE outcomes, hence Lemma~\ref{thm:opt_pes_plans}.

\medskip

\noindent{}\textit{Step 2.1. $\Phi$ and $\Psi$ are Optimal Potentials:}

Given the MCF problem \eqref{equ:flow_LP} with graph $\graph = (\nodeSet, \edgeSet)$ and an optimal integral solution $\flow$ (which is guaranteed to exist), we first construct the standard residual graph $\resGraph = (\nodeSet, \resEdgeSet)$ where the set of nodes remains the same, and the set of edges $\resEdgeSet = \resEdgeSet_1 \cup \resEdgeSet_2 \cup \resEdgeSet_3 \cup \resEdgeSet_4$ consists of:
\begin{enumerate}[$\bullet$]
	\item $\resEdgeSet_1 = \set{\riderEdge_j}{j \in \riders,~\flow(\riderEdge_j) = 0} \cup \set{\resRiderEdge_j}{j \in \riders,~\flow(\riderEdge_j) = 1}$, where $\riderEdge_j = ((\origin_j, \tr_j),~(\dest_j, \tr_j + \dist(\origin_j, \dest_j)))$ is the edge corresponding to rider $j$ with $\edgeCost(\riderEdge_j) = -\val_j + \cost_{\origin_j,\dest_j,\tr_j}$, $\minCap(\riderEdge_j) = 0$, and $\maxCap(\resRiderEdge_j)= 1$; $\resEdge_j = ((\dest_j, \tr_j + \dist(\origin_j, \dest_j),~(\origin_j, \tr_j)))$ is the reversed edge corresponding to rider $j$ s.t. $\flow(\riderEdge_j) = 1$, with $\edgeCost(\resRiderEdge_j) = \val_j- \cost_{\origin_j, \dest_j,\tr_j}$  , $\minCap(\resRiderEdge_j) = 0$ and $\maxCap(\resRiderEdge_j) = 1$.
	\item $\resEdgeSet_2 = \edgeSet_2 \cup \set{\resEdge}{\edge \in \edgeSet_2,~\flow(\edge) > 0}$, where for each $\edge = ((a,t),~(b,t+\dist(a,b))) \in \edgeSet_2$ with $\flow(\edge) > 0$, $\resEdge  = ((b, t+\dist(a,b)),~ (a,t))$, and has $\edgeCost(\resEdge) = -\cost_{a,b,t}$, $\minCap(\resEdge) = 0$ and $\maxCap(\resEdge) = \flow(\edge)$.
	\item $\resEdgeSet_3 = \edgeSet_3 \cup \set{\resEdge}{\edge \in \edgeSet_3, \flow(\edge) > 0}$ where for each $\edge = ((a, t), ~\sink) \in \edgeSet_3$, $\resEdge  = (\sink,~(a, t))$ with $\edgeCost(\resEdge) = -\exitCost_{\horizon - t}$, $\minCap(\resEdge) = 0$ and $\maxCap(\resEdge) = \flow(\edge)$.
	\item $\resEdgeSet_4 = \set{(\driverNode_i, (\re_i,\te_i))}{i \in \driverSet, ~\flow((\driverNode_i, (\re_i,\te_i))) = 0}  
	\cup \set {((\re_i,\te_i),\driverNode_i) }{i \in \driverSet,  ~\flow((\driverNode_i, (\re_i,\te_i))) = 1}$  
	$\cup \set{(\driverNode_i, \sink)}{i \in \driverSet, ~\driverEntrance_i = 0,~ \flow((\driverNode_i, \sink)) = 0} \cup \set {(\sink,\driverNode_i)}{i \in \driverSet,~ \driverEntrance_i = 0,\flow((\driverNode_i, \sink)) = 1}$. For the forward edges, i.e. $\edge \in \resEdgeSet_4$ s.t. $\edge = (\driverNode_i, (\re_i,\te_i))$ or $\edge = (\driverNode_i, \sink)$, we have $\edgeCost(\edge) = 0$, $\minCap(\resEdge) = 0$.  and $\maxCap(\resEdge) = +\infty$. For each $\edge = ((\re_i,\te_i),\driverNode_i) \in \resEdgeSet_4$, we have $\edgeCost(\edge) = 0$, $\minCap(\resEdge) = 0$, and $\maxCap(\resEdge) = \flow((\driverNode_i, (\re_i,\te_i)))$, and for each $\edge = (\sink,\driverNode_i) \in \resEdgeSet_4$, we have $\edgeCost(\edge) = 0$, $\minCap(\resEdge) = 0$, and $\maxCap(\resEdge) = \flow((\driverNode_i, \sink))$.
\end{enumerate}

From the standard argument on the residual graphs~\cite{ahuja1993network}, we know that the cost of a feasible flow in the residual graph is equal to the incremental cost of the same flow in the original graph. 
For any node $\node = \nodeSet$, the ``shortest distance" from this node to the sink $\sink$ refers to the smallest total cost among all paths from $\node$ to $\sink$ in the residual graph. 
Since the edge costs are equal to driver costs minus rider values, the shortest distance corresponds to the negation of the maximum incremental welfare created by an additional unit of driver flow starting from $\node$, i.e. $-\Phi_{a,t}$ at node $(a,t)$, or $-\Phi_{\driverNode_i}$ at node $\driverNode_i$.
Given $-\Phi$ as the (negation of the) shortest distances, define:
\begin{align*}
	\price_{a,b,t} & \triangleq \Phi_{a,t} - \Phi_{b,t + \dist(a,b)} + \cost_{a,b,t},~\forall (a,b,t) \in \trips, \\
	\mu_j & \triangleq \max \{ \val_j - \price_{\origin_j, \dest_j, \tr_j}, ~0 \},~\forall j \in \riders,
\end{align*}
we show that $(\Phi, \mu)$ forms an optimal solution to \eqref{equ:flow_dual}. The argument is very similar to that of the reduced cost optimality, however, we include the proof here for completeness. %
We first show the feasibility of $(\Phi, \mu)$:
\begin{enumerate}[1.]
	\setlength\itemsep{0.0em}
	\item Constraint \eqref{equ:flow_dual_rider_cnst} holds by definition of $\mu$. 
	\item For \eqref{equ:flow_dual_nonneg_1}, observe that for all $(a,b,t) \in \trips$, there exists an edge $((a,t),~(b, t+ \dist(a,b))) \in \resEdgeSet$ with cost $\cost_{a,b,t}$, thus the shortest distance from $(a,t)$ to $\sink$ is at most $\cost_{a,b,t}$ plus the shortest distance from $(b, t+\dist(a,b))$ to $\sink$, implying $-\Phi_{a,t} \leq -\Phi_{b, t+ \dist(a,b)} + \cost_{a,b,t} \Rightarrow \Phi_{a,t} - \Phi_{b, t+ \dist(a,b)} \geq -\cost_{a,b,t}$.
	\item For \eqref{equ:flow_dual_nonneg_2}, note that $\forall a \in \loc$ and $ \forall t \in \timeSet$, there exists $((a,t), \sink) \in \resEdgeSet$ with cost $\edgeCost(\edge) = \exitCost_{\horizon - t}$. Therefore, the shortest distance $-\Phi_{a, \horizon}$ between $(a, t)$ and $\sink$ is at most $\exitCost_{\horizon - t}$, i.e. $-\Phi_{a, \horizon} \leq \exitCost_{\horizon - t} \Rightarrow \Phi_{a, \horizon} \geq - \exitCost_{\horizon - t}$.
	\item For \eqref{equ:flow_dual_driver_cnst_1}, we know that for each $i \in \driverSet$, there exists an edge $(\driverNode_i, (\re_i, \te_i)) \in \resEdgeSet$ with unlimited capacity and zero cost, therefore the shortest path from $\driverNode_i$ to the sink $\sink$ satisfies $-\Phi_{\driverNode_i} \leq -\Phi_{\re_i, \te_i} \Rightarrow \Phi_{\driverNode_i} \geq \Phi_{\re_i, \te_i}$.
	\item For \eqref{equ:flow_dual_driver_cnst_2}, since for each $i\in \driverSet$ s.t. $\driverEntrance_i = 0$, there exists $\edge = (\driverNode, \sink) \in \resEdgeSet$ with unlimited capacity and zero cost, thus $-\Phi_{\driverNode_i} \leq 0 \Rightarrow = \Phi_{\driverNode_i} \geq 0$. 
	\item \eqref{equ:flow_dual_nonneg_3} holds by definition of $\mu$.
\end{enumerate}

We now show the optimality by examining that the CS conditions \fcs{1}-\fcs{6} hold between the optimal integral flow $\flow$ and $(\Phi, \mu)$:
\begin{enumerate}[1.]
	\setlength\itemsep{0.0em}
	\item To show \fcs{1}, %
	given how $\mu_j$ is defined, we only need to show that when $\flow(\riderEdge_j) > 0$, $\val_j - (\Phi_{\origin_j, \tr_j} - \Phi_{\dest_j, \tr_j + \dist(\origin_j, \dest_j)} + \cost_{\origin_j, \dest_j, \tr_j}) \geq 0$. %
	This holds, since in $\resGraph$, there exists edge $\resRiderEdge_j$ from $(\dest_j, \tr_j + \dist(\origin_j, \dest_j))$ to $(\origin_j, \tr_j)$ with cost $\val_j - \cost_{\origin_j,\dest_j,\tr_j}$, thus the shortest distances must satisfy: $-\Phi_{\dest_j, \tr_j + \dist(\origin_j, \dest_j)} \leq \val_j - \cost_{\origin_j, \dest_j + \tr_j} - \Phi_{\origin_j, \tr_j} \Rightarrow \val_j - (\Phi_{\origin_j, \tr_j} - \Phi_{\dest_j, \tr_j + \dist(\origin_j, \dest_j)} + \cost_{\origin_j, \dest_j, \tr_j}) \geq 0$.
		\item Now consider \fcs{2}. %
		Observe that when $\mu_j > 0$, we must have $\val_j - (\Phi_{\origin_j,\tr_j} - \Phi_{\dest_j, \tr_j + \dist(\origin_j, \dest_j)} + \cost_{\origin_j, \dest_j, \tr_j} ) > 0$, implying $-\Phi_{\origin_j,\tr_j} > - \Phi_{\dest_j,\tr_j + \dist(\origin_j,\dest_j)} - \val_j + \cost_{\origin_j,\dest_j, \tr_j}$, i.e. in the residual graph, the shortest distance from $(\origin_j,\tr_j)$ to the sink is longer than $-\val_j + \cost_{\origin_j,\dest_j,\tr_j}$ plus the shortest distance from $(\dest_j, \tr_j+\dist(\origin_j, \dest_j))$ to the sink. This means that the edge $\riderEdge_j = ((\origin_j, \tr_j), (\dest_j, \tr_j + \dist(\origin_j, \dest_j)))$ with capacity $1$ and cost $-\val_j + \cost_{\origin_j,\dest_j,\tr_j}$ cannot be present in the residual graph, which is the case only if $\flow(\riderEdge_j) = 1$.				
	\item For \fcs{3}: %
	we proved $\Phi_{a,t} - \Phi_{b, t + \dist(a,b)} \geq -\cost_{a,b,t}$ above for feasibility, thus we only need to show the other direction of the inequality. Observing that with $\flow((a,t),~(b, t + \dist(a,b)))> 0$, there exists an edge from $(b, t + \dist(a,b))$ to $(a,t)$ in the residual graph with cost $-\cost_{a,b,t}$ and non-zero capacity, thus the shortest distance from $(b,t+\dist(a,b))$ to the sink is at most $ -\Phi_{a,t} - \cost_{a,b,t}$, implying $\Phi_{a,t} - \Phi_{b, t+\dist(a,b)} \leq -\cost_{a,b,t}$,
	\item Assume \fcs{4} does not hold and given feasibility, %
	we know that there exists $a \in \loc$ and $t \in \timeSet$ s.t. $\flow(((a,t),\sink)) > 0$ and $-\Phi_{a, \horizon} < \exitCost_{\horizon- t}$. This implies that the minimum cost for an extra unit of flow from $(a,t)$ to the sink is lower than $\exitCost_{\horizon-t}$, and the objective of the flow LP can be improved by routing one unit of flow that goes form $(a,t)$ directly to the $\sink$ through this alternative shortest path. This contradicts the optimality of $\flow$.	
	\item For \fcs{5}: given $i \in \driverSet$ s.t. $\flow((\driverNode_i, (\re_i,\te_i))) > 0$, there exists  an edge  from $(\re_i,\te_i)$ to $\driverNode_i$ with zero cost, thus $-\Phi_{\re_i,\te_i} \leq -\Phi_{\driverNode_i} + 0 \Rightarrow \Phi_{\driverNode_i} \leq  \Phi_{\re_i,\te_i}$. Together with \eqref{equ:flow_dual_driver_cnst_1}, we know $\varphi_{\driverNode_i} = \Phi_{\re_i,\te_i}$. 
	\item For \fcs{6}, given $\flow((\driverNode_i, \sink)) > 0$, we know that there's a unit of flow from $\driverNode_i$ to $\sink$ generating a total cost of zero. If $-\Phi_{\driverNode_i} < 0$, there exists a path from $\driverNode_i$ to $\sink$ for which routing a unit of driver flow improves the objective (in comparison to going directly from $\driverNode_i$ to $\sink$). This contradicts the optimality, thus $-\Phi_{\driverNode_i} \geq 0 \Rightarrow \Phi_{\driverNode_i} \leq 0$. Given feasibility, we know $ \Phi_{\driverNode_i} = 0$.
	
\end{enumerate}

This completes the argument that $(\Phi, \mu)$ form an optimal solution to \eqref{equ:flow_LP_simp}, thus the unit replica welfare gain $\{\Phi_{\driverNode_i}\}_{i \in \driverSet}$ is indeed a CE driver utility profile. Similarly, we can show that $-\Psi_{a,t}$ and $-\Psi_{\driverNode_i}$ corresponds to the shortest distance from the sink $\sink$ to the node $(a,t)$ and the node $\driverNode_i$, respectively, and that there exists $\mu' \in \setR^{|\riders|}$ (can be constructed in similar ways as the above $\mu$) s.t. $(\Psi, \mu')$  forms an optimal solution to \eqref{equ:flow_LP_simp}.

\medskip

\noindent{}\textit{Step 2.2. $\Phi$ and $\Psi$ are the Bottom and Top of the Potential Lattice:}

What is left to show is that $\Phi$ and $\Psi$ must be the bottom and top of the lattice formed by all optimal potentials of~\eqref{equ:flow_LP_simp}. For convenience of notation, we now work with the dual of the original flow LP \eqref{equ:flow_LP} where the objective is to minimize the negation of the total social welfare: let $\psi_{a,t}$, $\psi_{\driverNode_i}$ and $\eta_j$ be the dual variables corresponding to the constraints \eqref{equ:flow_balance_constraint}-\eqref{equ:flow_capacity_constraint}, respectively, the dual of \eqref{equ:flow_LP} can be written in the following form:
\begin{align}
	\max ~ & \sum_{i \in \driverSet} \psi_{\driverNode_i} + \sum_{j \in \riders} \eta_j \label{equ:flow_dual_original} \\
	\txtst & \psi_{\origin_j, \tr_j} - \psi_{\dest_j, \tr_j + \dist(\origin_j, \dest_j)} + \eta_j \leq -\val_j + \cost_{\origin_j,\dest_j,\tr_j}, & \forall j \in \riders \notag \\ 
	 	& \psi_{a,t} - \psi_{b, t + \dist(a,b)} \leq \cost_{a,b,t},  & \forall (a,b,t) \in \trips  \notag \\
	 	& \psi_{a,t} \leq \exitCost_{\horizon, t}, & \forall a \in \loc,~ t \in \timeSet \notag \\ 
	 	& \psi_{\driverNode_i} \leq \psi_{\re_i, \te_i}, & \forall i \in \driverSet \notag 	\\ 
		& \psi_{\driverNode_i} \leq 0, & \forall i \in \driverSet \txtst \driverEntrance_i =  0 \notag \\
	 	& \eta_j \leq 0, & \forall j \in \riders \notag
\end{align}

For any optimal solution $(\varphi, \mu)$ to \eqref{equ:flow_dual},  $(\psi, \eta)$ where $\psi = -\varphi$ and $\eta = -\mu$ is an optimal solution to \eqref{equ:flow_dual_original}, and vice versa. Thus we know the $\psi$ variables among optimal solutions of \eqref{equ:flow_dual_original} also form a lattice, and what is left to show is that $-\Phi$ and $-\Psi$ must be the top and bottom of the lattice formed by all optimal potentials of~\eqref{equ:flow_dual_original}.

Recall that for a MCF problem, $\xi$ denotes the boundary condition, so that for each node $\node \in \nodeSet$, $\xi_n$ is the number of the units of flow that enters (or exits, if negative) the network from node $n$. For our problem, $\xi$ is a $|\driverSet| + |\loc|(\horizon + 1)$ dimensional vector, where $\xi_{\driverNode_i} = 1$ and $\xi_{(a,t)} = 0$ (recall that the condition $\xi_{\sink} = -|\driverSet|$ is redundant given the flow balance constraints, therefore is omitted).
Keeping everything else the same, the optimal objective of \eqref{equ:flow_LP} can be thought of as a function of the boundary condition $\xi$, which we denote as $\omega(\xi)$. 
It is known that any potential from the set of all optimal solutions of \eqref{equ:flow_dual_original} must be a subgradient of the function $\omega(\xi)$ (see the proof of Theorem~5.2 in \cite{bertsimas1997introduction}), but we still include the proof here for completeness. 
First, $\omega$ is a convex function of $\xi$ (Theorem~5.1 in~\cite{bertsimas1997introduction} can easily be generalized to incorporate inequality constraints).
Recall that a vector $\psi$ is a \emph{subgradient} of a convex function $\omega$ at $\xi$ if for all $\xi'$,
\begin{align*}
	\omega(\xi) + \psi \cdot (\xi'-\xi) \leq \omega(\xi').
\end{align*}
Let $(\psi, \eta)$ be an optimal solution to \eqref{equ:flow_dual_original}. The strong duality implies $\psi \cdot \xi + \eta \cdot \vec{1} = \omega(\xi)$. Now consider any arbitrary $\xi'$. For any feasible flow $\flow$ given the boundary condition $\xi'$, weak duality implies $\psi \cdot \xi' + \eta \cdot \vec{1} \leq \flow \cdot \gamma $ where $\gamma$ is the vector of all edge costs. Taking the minimum over all feasible flow $\flow$, we obtain $\psi \cdot \xi' + \eta \cdot \vec{1}  \leq \omega(\xi')$. Hence $ \psi \cdot \xi' + \eta \cdot \vec{1}  - ( \psi \cdot \xi + \eta \cdot \vec{1}) \leq \omega(\xi') -\omega(\xi) \Leftrightarrow \omega(\xi) + \psi(\xi'-\xi) \leq \omega(\xi')$, i.e. $\psi$ is a subgradient of $\omega$ at $\xi$.

\medskip

Now we show that for any subgradient $\psi$ of $\omega$ at $\xi$, the entries $\psi_{\driverNode_i}$ is bounded by $-\Psi_{\driverNode_i} \leq \psi_{\driverNode_i} \leq - \Phi_{\driverNode_i}$.\footnote{This is a result of the convexity of $\omega$ and the relationship between directional derivatives and subgradients (see Theorem 3.1.14 in~\cite{nesterov2013introductory}). We include a simple proof here for completeness.} %
Let $\chi_{\driverNode_i}$ be $|\driverSet|+|\loc| (\horizon+1)$ by 1 vector which takes value $0$ except for the $\driverNode_i$ entry, and $\chi_{\driverNode_i} = 1$. 
We know that for any subgradient $\psi$, $\omega(\xi) + \psi \cdot \chi_{\driverNode_i} \leq \omega(\xi + \chi_{\driverNode_i}) \Rightarrow \omega(\xi) + \psi_{\driverNode_i} \leq \omega(\xi + \chi_{\driverNode_i}) \Rightarrow \psi_{\driverNode_i} \leq \omega(\xi + \chi_{\driverNode_i}) - \omega(\xi) = -\Phi_{\driverNode_i}$. 
The last equality holds since the objective $\omega$ is the negation of the optimal total welfare achievable by the vector $\xi$ of driver inflow. Similarly, $\omega(\xi) + \psi \cdot (-\chi_{\driverNode_i}) \leq \omega(\xi - \chi_{\driverNode_i}) \Rightarrow \omega(\xi) - \psi_{\driverNode_i} \leq \omega(\xi - \chi_{\driverNode_i}) \Rightarrow \psi_{\driverNode_i} \geq \omega(\xi) - \omega(\xi - \chi_{\driverNode_i}) = -\Psi_{\driverNode_i}$. We can similarly prove $-\Psi_{a,t}\leq \psi_{a,t} \leq -\Phi_{a,t}$.
This implies that $-\Phi$ and $-\Psi$ are the top and the bottom of the lattice formed by the optimal potentials of \eqref{equ:flow_LP}, respectively, and therefore completes the proof of the lemma. 
\end{proof}

\subsection{Proof of Lemma~\ref{lem:core_equal_CE}} \label{appx:proof_lem_core}

\lemCoreCE*

\begin{proof} 

We first prove that every CE plan is in the core. Let $(x, \actpath, \price)$ be a CE plan with anonymous trip price $\price$, where the rider and driver utilities are given by $u$ and $\pi$. Fix any coalition $(\driverSet', ~\riders')$ of riders and drivers for some $\driverSet'\subseteq \driverSet$ and $\riders' \subseteq \riders$. We prove that
\begin{align}
	\sum_{i \in \driverSet'} \pi_i + \sum_{j \in \riders'} u_j \geq \sw(\driverSet', ~\riders'), \label{equ:coalitional_welfare}
\end{align}
meaning the total utilities for all drivers and riders in the coalition, under the CE plan, is weakly higher than the highest achievable welfare among themselves. This implies that there is no way for the coalition to make an alternative plan, so that everyone has weakly higher utilities, and at least one driver or rider is strictly better off.

We now prove \eqref{equ:coalitional_welfare}. Let $(x', \actpath')$ be an optimal dispatch that achieves the highest coalitional welfare $\sw(\driverSet', ~\riders')$. For all $j \in \riders'$ s.t. $x_j'=1$, let her payment be $\rPayment_j' = \price_{\origin_j, \dest_j, \tr_j}$, the anonymous trip price for the trip according to the original CE plan $(x, \actpath, \price)$. Accordingly, let the payment to each driver  $i \in \driverSet'$ be $\dPayment_i' = \sum_{j \in \riders'} \one{(\origin_j, \dest_j, \tr_j, j) \in \actpath_i'} \price_{\origin_j, \dest_j, \tr_j}$. 
Under this new plan $(x', \actpath', \rPayment', \dPayment')$,
the utility of each rider $j \in \riders'$ is therefore $u_j' = x_j'(\val_j - \rPayment_j')$, and the utility of driver $i \in \driverSet'$ is $\pi_i' = \dPayment_i' - \pathCost_{i,k}$, if the dispatched action path $\actpath_i'$ is consistent with the $k\th$ feasible path of driver $i$ and has total cost of $\pathCost_{i, k}$.

Note that the plan $(x', \actpath', \rPayment', \dPayment')$ is strictly budget balanced, therefore the utility of drivers and riders under this plan add up to the welfare: $\sum_{i \in \driverSet'} \pi_i' + \sum_{j \in \riders'} u_j' = \sw(\driverSet', ~\riders')$. What is left to show is that $u_j'\leq u_j$ for all $j \in \riders'$ and $\pi_i' \leq \pi_i$ for all $i \in \driverSet'$ both hold. This is a consequence of the original plan forming a CE. For the riders, if $x_j' = 0$, then $u_j' = 0 \leq u_j$; if $x_j' = 1$, then $u_j' = \val_j - \price_{\origin_j, \dest_j, \tr_j} \leq u_j$. For each  $i \in \driverSet'$ with $\actpath_i'$ consistent with the $k\th$ feasible path of driver $i$,
\begin{align*}
	\pi_i \geq  \sum_{(a,b,t) \in  \path_{i,k} } \max \{ \price_{a,b,t} , 0 \} - \pathCost_{i,k}  \geq \sum_{j \in \riders'} \one{(\origin_j, \dest_j, \tr_j, j) \in \actpath_i'} \price_{\origin_j, \dest_j, \tr_j} - \pathCost_{i,k}  = \pi_i'.
\end{align*}
This completes the proof of \eqref{equ:coalitional_welfare}, thus all CE plans are in the core.

\bigskip

We now prove the second part of this lemma, that every core outcome that balances budget can be ``priced" in CE. Let $(x, \actpath, \rPayment, \dPayment)$ be a budget balanced core outcome. The following are immediate implications of a core outcome: \vspace{-0.3em}
\begin{enumerate}[1.]
	\setlength\itemsep{0.0em}
	\item The outcome $(x, \actpath, \rPayment, \dPayment)$ must be strictly budget balanced, i.e. $\sum_{j \in \riders} \rPayment_j = \sum_{i \in \driverSet} \dPayment_i$, since otherwise, the entire economy $(\driverSet, \riders)$ will be a blocking coalition.
	\item The outcome $(x, \actpath, \rPayment, \dPayment)$ must be welfare-optimal and achieve $\sw(\driverSet, \riders)$, otherwise the entire economy $(\driverSet, \riders)$ is blocking since an improvement of total utilities is possible.
	\item The plan is individually rational for riders, i.e. $x_j = 0 \Rightarrow \rPayment_j \leq 0$ and $x_j = 1 \Rightarrow \rPayment_j \leq \val_j$  for each $j \in \riders$, otherwise dropping out improves their utilities. 	
\end{enumerate}

We now claim that $\forall i \in \driverSet$, $\sum_{j \in \riders} \rPayment_j \one{(\origin_j, \dest_j, \tr_j, j) \in \actpath_i} = \dPayment_i$, i.e. we must have strict budget balance among driver $i$ and the riders that she picked up. First, $\sum_{j \in \riders} \rPayment_j \one{(\origin_j, \dest_j, \tr_j, j) \in \actpath_i} \leq \dPayment_i$  holds since otherwise this set of driver and riders will be blocking. Since this inequality holds for all $i \in \driverSet$, the plan is budget balanced, and $\rPayment_j \leq 0$ for $j$ s.t. $x_j = 0$, we have 
\begin{align*}
	\sum_{i \in \driverSet} \dPayment_i   = & \sum_{j \in \riders} \rPayment_j  = \sum_{i \in \driverSet} \sum_{j \in \riders} \rPayment_j \one{(\origin_j, \dest_j, \tr_j, j) \in \actpath_i} + \sum_{j \in \riders, x_j = 0} \rPayment_j  \\
	\leq &  \sum_{i \in \driverSet} \sum_{j \in \riders} \rPayment_j \one{(\origin_j, \dest_j, \tr_j, j) \in \actpath_i} \leq \sum_{i \in \driverSet} \dPayment_i,
\end{align*}
which requires that all inequalities hold with equality. This implies  that $\sum_{j \in \riders} \rPayment_j \one{(\origin_j, \dest_j, \tr_j, j) \in \actpath_i} = \dPayment_i$ for all $ i \in \driverSet$, and moreover, $\rPayment_j = 0$ for all $j \in \riders$ s.t. $x_j = 0$.

We now show that prices must be anonymous for the riders, i.e. for any two riders $j \neq j'$ s.t. $(\origin_j,\dest_j,\tr_j) = (\origin_{j'},\dest_{j'},\tr_{j'})$ and $x_j = x_{j'} = 1$, we must have $\rPayment_j = \rPayment_{j'}$. Otherwise, assume w.l.o.g. that $\rPayment_j < \rPayment_{j'}$, and that riders $j$ and $j'$ are picked up by drivers $i$ and $i'$ respectively, we know that rider $j'$, driver $i$, and all of the riders picked up by driver $i$ except for rider $j$, would form a blocking coalition. 
We now construct a set of anonymous trip prices $\price$. For any trip $(a,b,t) \in \trips$,  \vspace{-0.3em}
\begin{enumerate}[(i)]
	\setlength\itemsep{0.0em}
	\item if no rider requests this trip, i.e. $(\origin_j,\dest_j,\tr_j) \neq (a,b,t)$ for all $j \in \riders$, then let $\price_{a,b,t} = 0$.
	\item if some rider requests this trip, but no rider is picked up, then let the price be the highest value for this trip: $\price_{a,b,t} = \max_{j \in \riders,~(\origin_j, \dest_j, \tr_j) = (a,b,t)} \val_j.$
	\item if some rider is picked up, i.e. if $\exists j \in \riders$ s.t. $(\origin_j,\dest_j,\tr_j) = (a,b,t)$ and $ x_j = 1$, let $\price_{a,b,t} = \rPayment_j$. 
\end{enumerate}

Given the anonymity that we proved above, for any rider that is picked up, she pays $\price_{\origin_j, \dest_j, \tr_j}$.  
We claim that $\price_{a,b,t} \geq 0$ for all $(a,b,t) \in \trips$. This is obvious for cases (i) and (ii) above. For case (iii), we only need to show that payments made by riders that are picked up must be non-negative, i.e. $\rPayment_j \geq 0$. This holds, since otherwise the driver who picks up this rider, together with the rest of the riders that this driver picks up, will form a blocking coalition. 

We also claim that $\price_{a,b,t} = 0$ for trips with excessive supply, i.e. if $(a,b,t) \in \actpath_i$ for some $i \in \driverSet$. Consider some trip with $\price_{a,b,t} > 0$. We know that either  case (ii) holds, where there is some rider $j$ willing to pay up to $\price_{a,b,t}$ but is not picked up, or case (iii) holds, where some rider $j$ is paying $\price_{a,b,t}$ to be picked up. In both cases, rider $j$, driver $i$ (who takes the trip $(a,b,t)$ without a rider), and the rest of the riders picked up by driver $i$ will form a blocking coalition--- the rest of the riders can pay the same amounts, driver $i$ can get a higher payment, whereas rider $j$ either gets picked up (case (ii)) or pays less (case (iii)).

We now prove that under plan with anonymous trip prices $(x, \actpath, \price)$, the rider and driver total payments (and therefore utilities) coincide with the original plan. This is obvious for the riders. For each driver, the total payment under plan $(x, \actpath, \price)$ is equal to $\sum_{(a,b,t) \in \path_{i,k}} \max \{ \price_{a,b,t}, ~0 \}$, if $ \actpath$ is consistent with the $k\th$ path of driver $i$. Given the non-negativity of $\price$, and the fact that trips with excessive supply has zero prices, we know that driver $i$ is paid $\sum_{(a,b,t) \in \path_{i,k}} \max \{ \price_{a,b,t}, ~0 \} = \sum_{(a,b,t) \in \path_{i,k}}  \price_{a,b,t} =  \sum_{j \in \riders} \price_{\origin_j, \dest_j, \tr_j} \one{(\origin_j, \dest_j, \tr_j, j) \in \actpath_i} = \sum_{j \in \riders} \rPayment_j \one{(\origin_j, \dest_j, \tr_j, j) \in \actpath_i} = \dPayment_i$.

We complete the proof of this lemma by showing that $(x, \actpath, \price)$ forms a CE. 
Rider best-response is implied by IR for riders requesting trips where no rider is picked up (case (ii)), and for riders that are already picked up (case (iii)). For rider $j$ s.t. $x_j = 0$ but there exists $j'$ s.t. $(\origin_j,\dest_j,\tr_j) = (\origin_{j'},\dest_{j'},\tr_{j'})$ and $x_{j'} = 1$, we claim $\val_j \leq \price_{\origin_j,\dest_j,\tr_j}$. Otherwise, assume that rider $j'$ is picked up by driver $i$, we know rider $j$, driver $i$ and the rest of the riders picked up by driver $i$ will form a blocking coalition. 

What is left to show is driver best response. Assume that driver best response doesn't hold for driver $i$, we know that if $\actpath_i$ is consistent with $\path_{i,k}$, there exists an alternative path $\path_{i,k'} \neq \path_{i,k}$ s.t. $\sum_{(a,b,t) \in \path_{i,k'}} \max \{ \price_{a,b,t}, ~0 \} - \pathCost_{i,k'}  > \sum_{(a,b,t) \in \path_{i,k}} \max \{ \price_{a,b,t}, ~0 \} - \pathCost_{i,k}$. For each $(a,b,t) \in \path_{i,k'}$ s.t. $\price_{a,b,t} > 0$, there exists a rider that is either paying $\price_{a,b,t}$ to be picked up, or is not picked up but willing to pay $\price_{a,b,t}$. Driver $i$, together with all of these riders, will form a blocking coalition. This proves that $(x, \actpath, \price)$ forms a CE, and therefore completes the proof of this theorem. 
\end{proof}

Note that a core outcome does not necessarily use anonymous trip prices. The following example shows that the CE plan with anonymous trip prices $(x, \actpath, \price)$ constructed from a core outcome $(x, \actpath, \rPayment, \dPayment)$ may not pay $\dPayment_{i,t}$ to driver $i$ at time $t$, and we can only guarantee utility equivalence, i.e. the total payment to each driver is equal to $\dPayment_i$. 

\begin{example} Consider an economy with a single location $A$, two time periods, one driver and four riders with $\origin_j = A$, $\dest_j = A$, and $\val_j = 4$ for all $j \in \riders$. Moreover, $\tr_1 = \tr_2 = 0$ and  $\tr_3 = \tr_4 = 1$. Assume all costs are zero. Consider the plan $(x, \actpath, \rPayment, \dPayment)$, where riders $1$ and $3$ are picked up and each pays $4$: $x_1 = x_3 = 1$ and  $\rPayment_1 = \rPayment_3 = 4$. The driver takes action path $\actpath_1 = ((A,A,0,1),~(A,A,1,3))$, however $\dPayment_{1,0} = 2$  and $\dPayment_{1,1} = 6$, i.e. the driver is paid $2$ at time $0$ and $6$ at time $1$. It is easy to see that the outcome is in the core, however, given any CE plan with anonymous trip prices, $\price_{A,A,0} = \price_{A,A,1} = 4$, so the driver needs to be paid $4$ at each of time $0$ and time $1$. \qed
\end{example}

\subsection{Proof of Theorem~\ref{thm:SPE}} \label{appx:proof_thm_spe}

\spe*

\begin{proof}

As is outlined in the body of the paper, what is left to show is incentive alignment. We first show a correspondence of drivers' continuation utilities and the unit replica welfare gains (which implies that the plan determined by the STP mechanism at any time forms a competitive competitive equilibrium), then we show that there is no useful single deviation, implying that always accepting the mechanism's dispatches forms an SPE.

\bigskip

\noindent{}\emph{Step 1.}
Let $(x, \actpath)$ be the optimal dispatch determined by the STP mechanism, and let $\flow$ be a corresponding optimal solution to the flow LP \eqref{equ:flow_LP_simp}, constructed in the same way as in the proof of Lemma~\ref{thm:lp_integrality}.
Setting $u_j = \max\{\val_j - \price_{\origin_j, \dest_j, \tr_j}, ~0 \}$, we know from the proof of Lemma~\ref{thm:opt_pes_plans} that $(\Phi, u)$ forms an optimal solution to the dual of the flow LP \eqref{equ:flow_dual}, and satisfies the CS conditions with $\flow$.

Consider any driver $i \in \driverSet$, who is in the platform, and is available at some location $a$ and time $t$. Assume that the dispatched action path $\actpath_i$ is consistent with $\path_{i,k}$, the $k\th$ feasible path of driver $i$, and assume that path $\path_{i,k}$ ends at location $a'$ and time $t'$ (i.e. the driver is dispatched to exit the platform at $(a',t')$). 
Assuming that all drivers follow the dispatches of the platform at all times, the total payment to driver $i$ from time $t$ onward is:
\begin{align*}
	\sum_{j \in \riders, ~\tr_j \geq t}  \one{(\origin_j, \dest_j, \tr_j, j) \in \actpath_i}\price_{\origin_j, \dest_j, \tr_j} 
	 = \sum_{t'' \geq t}  \one{(a'',b'',t'') \in \path_{i,k}} \price_{a'',b'',t''},
\end{align*}
since when the driver takes a relocation trip $(a'', b'', t'')$ without a rider, $\flow(((a'',t''),~(b'', t+ \dist(a'',b'')))) > 0$, and the complementary slackness condition \fcs{3} implies that the trip price $\price_{a'',b'',t''} = 0$. Moreover, when there exists a driver exiting from $(a',t')$, $\flow((a',t'),~\sink) > 0$, \fcs{4} implies $\Phi_{a',t'} = - \exitCost_{\horizon-t'}$. As a result, the utility of driver $i$ from time $t$ onward is %
\begin{align*}
	& \sum_{j \in \riders,~\tr_j \geq t}  \one{(\origin_j, \dest_j, \tr_j, j) \in \actpath_i}\price_{\origin_j, \dest_j, \tr_j}  - \sum_{t'' \geq t} \one{(a'',b'',t'') \in \path_{i,k}}  \cost_{a'',b'',t''} - \exitCost_{\horizon-t'} \\
	 = & \sum_{t'' \geq t} \one{(a'',b'',t'') \in \path_{i,k}} (\price_{a'',b'',t''} - \cost_{a'',b'',t''} ) - \exitCost_{\horizon-t'} \\
	  = & \sum_{t'' \geq t} \one{(a'',b'',t'') \in \path_{i,k}} \left( \Phi_{a'',t''} - \Phi_{b'', t''+\dist(a'',b'')} + \cost_{a'',b'',t''} - \cost_{a'',b'',t''} \right) - \exitCost_{\horizon-t'} \\
	  =& \Phi_{a,t}  - \Phi_{a',t'} - \exitCost_{\horizon-t'} \\
	  = & \Phi_{a,t}.
\end{align*}

A first implication is that for a driver with $\driverEntrance_i = 1$, her total utility over the planning horizon is $\pi_i = \Phi_{\re_i,\te_i}$. This is equal to $\Phi_{\driverNode_i}$ give \fcs{5}, since $\flow((\driverNode_i, (\re_i,\te_i))) = 1 \rightarrow \Phi_{\driverNode_i} = \Phi_{\re_i,\te_i}$. For a driver with $\driverEntrance_i = 0$ but was dispatched to enter the platform, $\pi_i = \Phi_{\driverNode_i} = \Phi_{\re_i,\te_i}$ holds for the same reason. Drivers with $\driverEntrance_i = 0$ and were not dispatched to enter the platform get $\pi_i = 0$, which is also equal to $\Phi_{\driverNode_i}$, since $\flow((\driverNode_i, \sink)) > 0 \Rightarrow \Phi_{\driverNode_i} = 0$. 
Therefore, $\pi_i = \Phi_{\driverNode_i}$ holds for all $ i \in \driverSet$. Lemma~\ref{thm:optimal_plans} and Lemma~\ref{lem:dual_payment_correspondence} then imply that the plan determined by the STP mechanism forms a CE.

\bigskip

\noindent{}\emph{Step 2.} We now prove that a single deviation from the mechanism's dispatches by any driver at any time is not useful. For drivers who are (at time $t$ given state $\state_t$) \emph{en route}, or have already exited, or has not yet entered, there is effectively only one actions that is available to them, so there is no useful deviation. Therefore we only need to consider a driver that is at time $t$ available.

Given any time $t'$ and state $\state_{t'}$, let $\Phi_{a, t}^{(t')}(\state_{t'})$ be the welfare gain from adding an additional driver (who is already in the platform) at time $t \geq t'$ in location $a$, to the economy starting at time $t'$ and state $\state_{t'}$:
\begin{align}
	\Phi_{a, t}^{(t')}(\state_{t'}) \triangleq \sw( \driverSet^{(t')} (\state_{t'} ) \cup \{(1, a, t - t', \horizon - t')\})  -  \sw( \driverSet^{(t')} (\state_{t'} )). \label{equ:welfare_gain_t}
\end{align}
Here, $(1, a, t - t', \horizon - t')$ is the type of the additional driver that starts at $a$ at time $t$ in the original economy, and therefore at time $t - t'$ in the time-shifted economy $\econ^{(t')}(\state_{t'})$ (where the length of the planning horizon is $\horizon - t'$).

Assume that the current plan $(x\suptprime(\state_{t'}), z\suptprime(\state_{t'}), \rPayment\suptprime(\state_{t'}), \pi\suptprime(\state_{t'}))$ is computed at time $t'$ given state $\state_{t'}$, and that no driver had deviated from the plan since time $t'$. 
Fix any time $t \geq t'$ and let $\state_t$ be the state of the platform at time $t$, if all drivers followed the plan up to time $t$. Consider a driver, say driver $i$, that is available at time $t$ at location $a$, i.e. $\state_{i,t} = (1,a,t)$ or $\state_{i,t} = (0,a,t)$. %
We first argue that deviating from the dispatch to exit (when dispatched to stay) or not enter (when dispatched to enter) is not a useful deviation. This is because exiting or not entering is equivalent to the driver's choosing a path different than the one determined by the plan, and the plan forming a CE implies that no alternative path is more profitable.

What is left to consider is the case where a driver deviated from the dispatches (regardless of which action she is dispatched to take), and did not exit the platform. The only possible deviation action that the driver can take in this case is to relocate to some location $b \in \loc$ that is within reach (i.e. $b \in \loc$ s.t. $t + \dist(a,b) \leq \horizon$).
From Step~1, we know that if all drivers follow the plan until the end of the planning horizon, then a driver with $\state_{i,t} = (1,a,t)$ gets utility $\Phi_{a, t}^{(t')}(\state_{t'})$ in the remaining time periods, and a driver with $\state_{i,t} = (0,a,t)$ gets utility $\max \{ \Phi_{a, t}^{(t')}(\state_{t'}), ~0 \}$. 
The rest of the proof of this theorem shows that by deviating to drive to $b$, the utility of the driver from time $t$ onward is upper bounded by $\Phi_{a, t}^{(t')}(\state_{t'})$, thus this is not a useful deviation.

If all drivers followed the plan at time $t$, denote the state of the platform at time $t+1$ as $\state_{t+1}$.
Now, at state $\state_{t}$, consider the scenario where the rest of the drivers all follow the plan at time $t$, but driver $i$ deviates and relocates to some location $b \in \loc$. Denote the state of the platform at time $t+1$ as $\tilde{\state}_{t+1} \triangleq (\state_{-i,t+1}, (1, b, t+\dist(a,b)))$--- the states of the rest of the drivers are the same as the case if all drivers follow the plan, and driver $i$ will be available at location $b$ at time $t + \dist(a,b)$. Driver $i$ is not paid at time $t$, but incurs cost $\cost_{a,b,t}$ from driving toward $b$. The mechanism replans at time $t+1$, and from time $t+1$ onward, driver $i$'s total utility under $\strategy^\ast$ would be $\Phi_{b, t + \dist(a,b)}^{(t+1)}(\tilde{\state}_{t+1})$, the welfare gain from replicating the driver at $(b, t + \dist(a,b))$, computed at time $t+1$ given state $\tilde{\state}_{t+1}$.
We prove $\Phi_{a, t}^{(t')}(\state_{t'}) \geq \Phi_{b, t + \dist(a,b)}^{(t+1)}(\tilde{\state}_{t+1}) -\cost_{a,b,t}$  by showing: 
\begin{enumerate}[(i)]
	\item $\Phi_{a, t}^{(t')}(\state_{t'}) \geq \Phi_{a, t}^{(t)}(\state_{t}) $,
	\item  $\Phi_{a, t}^{(t)}(\state_{t})  \geq \Phi_{b, t + \dist(a,b)}^{(t+1)}(\state_{t+1}) - \cost_{a,b,t}$ for all $b \in \loc$ s.t. $t + \dist(a,b) \leq \horizon$, and 
	\item $\Phi_{b, t + \dist(a,b)}^{(t+1)}(\state_{t+1}) \geq  \Phi_{b, t + \dist(a,b)}^{(t+1)}(\tilde{\state}_{t+1})$.
\end{enumerate}

\medskip

\noindent{}\textit{Part (i): $\Phi_{a, t}^{(t')}(\state_{t'}) \geq \Phi_{a, t}^{(t)}(\state_{t}) $.} The inequality trivially holds if $t = t'$. Consider $t > t'$. The highest achievable welfare at state $\state_{t'}$ with an additional driver at $(a,t)$ is weakly higher than the welfare of the scenario where all drivers follow the original plan %
until time $t$, and then optimize at time $t$ with all the existing drivers (whose states are now $\state_t$) and the additional driver at $(a,t)$: 
\begin{align*}
	\sw( \driverSet^{(t')} (\state_{t'} ) \cup \{(1, a, t - t', \horizon - t')\}) \geq \left[  \sw( \driverSet^{(t')} (\state_{t'} )) -  \sw( \driverSet^{(t)} (\state_{t} )) \right] + \sw( \driverSet^{(t)} (\state_{t} ) \cup \{(1, a, 0, \horizon - t)\} ).
\end{align*}
Here, %
$(1, a, 0, \horizon-t) = (1, a, t-t, \horizon-t)$ is the type of the additional driver entering at $(a,t)$, in the time-shifted economy starting from $\state_t$. The unit-replica welfare gain is therefore
\begin{align*}
	\Phi_{a, t}^{(t')}(\state_{t'}) = &  \sw( \driverSet^{(t')} (\state_{t'} ) \cup \{(1, a, t-t', \horizon - t')\} ) - \sw( \driverSet^{(t')} (\state_{t'} ))  \\
	\geq& \sw( \driverSet^{(t)} (\state_{t} ) \cup \{(1, a, 0, \horizon - t)\} )  - \sw( \driverSet^{(t)} (\state_{t} ))  \\
	= & \Phi_{a, t}^{(t)}(\state_{t}).
\end{align*}

\medskip

\noindent{}\textit{Part (ii): $\Phi_{a, t}^{(t)}(\state_{t})  \geq \Phi_{b, t + \dist(a,b)}^{(t+1)}(\state_{t+1}) - \cost_{a,b,t}$ for all $b \in \loc$ s.t. $t + \dist(a,b) \leq \horizon$.} This is similar to part (i), observing that at state $\state_t$, the additional driver at $(a,t)$ can relocate to $b$ at a cost of $\cost_{a,b,t}$ %
while the rest of the drivers follow the original plan at time $t$, and then optimize at time $t+1$:
\begin{align*}
	& \sw( \driverSet^{(t)} (\state_{t} ) \cup \{(1, a, 0, \horizon-t)\}) \\
	 \geq &  \left[ \sw( \driverSet^{(t)} (\state_{t} )) -  \sw( \driverSet^{(t+1)} (\state_{t+1} )) \right] - \cost_{a,b,t} + \sw( \driverSet^{(t+1)} (\state_{t+1} ) \cup \{ (1, b, \dist(a,b) -1, \horizon - (t+1)) \}). 
\end{align*}
Here, $(1, b, \dist(a,b) - 1, \horizon - (t+1)) = (1, b, t + \dist(a,b) - (t+1), \horizon - (t+1)) $ is the type of the additional driver at $(b, t+ \dist(a,b))$, time shifted by $t+1$. This gives us:
\begin{align*}
	 \Phi_{a, t}^{(t)}(\state_{t}) 
	 =& \sw( \driverSet^{(t)} (\state_{t} ) \cup \{(1, a, 0, \horizon-t)\} )  - \sw( \driverSet^{(t)} (\state_{t} ) ) \\
	 \geq & \sw( \driverSet^{(t+1)} (\state_{t+1} ) \cup \{(1, b,\dist(a,b)-1, \horizon -(t+1))\} ) - \sw( \driverSet^{(t+1)} (\state_{t+1} )) - \cost_{a,b,t} \\
	 = &\Phi_{b, t + \dist(a,b)}^{(t+1)}(\state_{t+1})  - \cost_{a,b,t} .
\end{align*}

\medskip

\noindent{}\textit{Part (iii): $\Phi_{b, t + \dist(a,b)}^{(t+1)}(\state_{t+1}) \geq  \Phi_{b, t + \dist(a,b)}^{(t+1)}(\tilde{\state}_{t+1})$.} First, observe that the only possible difference between $\state_{t+1}$ and $\tilde{\state}_{t+1}$ is the state of driver $i$. 
Fixing the state of the rest of the riders as $\state_{-i, t+1}$, $\Phi_{b, t + \dist(a,b)}^{(t+1)}(\state_{t+1})$ is the welfare gain from adding an additional driver at $(b, t + \dist(a,b))$ where driver $i$ is at $\state_{i,t+1}$ (the state of driver $i$ if she followed the dispatch at time $t$), whereas $\Phi_{b, t + \dist(a,b)}^{(t+1)}(\tilde{\state}_{t+1})$ is the welfare gain from adding an additional driver at $(b, t + \dist(a,b))$ where driver $i$ is at $(b, t + \dist(a,b))$ (the state of driver $i$ that had deviated while the replan happens at time $t+1$).

When $\state_{i,t+1} = (b, t + \dist(a,b))$ (i.e. the driver's deviation resulted in the same future state at time $t+1$ as in the scenario that she didn't deviate, e.g. instead of picking up rider $j$ who travels to $\dest_j$, the driver relocates with an empty car to $\dest_j$ instead), the inequality trivially holds. When $\state_{i,t+1} = \phi$, i.e. when the driver is asked to exit (or not enter) at time $t$, then $\Phi_{b, t + \dist(a,b)}^{(t+1)}(\state_{t+1}) \geq  \Phi_{b, t + \dist(a,b)}^{(t+1)}(\tilde{\state}_{t+1})$ is implied by the fact that drivers are substitutes, i.e. the more drivers there are, the smaller the marginal welfare contribution of each driver. 
When $\state_{i,t+1} \neq (b, t + \dist(a,b))$ and when $\state_{i,t+1} \neq \phi$, intuitively, the marginal value of an available driver when there is another available driver at the same location is smaller than the marginal value of an available driver when the existing available driver at some other location, i.e. there is stronger substitution among drivers at the same locations, in comparison to that among drivers at different locations.

More formally, let $\xi^\ast \in \setZ^{|\driverSet|+|\loc|(\horizon+1)}$ be the vector of sources of driver flow given state $\state_{-i, t+1}$, s.t. for all $a' \in \loc$, for all $t' \in \timeSet$, let
\begin{align*}
	\xi_{a',t'}^\ast = & \sum_{i' \neq i}\one{\state_{i', t+1} = (1, a', t')} +\sum_{i' \neq i} \one{\state_{i', t+1} = (a'', a', t''),~t'' + \dist(a'',a') = t'}  \\
	& + \sum_{i' \neq i} \one{\state_{i', t+1} = (\origin_j, \dest_j, \tr_j, j),~\dest_j = a', ~\tr_j + \dist(\origin_j, \dest_j) = t'},
\end{align*}
and for or each driver $i' \neq i$, let $\xi_{\driverNode_{i'}}^\ast = \one{\te_i \geq t+1}$. %
Intuitively, $\xi_{a',t'}^\ast$ is the number of drivers in $\driverSet \backslash \{i\}$ who are in the platform and available at $(a',t')$, plus the number of drivers who are en-route relocating to $(a', t')$, plus the number of drivers who are driving a rider to $(a', t')$.

Let $\omega(\xi)$ be the objective of the flow LP \eqref{equ:flow_LP_simp} where the flow boundary condition is given by $\xi$, and $\chi_\node$ be the vector of all zeros but a single $1$ for the entry corresponding to node $\node$. 
If driver $i$ is dispatched to exit (or not enter) the platform at time $t$, i.e. $\state_{i, t+1} = \phi$, the desired property $\Phi_{b, t + \dist(a,b)}^{(t+1)}(\state_{t+1}) \geq  \Phi_{b, t + \dist(a,b)}^{(t+1)}(\tilde{\state}_{t+1})$ is equivalent to
\begin{align*}
	& \omega(\xi^\ast + \chi_{(b, t + \dist(a,b))}) -  \omega(\xi^\ast ) \geq \omega(\xi^\ast + 2 \chi_{(b, t + \dist(a,b))}) -  \omega(\xi^\ast +\chi_{b, t + \dist(a,b)} ). %
\end{align*}
This identity corresponds to the first local exchange property of $M^\natural$ concave functions (equation (4.5) of Theorem~4.1 in~\cite{murota2016discrete}), and that the objective function of MCF problems %
is $M^\natural$ concave (see Example~5 in Section 4.6 of~\cite{murota2016discrete}).\footnote{See~\cite{murota2016discrete} for a general introduction of $M^\natural$ concavity, and also Chapter~9 of~\cite{murota2003discrete} for the related properties of the objectives of network flow problems.} The objective of the flow problem there is defined as a function of the ``sink" nodes in the flow graph, however, the roles of sinks and sources are symmetric: our MCF problem can also be formulated as having a source node at the end of time, where edges go back in time, and the node corresponding to the entering location/time of each driver sinks at most one unit of flow.  

Finally, for the case where driver $i$ is not dispatched to exit the platform at time $t$, let $(a', t')$ be the location and time where driver $i$ will become available again if she followed the dispatch.

The identity that we need to prove $\Phi_{b, t + \dist(a,b)}^{(t+1)}(\state_{t+1}) \geq  \Phi_{b, t + \dist(a,b)}^{(t+1)}(\tilde{\state}_{t+1})$ can now be written as
\begin{align*}
	& \omega(\xi^\ast + \chi_{(a', t')} + \chi_{(b, t + \dist(a,b))}) -  \omega(\xi^\ast + \chi_{(a', t')} ) \geq \omega(\xi^\ast + 2 \chi_{(b, t + \dist(a,b))}) -  \omega(\xi^\ast +\chi_{(b, t + \dist(a,b))} ). %
\end{align*}
This corresponds to the third local exchange property of $M^\natural$ concave functions (equation (4.7) of Theorem~4.1 in~\cite{murota2016discrete}), which intuitively means that there is stronger substitution among drivers at the same location and time, in comparison to drivers that are at different locations and time. 
This completes the proof of the theorem. 
\end{proof}

\section{Additional Discussions and Examples} \label{appx:examples}

We provide in this section additional examples and discussions omitted from the body of the paper.

\subsection{LP Integrality and Existence of CE}  \label{appx:integrality}

We show via the following two examples that when either of the assumptions (S1) and (S2) is violated, the LP relaxation \eqref{equ:lp} of the ILP~\eqref{equ:ilp} may no longer be integral, and that welfare-optimal competitive equilibrium outcomes as defined in Definition~\ref{defn:CE} may not exist. 
We first examine the case where drivers may have different times of exiting the platform.

\begin{example}[Different driver exit times] \label{exmp:different_exit_times}
Consider the economy as shown in Figure~\ref{fig:exmp_diff_exit_times} with three locations $\loc = \{A,B,C \}$ and three time periods. The distances are symmetric and given by $\dist(A,A) = \dist(B,B) = \dist(C,C) = \dist(A,B) = \dist(B,C) = 1$, and $\dist(A,C) = 2$, and assume all trip costs and exit costs are zero. There are three drivers, entering and exiting at:
\vspace{-0.5em}
\begin{multicols}{2}
\begin{itemize}
	\setlength\itemsep{0.0em}
	\item $\re_1 = A$, $\te_1 = 0$, $\tl_1 = 3$,
	\item $\re_2 = B$, $\te_2 = 0$, $\tl_2 = 2$,
	\item $\re_3 = B$, $\te_3 = 1$, $\tl_3 = 3$,	
\end{itemize}
\end{multicols}
\vspace{-0.5em}
\noindent{}
and there are six riders with types:
\vspace{-0.5em}
\begin{multicols}{2}
\begin{itemize}
	\setlength\itemsep{0.0em}
	\item $\origin_1 = A$, $\dest_1 = C$, $\tr_1 = 0$, $\val_1 = 5$,
	\item $\origin_2 = A$, $\dest_2 = B$, $\tr_2 = 1$, $\val_2 = 7$,
	\item $\origin_3 = A$, $\dest_3 = B$, $\tr_3 = 1$, $\val_3 = 1$,
	\item $\origin_4 = B$, $\dest_4 = A$, $\tr_4 = 1$, $\val_4 = 2$,
	\item $\origin_5 = B$, $\dest_5 = A$, $\tr_5 = 1$, $\val_5 = 5$,
	\item $\origin_6 = B$, $\dest_6 = A$, $\tr_6 = 2$, $\val_6 = 4$.	
\end{itemize}
\end{multicols}
\vspace{-0.5em}

\newcommand{\nodeScaleIV}{0.8}
\newcommand{\hdistIV}{4.5}
\newcommand{\vdistIV}{1.8}	

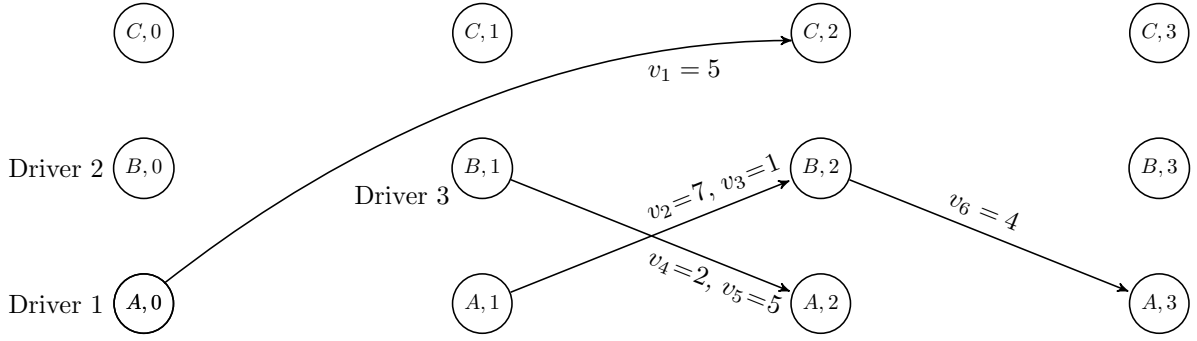
\begin{figure}[t!]
\centering
\begin{tikzpicture}[->,>=stealth',shorten >=1pt, auto, node distance=2cm,semithick][font = \small]
\tikzstyle{vertex} = [fill=white,draw=black,text=black,scale=0.9]

\node[state]         (A0) [scale = \nodeScaleIV] {$A,0$};
\node[state]         (B0) [above of=A0, node distance = \vdistIV cm, scale=\nodeScaleIV] {$B,0$};
\node[state]         (A0) [scale = \nodeScaleIV] {$A,0$};
\node[state]         (C0) [above of=B0, node distance = \vdistIV cm, scale=\nodeScaleIV] {$C,0$};

\node[state]         (A1) [right of=A0, node distance = \hdistIV cm, scale=\nodeScaleIV] {$A,1$};
\node[state]         (B1) [right of=B0, node distance = \hdistIV cm, scale=\nodeScaleIV] {$B,1$};
\node[state]         (C1) [right of=C0, node distance = \hdistIV cm, scale=\nodeScaleIV] {$C,1$};

\node[state]         (A2) [right of=A1, node distance = \hdistIV cm, scale=\nodeScaleIV] {$A,2$}; 
\node[state]         (B2) [right of=B1, node distance = \hdistIV cm, scale=\nodeScaleIV] {$B,2$};
\node[state]         (C2) [right of=C1, node distance = \hdistIV cm, scale=\nodeScaleIV] {$C,2$};

\node[state]         (A3) [right of=A2, node distance = \hdistIV cm, scale=\nodeScaleIV] {$A,3$}; 
\node[state]         (B3) [right of=B2, node distance = \hdistIV cm, scale=\nodeScaleIV] {$B,3$};
\node[state]         (C3) [right of=C2, node distance = \hdistIV cm, scale=\nodeScaleIV] {$C,3$};

\draw[->]  (0.27, 0.27) parabola[bend at end] (8.65, 3.5);
\node[text width=3cm] at (8.2, 3.1) {$\val_1=5$};

\path (A1) edge	node[pos=0.75, sloped, above] {$\val_2\shorteq7$, $\val_3 \shorteq 1$} (B2);
\path (B1) edge	node[pos=0.75, sloped, below] {$\val_4\shorteq2$, $\val_5 \shorteq 5$} (A2);

\path (B2) edge	node[pos=0.45, sloped, above] {$\val_6 = 4$} (A3);

\node[text width=3cm] at (-0.3, 0) {Driver 1};
\node[text width=3cm] at (-0.3, \vdistIV) {Driver 2};
\node[text width=3cm] at (-0.2 + \hdistIV, \vdistIV - 0.35) {Driver 3};
\end{tikzpicture}

\caption{The economy in Example~\ref{exmp:different_exit_times} with three locations $A$, $B$, $C$, three time periods, 6 riders, three drivers starting at $(A,0)$, $(B,0)$ and $(B,1)$, where driver $2$ exits the platform at time $2$. \label{fig:exmp_diff_exit_times}}
\end{figure}

\noindent{}In the unique optimal integral solution, Driver 1 takes the path $z_1^\ast = ((A,A,0),~(A,B,1),~(B,A,2))$ and picks up riders $2$ and $6$. Driver $2$ takes the path $z_2^\ast = ((B,B,0),~(B,A,1))$ and picks up rider $4$. Driver $3$ takes the path $z_3^\ast = ((B,A,1),~(A,A,2))$ and picks up rider $5$. The total social welfare is $\val_2 + \val_6 + \val_4 + \val_5 = 18$. 
The optimal solution of the LP, however, is not integral. Each driver $i$ takes each of their two paths $z_i$ and $z_i'$ with probably 0.5:
\begin{itemize}
	\setlength\itemsep{0.0em}
	\item $z_1 = ((A,C,0),~(C,C,2))$, $z_1' = ((A,A,0),~(A,B,1),~(B,A,2))$,
	\item $z_2 = ((B,B,0),~(B,A,1))$, $z_2' = ((B,A,0),~(A,B,1))$,
	\item $z_3 = ((B,A,1),~(A,A,2))$, $z_3' = ((B,B,1),~(B,A,2))$.
\end{itemize}
The riders $2$, $5$ and $6$ are picked up with probability $1$, whereas rider $1$ is picked up with probability 0.5. The total social welfare is $0.5v_1 + \val_2  + \val_5 + \val_6 = 18.5 > 18$. 
There is a unique solution to the dual LP~\eqref{equ:dual}, which implies anonymous trip prices of $\price_{A,C,0} = 5$ and $\price_{A,B,1} = \price_{B,A,1} = \price_{B,A,2} = 2.5$. These prices do not support the optimal integral solution, since rider $4$ is willing to pay only $\val_2 = 2$ but is picked up and charged $2.5$.

\medskip 

Moreover, we show that no anonymous origin-destination prices support the optimal integral dispatch in competitive equilibrium.\footnote{In general, the non-existence of CE does not imply that there do not exist dynamic ridesharing mechanisms that are SPIC, since a mechanism determining a CE plan is not necessary for the mechanism to be incentive compatible.}
First, rider $1$ with value $5$ is not picked up, therefore the price for the $(A,C,0)$ trip needs to be at least $\price_{A,C,0} \geq 5$. In order for driver $1$ to not regret not taking this trip, the prices for her trips need to be at least $\price_{A,B,1} + \price_{B,A,2} \geq 5$.

Since rider $4$ with value $2$ is picked up, the price for trip $(B,A,1)$ can be at most $2$. As a consequence, the price for the trip $(A,B,1)$ cannot exceed $2$ either, since otherwise, driver $2$ would have incentive to take the path $((B,A,0),~(A,B,1))$ instead. This implies that the price for the trip $(B,A,2)$ needs to be at least $3$.
Note that driver $3$ now prefers taking the path $((B,B,1),~(B,A,2))$ and get paid at least $3$, in comparison to the dispatched trip $((B,A,1),~(A,A,2))$ and gets paid at most $2$. This is a contradiction, and shows that no anonymous OD price supports the welfare-optimal outcome in competitive equilibrium. 
\qed
\end{example}

The reason integrality fails is that the ridesharing problem can no longer be reduced to an MCF problem in the way that we discuss in Appendix~\ref{appx:planning_to_mcf} without loss of generality. In the standard MCF problem, there is a single type of flow flowing through the network, and the optimal flow is guaranteed to be integral.
When drivers have different exiting times, if all units of flow are still treated as homogeneous, the resulting decomposed flow may not send the correct drivers to leave at the correct times. As an example, the optimal homogeneous flow with the same boundary condition in this example can be decomposed into the following three paths: $((A,C,0))$, $((B,A,0),(A,B,1),(B,A,2)), ((B,A,1),(A,A,2))$ with a total social welfare of $21$. However, it cannot be implemented since it is driver 2 who enters at $(B,0)$ and needs to exit at time $2$, but in this decomposition, the flow that corresponds to driver $1$ exits at time $2$. 

When drivers have different exiting times, the MCF problem has heterogeneous flows. 
Similarly, we can construct examples showing that the optimal solutions to the ILP and LP do not coincide if drivers have preference over which location to they end up with at the end of the planning horizon, unless all drivers start at the same time and location.

\bigskip

The following example examines the case where riders are patient and may be willing to wait.

\begin{example}[Patient riders] \label{exmp:patient_riders}

Consider the economy as illustrated in Figure~\ref{fig:exmp_patient_riders} with three locations $\loc= \{A,B,C\}$ and three time periods. The pairwise distances are symmetric and given by $\dist(A,A) = \dist(B,B) = \dist(C,C) = \dist(B,C) = 1$ whereas $\dist(A,B) = \dist(A,C) = 2$. Assume all trip costs and early exit costs are zero. There is a single driver entering at location $B$ at time $0$ who would stay until the end of the planning horizon. There are two riders. Rider $1$ is impatient, requesting a trip at time $0$ from $B$ to $A$, and is willing to pay $\val_1 = 9$. Rider $2$ is willing to pay $\val_2 = 5$ for a trip from $B$ to $C$ at time $0$, but is willing to wait for at most two time periods.

\newcommand{\nodeScaleV}{0.8}
\newcommand{\hdistV}{3}
\newcommand{\vdistV}{1.8}	

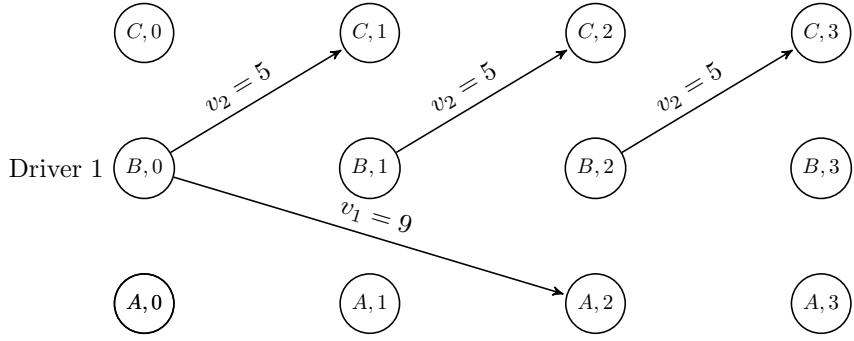
\begin{figure}[t!]
\centering
\begin{tikzpicture}[->,>=stealth',shorten >=1pt, auto, node distance=2cm,semithick][font = \small]
\tikzstyle{vertex} = [fill=white,draw=black,text=black,scale=0.9]

\node[state]         (A0) [scale = \nodeScaleV] {$A,0$};
\node[state]         (B0) [above of=A0, node distance = \vdistV cm, scale=\nodeScaleV] {$B,0$};
\node[state]         (A0) [scale = \nodeScaleV] {$A,0$};
\node[state]         (C0) [above of=B0, node distance = \vdistV cm, scale=\nodeScaleV] {$C,0$};

\node[state]         (A1) [right of=A0, node distance = \hdistV cm, scale=\nodeScaleV] {$A,1$};
\node[state]         (B1) [right of=B0, node distance = \hdistV cm, scale=\nodeScaleV] {$B,1$};
\node[state]         (C1) [right of=C0, node distance = \hdistV cm, scale=\nodeScaleV] {$C,1$};

\node[state]         (A2) [right of=A1, node distance = \hdistV cm, scale=\nodeScaleV] {$A,2$}; 
\node[state]         (B2) [right of=B1, node distance = \hdistV cm, scale=\nodeScaleV] {$B,2$};
\node[state]         (C2) [right of=C1, node distance = \hdistV cm, scale=\nodeScaleV] {$C,2$};

\node[state]         (A3) [right of=A2, node distance = \hdistV cm, scale=\nodeScaleV] {$A,3$}; 
\node[state]         (B3) [right of=B2, node distance = \hdistV cm, scale=\nodeScaleV] {$B,3$};
\node[state]         (C3) [right of=C2, node distance = \hdistV cm, scale=\nodeScaleV] {$C,3$};

\path (B0) edge	node[pos=0.5, sloped, above] {$\val_1=9$} (A2);

\path (B0) edge	node[pos=0.45, sloped, above] {$\val_2 = 5$} (C1);
\path (B1) edge	node[pos=0.45, sloped, above] {$\val_2 = 5$} (C2);
\path (B2) edge	node[pos=0.45, sloped, above] {$\val_2 = 5$} (C3);

\node[text width=3cm] at (-0.3, \vdistV) {Driver 1};

\end{tikzpicture}
\caption{The economy in Example~\ref{exmp:patient_riders} with three locations $A$, $B$, $C$, three time periods, 2 riders, and one driver entering the platform at $(B,0)$.   
\label{fig:exmp_patient_riders}}
\end{figure}

In the optimal integral solution, driver $1$ takes the path $((B,A,0),~(A,A,2))$, picks up rider $1$, and achieves total social welfare of $\val_1 = 9$. In the optimal solution to the LP~\eqref{equ:lp}, however, the driver takes the paths $((B,A,0),~(A,A,2))$ and $((B,C,0),~(C,B,1),~(B,C,2))$, each with probability 0.5. The path $((B,C,0),~(C,B,1),~(B,C,2))$ seemingly have a total value of $10$, therefore the objective of the LP would be $10\times 0.5 + 9 \times 0.5 = 9.5 > 9$. 
The optimal integral solution is not supported by anonymous OD prices in CE either--- since rider $2$ is not picked up, the prices for the $(BC)$ trips starting at times $0,~1$ and $2$ need to be at least $5$. Thus the total payment for the path $((B,C,0),~(C,B,1),~(B,C,2))$ is at least $10$, however, the driver's dispatched trip $(B,A,0)$ pays at most $\val_2 = 9$.
\qed
\end{example}

Similar to the case when drivers have different exit times, integrality fails with patient riders also because there is no direct way of reducing the problem to an integral MCF problem without loss.
In the MCF problem, each rider corresponds to a single edge in the flow graph with edge cost equal to the trip cost minus the rider's value. If the rider is patient, there may be multiple edges that  correspond to the same rider, and there is no easy way expressing the constraint that a rider cannot be picked up more than once without breaking the integrality of the MCF problem.

\subsection{Rider Incentives} \label{appx:rider_incentives}

The following example illustrates (i) the trip-prices and rider-utilities under all CE outcomes may not have lattice structure 
 and (ii) the rider-side VCG prices do not coincide with the prices in the driver-pessimal CE plan,
 and (iii) no welfare-optimal CE mechanism, including the STP mechanism, incentivizes riders to truthfully report their values.

\begin{example} %
\label{exmp:toy_econ_1_continued}

Consider the economy in Figure~\ref{fig:toy_econ_1_appx}, where all trip costs and early exit costs are assumed to be zero. Driver $1$ enters at location $A$ and time $0$, and stays until the end of the planning horizon.
Under the welfare-optimal dispatching, the driver takes the path $((A,A,0), (A,A,1))$ and picks up riders $1$ and $2$, achieving  social welfare $\val_1 + \val_2 = 11$. In the driver pessimal CE plan, the prices for the trips are be $\price_{A,B,0} = 8$, and $\price_{A,A,0} + \price_{A,A,1} = 8$, therefore $\price_{A,A,0} = 5$, $\price_{A,A,1} = 3$ and $\price_{A,A,0} = 2$ and $\price_{A,A,1} = 6$ would both support the driver-pessimal CE outcome.

\newcommand{\nodeScaleI}{0.8}
\newcommand{\hdistI}{3cm}
\newcommand{\vdistI}{2cm}	

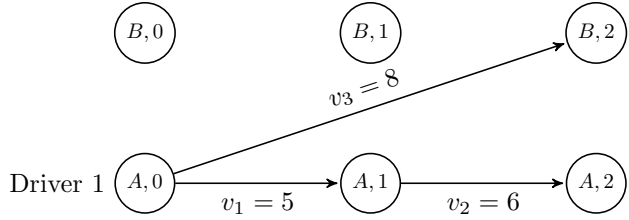
\begin{figure}[t!]
\centering
\begin{tikzpicture}[->,>=stealth',shorten >=1pt,auto,node distance=2cm,semithick][font = \small]
\tikzstyle{vertex}=[fill=white,draw=black,text=black,scale=0.9]

\node[state]         (A0) [scale = \nodeScaleI] {$A,0$};
\node[state]         (B0) [above of=A0, node distance = \vdistI, scale=\nodeScaleI] {$B,0$};
\node[state]         (A1) [right of=A0, node distance = \hdistI, scale=\nodeScaleI] {$A,1$};
\node[state]         (B1) [right of=B0, node distance = \hdistI, scale=\nodeScaleI] {$B,1$};
\node[state]         (A2) [right of=A1, node distance = \hdistI, scale=\nodeScaleI] {$A,2$}; 
\node[state]         (B2) [right of=B1, node distance = \hdistI, scale=\nodeScaleI] {$B,2$};

\path (A0) edge	node[pos=0.5, sloped, below] {$\val_1=5$} (A1);
\path (A1) edge	node[pos=0.5, sloped, below] {$\val_2=6$} (A2);
\path (A0) edge	node[pos=0.5, sloped, above] {$\val_3=8$} (B2);

\node[text width=3cm] at (-0.3, 0) {Driver 1};

\end{tikzpicture}
\caption{The economy in Example~\ref{exmp:toy_econ_1_continued}, with two locations $A$, $B$, two time periods and three riders. 
\label{fig:toy_econ_1_appx}}
\end{figure}

\noindent{}\textit{Lattice Structure: }
It is easy to check that the lowest prices for the trips $(A,A,0)$ and $(A,A,1)$ under all CE outcomes are be $2$ and $3$ respectively. However, setting $\price_{A,A,0} = 2$ and $\price_{A,A,1} = 3$ would not form a CE, since rider $3$ is willing to pay $8$, thus $\price_{A,B,0} \geq 8$ and this violates driver best-response. This implies that trip prices under all CE outcomes do not form a lattice. As a consequence, riders' utilities under all CE outcomes do not form a lattice either.

\medskip
\noindent{}\textit{Rider-side VCG Prices:}
Moreover, we can check that $\price_{A,A,0} = 2$ is what rider $1$ should be charged under the rider-side VCG payment rule: if rider $1$ is not present, rider $3$ gets picked up, thus the total welfare for the rest of the economy increases from $\val_2 = 6$ to $\val_3 = 8$. Similarly, rider $2$'s VCG payment would be $\price_{A,A,1} = 3$. This shows that the VCG payment on the rider side may not support a welfare-optimal outcome in CE.

\medskip
\noindent{}\textit{Rider-side IC:}
This example also implies that the STP mechanism is not incentive compatible on the rider's side. Under any driver-pessimal outcome, we know that one of the riders $1$ and $2$ would be charged a payment that is higher than their VCG price. A simple analysis would show that if the rider who is charged higher than the VCG price reports the VCG price as her value, then her payment under the STP mechanism would be exactly her VCG price. This is a useful deviation. 
More generally, this shows that no welfare-optimal CE outcome would be incentive compatible on the rider's side, since $\price_{A,A,0} + \price_{A,A,1} \geq 8$ under any CE outcome.  
\qed
\end{example}

It is not a coincidence that the lowest possible prices for each rider under all CE outcomes is equal to their rider-side VCG prices. The following theorem shows that the minimum CE prices and the rider-side VCG prices always coincide.

\begin{theorem}[Minimum CE = rider-side VCG] \label{thm:minCE_riderVCG}
For any rider that is picked up in some welfare-optimal dispatching, her rider-side VCG price is equal to the minimum price for her trip among all CE outcomes.
\end{theorem}

\begin{proof} 
For simplicity of notation, assume that driver $j \in \riders$ requests the trip $(a,b,t)$, has value $\val_j$, and is picked up under some welfare-optimal dispatching. We are going to prove:
\begin{enumerate}[(i)]
	\item the price $\price_{a,b,t}$ under any CE outcome is at least the rider-side VCG payment for rider $j$, and
	\item there exists an CE outcome where rider $j$'s trip price is at most her rider-side VCG payment.
\end{enumerate}
Combining (i) and (ii), we know that the rider-side VCG payment has to be the lowest CE price for the trip among all CE outcomes.

\medskip

Let $\tilde{\sw}(\driverSet, \riders)$ be the optimal welfare achieved by the set of drivers $\driverSet$ and the set of riders $\riders$, i.e. the optimal objective of \eqref{equ:ilp}. Moreover, we denote $\tilde{\sw}(\driverSet \cup \{(a,b,t)\}, \riders)$ as the optimal objective of \eqref{equ:ilp} if the trip capacity constraint \eqref{equ:ilp_trip_capacity} for this specific trip is relaxed by $1$, i.e. where there is an additional driver that is only able to complete an $(a,b,t)$ trip.

Similarly, denote $\sw(\driverSet, \riders)$ as the optimal objective of \eqref{equ:lp} and $\sw(\driverSet \cup \{(a,b,t)\}, \riders)$ as the optimal objective of \eqref{equ:lp} with an extra $(a,b,t)$ trip capacity. From the integrality of the LP \eqref{equ:lp} under (S1) and (S2), we know that $\tilde{\sw}(\driverSet, \riders) = \sw(\driverSet, \riders)$, however, we only know $\tilde{\sw}(\driverSet \cup \{(a,b,t)\}, \riders) \leq \sw(\driverSet \cup \{(a,b,t)\}, \riders)$ since with the additional capacity of $1$ for the $(a,b,t)$ trip, it is not obvious whether the LP relaxation would still integral.

\medskip

\noindent{}\textit{Part (i):} To prove (i), first observe that with the same argument on subgradients as in the proof of Lemma~\ref{thm:opt_pes_plans}, we can show that under any CE outcome, the price $\price_{a,b,t}$ as the subgradient w.r.t. the RHS of capacity constraint \eqref{equ:lp_trip_capacity_constraint} in the LP \eqref{equ:lp}, is lower bounded by the welfare gain from relaxing the capacity constraint by 1, i.e. $\price_{a,b,t} \geq \sw(\driverSet \cup \{(a,b,t)\}, \riders) - \sw(\driverSet, \riders)$. This implies that $\price_{a,b,t} \geq \tilde{\sw}(\driverSet \cup \{(a,b,t)\}, \riders) - \tilde{\sw}(\driverSet, \riders)$, i.e. any CE price must be at least the welfare contribution of an additional $(a,b,t)$ trip to the original economy at no cost.

What is left to prove is that $\tilde{\sw}(\driverSet \cup \{(a,b,t)\}, \riders) - \tilde{\sw}(\driverSet, \riders) \geq \price_{a,b,t}^{vcg}$, where $\price_{a,b,t}^{vcg} = \tilde{\sw}(\driverSet, \riders \backslash\{j\}) - (\tilde{\sw}(\driverSet, \riders) - \val_j)$, i.e. the optimal welfare in the economy without rider $j$ minus the welfare of the rest of the riders in the economy with rider $j$. This holds since:
\begin{align*}
	& \tilde{\sw}(\driverSet \cup \{(a,b,t)\}, \riders) - \tilde{\sw}(\driverSet, \riders) -  \price_{a,b,t}^{vcg} \\
	= & \tilde{\sw}(\driverSet \cup \{(a,b,t)\}, \riders) - \tilde{\sw}(\driverSet, \riders)  - (\tilde{\sw}(\driverSet, \riders \backslash\{j\}) - (\tilde{\sw}(\driverSet, \riders) - \val_j)) \\
	=& \tilde{\sw}(\driverSet \cup \{(a,b,t)\}, \riders) - (\tilde{\sw}(\driverSet, \riders \backslash\{j\}) + \val_j)\\
	 \geq & 0.
\end{align*}
The last inequality holds because $\tilde{\sw}(\driverSet \cup \{(a,b,t)\}, \riders)$, the optimal welfare from adding both a trip $(a,b,t)$ (at no cost) and a rider $j$ to the economy $(\driverSet, \riders \backslash \{j\})$, is weakly higher than assigning the trip $(a,b,t)$ to rider $j$ and keeping the plan for the rest of the economy unchanged.

This completes the proof of part (i), that any CE price is weakly higher than the VCG payment. 

\medskip

\noindent{}\textit{Part (ii):} Given $\driverSet$ and $\riders$, we construct an alternative economy $\econ' = (\driverSet, \riders')$ where $\riders'$ and $\riders$ coincide, except for the value of rider $j$: instead of having value $\val_j$, we change her value to her VCG payment in the original economy, i.e. 
\begin{align*}
	\val_j' = \price_{a,b,t}^{vcg}  =  \tilde{\sw}(\driverSet, \riders \backslash\{j\}) - (\tilde{\sw}(\driverSet, \riders) - \val_j). %
\end{align*}

Now consider the optimal dispatching in the economy $\econ'$. If rider $j'$ is not picked up, the optimal welfare is equal to $\tilde{\sw}(\driverSet, \riders \backslash\{j\})$, the highest welfare achievable for the rest of the economy. If rider $j'$ is picked up, the highest achievable welfare for the rest of the economy is equal to $(\tilde{\sw}(\driverSet, \riders) - \val_j)$, therefore the total welfare is $\val_j' + (\tilde{\sw}(\driverSet, \riders) - \val_j) = \tilde{\sw}(\driverSet, \riders \backslash\{j\})$. This implies that in at least one of the optimal plans in economy $\econ'$, rider $j'$ is picked up. 
Let $(x', y')$ be an optimal dispatching in economy $\econ'$ where rider $j'$ is picked up, and let $\price'$ be any CE prices. First, $\price_{a,b,t}' \leq \val_j' = \price_{a,b,t}^{vcg}$ holds, since the outcome forms a CE and rider $j'$ is picked up and must have non-negative utility.

We also claim that the plan with anonymous prices $(x',y', \price')$ also forms a CE in the original economy. 
Since $(x', y', \price')$ is a CE in economy $\econ'$, we know that for drivers, the dispatched paths under $y'$ gives them the highest total payment given prices $\price'$. We also know that trips with excess supply have zero prices. 
For any rider other than $j$, her values in $\econ'$ and $\econ$ are the same thus rider best-response holds. For rider $j$, her value in $\econ$ is $\val_j \geq \price_{a,b,t}^{vcg}  =  \val_j' \geq \price_{a,b,t}'$ thus she weakly prefers  getting picked up $x_j' = 1$ and is also best-responding.
This shows that there exists a CE outcome in $\econ$ where the price for the $(a,b,t)$ trip is at most $\val_j'$, rider $j$'s VCG payment. This completes the proof of part (ii), and also the theorem.
\end{proof}

Finally, we show via the following example that a mechanism that always computes an welfare-optimal dispatching together with rider-side VCG prices (at the beginning of the planning horizon, and also after any driver deviation), is not SPIC for the drivers. This is not implied by Theorem~\ref{thm:minCE_riderVCG} since a mechanism's plan  forming a CE is not necessary for a mechanism being SPIC for drivers.

\begin{example} \label{exmp:spic_dsic}

Consider the economy in Figure~\ref{fig:exmp_dsic_spic}, with $\horizon = 2$ and two locations $\loc = \{A, B \}$ with unit distances: $\dist(a,b) = 1$, $\forall a,b \in \loc$. Assume all trip costs and exit costs are zero. 
There is one driver entering at time $\te_1 = 0$ at location $\re_1 = A$ and leaves at time $\tl_1 = 2$. There are four riders: 
\vspace{-0.5em}
\begin{multicols}{2}
\begin{itemize}
	\setlength\itemsep{0.0em}
	\item Rider 1: $\origin_1 = A$, $\dest_1 = A$, $\tr_1 = 0$, $\val_1=5$, 
	\item Rider 2: $\origin_2 = A$, $\dest_2 = A$, $\tr_2 = 1$, $\val_2=6$,
	\item Rider 3: $\origin_3 = B$, $\dest_3 = B$, $\tr_3 = 1$, $\val_3 = 8$,
	\item Rider 4: $\origin_4 = B$, $\dest_4 = B$, $\tr_4 = 1$, $\val_4 = 8$.
\end{itemize}
\end{multicols}
\vspace{-0.5em}

\newcommand{\nodeScaleI}{0.8}
\newcommand{\hdistI}{4cm}
\newcommand{\vdistI}{1.5cm}	

\begin{figure}[t!]
\centering
\begin{tikzpicture}[->,>=stealth',shorten >=1pt,auto,node distance=2cm,semithick][font = \small]
\tikzstyle{vertex}=[fill=white,draw=black,text=black,scale=0.9]

\node[state]         (A0) [scale = \nodeScaleI] {$A,0$};
\node[state]         (B0) [above of=A0, node distance = \vdistI, scale=\nodeScaleI] {$B,0$};
\node[state]         (A1) [right of=A0, node distance = \hdistI, scale=\nodeScaleI] {$A,1$};
\node[state]         (B1) [right of=B0, node distance = \hdistI, scale=\nodeScaleI] {$B,1$};
\node[state]         (A2) [right of=A1, node distance = \hdistI, scale=\nodeScaleI] {$A,2$}; 
\node[state]         (B2) [right of=B1, node distance = \hdistI, scale=\nodeScaleI] {$B,2$};

\path (A0) edge	node[pos=0.5, sloped, below] {$\val_1=5$} (A1);
\path (A1) edge	node[pos=0.5, sloped, below] {$\val_2=6$} (A2);

\path (B1) edge	node[pos=0.5, sloped, above] {$\val_3=8$, $\val_4 = 8$} (B2);

\node[text width=3cm] at (-0.3, 0) {Driver 1};

\end{tikzpicture}
\caption{The economy in Example~\ref{exmp:spic_dsic}, with two locations $A$, $B$, two time periods and four riders. 
\label{fig:exmp_dsic_spic}}
\end{figure}
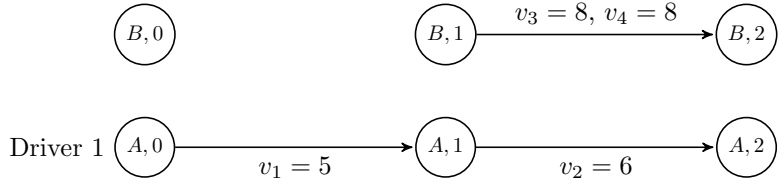

The optimal plan computed at time $0$ has driver $1$ taking the path $((A,A,0),~(A,A,1))$ and picking up riders $1$ and $2$. The rider-side VCG prices for riders $1$ and $2$ would be $2$ and $3$ respectively, thus the driver's total payment, if she follows the dispatches at all times, would be $5$.
Now consider the scenario where the driver relocates to location $B$ at time $0$ instead. When time $1$ comes, the updated plan would dispatch driver $1$ to pick up one of riders $3$ or $4$, and the updated VCG payment for this trip would be $8$. This is a useful deviation, thus the mechanism is not SPIC. \qed
\end{example}

\subsection{Truthful Reporting of Driver Entrance} \label{appx:driver_entrance}

Throughout the paper, we assumed a complete information model, where the mechanism knows about the entering location and time for all the drivers. In this section, we discuss the scenario where the location and time where a driver first becomes available to pick up are drivers' private information, and the mechanism needs to ask the drivers to report their entrance information. %
Here, we still assume that all drivers stay until at least the end of the planning horizon.

\begin{theorem} Under the STP mechanism, for driver $i$ who is available to pick up at location $\re_i$ starting at time $\te_i$, it is not useful for her to report some entrance location and time $(\hat{\te_i}, \hat{\re_i}) \in \loc \times \timeSet$
where $\hat{\te_i} \geq \te_i + \dist(\re_i, \hat{\re_i})$, and then enter the platform at $(\hat{\re_i}, \hat{\te_i})$. 
\end{theorem}

\begin{proof}

First, observe that for driver $i$ whose true entering location and time is $(\re_i, \te_i)$, the driver is only able to enter at $(\re_i, \te_i)$, or at some $(\hat{\re_i}, \hat{\te_i}) \in \loc \times \timeSet$ where $\hat{\te_i} \geq \te_i + \dist(\re_i, \hat{\re_i})$. Assume that drivers all follow the SPE once they entered the platform, and always accepts the dispatches of the mechanisms. If driver $i$ reports truthfully, her total payment would be $\Phi_{\driverNode_i} = \max\{ \Phi_{\re_i, \te_i}, ~0 \}$, the welfare gain of the economy from replicating this driver.

Following the same notation as in the proof of Theorem~\ref{thm:SPE}, we use $\xi^\ast$ to denote the boundary condition (given the initial state) for the flow problem of the economy except for driver $i$, i.e. $\xi^\ast_{\driverNode_{i'}} = 1$ for all $i'\neq i$, and $\xi_{\node}^\ast = 0$ for all other node $ \node \in \nodeSet$. 
Let $\omega(\cdot)$ be the optimal objective of the corresponding flow problem, $\Phi_{\re_i, \te_i}$ can be written as:
\begin{align*}
	\Phi_{\re_i, \te_i} = \omega(\xi^\ast + 2\chi_{(\re_i, \te_i)}) - \sw(\xi^\ast + \chi_{(\re_i, \te_i)}).
\end{align*}
Here, $\chi_{(\re_i, \te_i)}$ is a $|\loc|( \horizon + 1) + |\driverSet|$ by $1$ vector with all zero entries, except a single $1$ at the $(\re_i, \te_i)$ entry. 
If the driver reports $(\hat{\re_i}, \hat{\te_i})$ as her entering location and time, and actually enters at $(\hat{\re_i}, \hat{\te_i})$, her equilibrium payoff for the rest of the planning horizon can be written as
\begin{align*}
	\omega(\xi^\ast + 2 \chi_{(\hat{\re_i}, \hat{\te_i})}) - \sw(\xi^\ast + \chi_{(\hat{\re_i}, \hat{\te_i})} ).
\end{align*}
Let $g$ be the lowest cost that that the driver has to incur, while moving from $(\re_i, \te_i)$ to $(\hat{\re_i}, \hat{\te_i})$. 
By reporting and entering at $(\hat{\re_i}, \hat{\te_i})$, the agent's total utility is at most $\omega(\xi^\ast + 2 \chi_{(\hat{\re_i}, \hat{\te_i})}) - \sw(\xi^\ast + \chi_{(\hat{\re_i}, \hat{\te_i})} ) - g$. 
We show that this is not a useful deviation, since
\begin{align*}
	& \sw(\xi + 2\chi_{\re_i, \te_i}) - \sw(\xi + \chi_{\re_i, \te_i}) - (\sw(\xi + 2\chi_{\hat{\re_i}, \hat{\te_i}}) - \sw(\xi^\ast + \chi_{(\hat{\re_i}, \hat{\te_i})} ) - g ) \\
	\geq & \sw(\xi + 2\chi_{\re_i, \te_i}) - \sw(\xi + \chi_{\re_i, \te_i}) - (\sw(\xi + \chi_{\re_i, \te_i} + \chi_{\hat{\re_i}, \hat{\te_i}}) - \sw(\xi + \chi_{\re_i, \te_i}) - g)  \\
	=  & \sw(\xi + 2\chi_{\re_i, \te_i})  - \sw(\xi + \chi_{\re_i, \te_i} + \chi_{\hat{\re_i}, \hat{\te_i}}) - g \\ 
	\geq & 0.
\end{align*}
The first inequality holds due to the local exchange property of the $M^\natural$ concave functions (Equation~(4.7) in~\cite{murota2016discrete}), and the last inequality holds since the highest achievable welfare achievable from two additional drivers at $(\re_i, \te_i)$ is (weakly) higher than the scenario where one of these drivers has to move to $(\hat{\re_i}, \hat{\te_i})$ (at the lowest possible cost $g$). %
\end{proof}

\bigskip

This result on the truthfulness of driver entrance reports, however, considers only the scenario that the driver actually enters at the location and time as she reported. We may also consider a mechanism that takes the drivers' reports of entering location and time, plans accordingly at the beginning of the planning horizon, but replans if any driver's entrance action turns out to be different from expected/reported, without penalizing any driver that had deviated. 
The following example shows that when allowing arbitrary driver entrance regardless of their report, the STP mechanism does not incentivize the drivers to truthfully report their entering location and time. %

\begin{example} \label{exmp:driver_entrance}

Consider the economy as shown in Figure~\ref{fig:exmp_driver_entrance}. The planning horizon is $\horizon = 3$ and there are two locations $\loc = \{A, B\}$ with unit distances $\dist(a,b) = 1$ for all $a,b \in \loc$.  Assume that all trip costs and early exiting costs are zero.
There is one driver entering at time $\te_1 = 0$ at location $\re_1 = A$ and leaves at time $\tl_1 = 3$. There is another driver, whose true entering time and location is $\te_2 = 2$ and $\re_2 = B$.
There are four riders with type: 
\vspace{-0.5em}
\begin{multicols}{2}
\begin{itemize}
	\setlength\itemsep{0.0em}
	\item Rider 1: $\origin_1 = B$, $\dest_1 = B$, $\tr_1 = 1$, $\val_1 = 10$, 
	\item Rider 2: $\origin_2 = A$, $\dest_2 = A$, $\tr_2 = 1$, $\val_2 = 8$,
	\item Rider 3: $\origin_3 = B$, $\dest_3 = B$, $\tr_3 = 2$, $\val_3 = 5$,
	\item Rider 4: $\origin_4 = B$, $\dest_4 = B$, $\tr_4 = 2$, $\val_4 = 4$.
\end{itemize}
\end{multicols}
\vspace{-0.5em}

\newcommand{\nodeScaleI}{0.8}
\newcommand{\hdistI}{3.5cm}
\newcommand{\vdistI}{1.5cm}	

\begin{figure}[t!]
\centering
\begin{tikzpicture}[->,>=stealth',shorten >=1pt,auto,node distance=2cm,semithick][font = \small]
\tikzstyle{vertex}=[fill=white,draw=black,text=black,scale=0.9]

\node[state]         (A0) [scale = \nodeScaleI] {$A,0$};
\node[state]         (B0) [above of=A0, node distance = \vdistI, scale=\nodeScaleI] {$B,0$};

\node[state]         (A1) [right of=A0, node distance = \hdistI, scale=\nodeScaleI] {$A,1$};
\node[state]         (B1) [right of=B0, node distance = \hdistI, scale=\nodeScaleI] {$B,1$};

\node[state]         (A2) [right of=A1, node distance = \hdistI, scale=\nodeScaleI] {$A,2$}; 
\node[state]         (B2) [right of=B1, node distance = \hdistI, scale=\nodeScaleI] {$B,2$};
\node[state]         (A3) [right of=A2, node distance = \hdistI, scale=\nodeScaleI] {$A,3$}; 
\node[state]         (B3) [right of=B2, node distance = \hdistI, scale=\nodeScaleI] {$B,3$};

\path (B1) edge	node[pos=0.5, sloped, below] {$\val_1=10$} (B2);
\path (A1) edge	node[pos=0.5, sloped, below] {$\val_2=8$} (A2);
\path (B2) edge	node[pos=0.5, sloped, below] {$\val_3=5$, $\val_4 = 4$} (B3);

\node[text width=3cm] at (-0.3, 0) {Driver 1};
\node[text width=3cm] at (8, 2.2) {Driver 2};

\end{tikzpicture}
\caption{The economy in Example~\ref{exmp:driver_entrance}, with two locations, three time periods and four riders. 
\label{fig:exmp_driver_entrance}}
\end{figure}
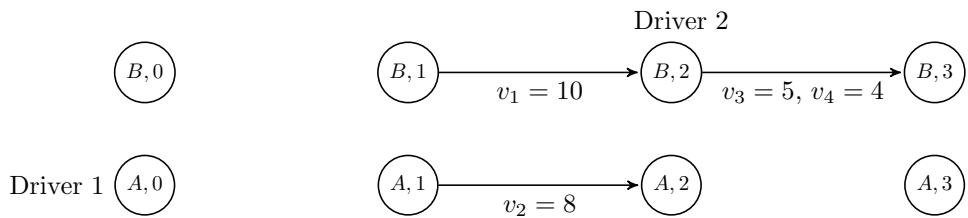

Under the STP mechanism, if both drivers report their entrance location and time truthfully, the welfare-optimal plan dispatches driver $1$ to take the path $((A,B,0),~(B,B,1),~(B,B,2))$ and pick up riders $1$ and $3$. Driver $2$ is dispatched to pick up rider $4$, and her payment would be $0$, the welfare gain from an additional driver entering at $(B,2)$.

However, if driver $2$ reports $(B,1)$ as her entering time and location, then under the STP mechanism, driver $1$ would be dispatched to go to $(A,1)$ to pick up rider $2$ at time $1$. When time $1$ comes, driver $2$ fails to enter at $(B,1)$, and regardless of any future entrance of driver $2$, the optimal plan at time $1$ is for driver $1$ to pick up rider $2$. When time $2$ comes, driver $2$ can then decide to actually enter at location $B$. The mechanism would replan again, dispatching driver $2$ to pick up rider $3$. The welfare gain from an additional driver at $(B,2)$ would now be $4$, thus the driver's new payment would be $4$, and this is a useful deviation. 
\qed
\end{example}

\subsection{Naive Update of Static Plans}

The last example in this section shows that a mechanism that always re-computes a driver-optimal plan at all times is not envy-free for the drivers, and may not be incentive compatible for drivers, depending on how the mechanism breaks ties among different driver-optimal plans.

\begin{example}[Repeated driver-optimal static CE mechanism]\label{exmp:always_compute_driver_opt} 

\newcommand{\nodeScaleV}{0.8}
\newcommand{\hdistV}{2.8}
\newcommand{\vdistV}{1.8}	

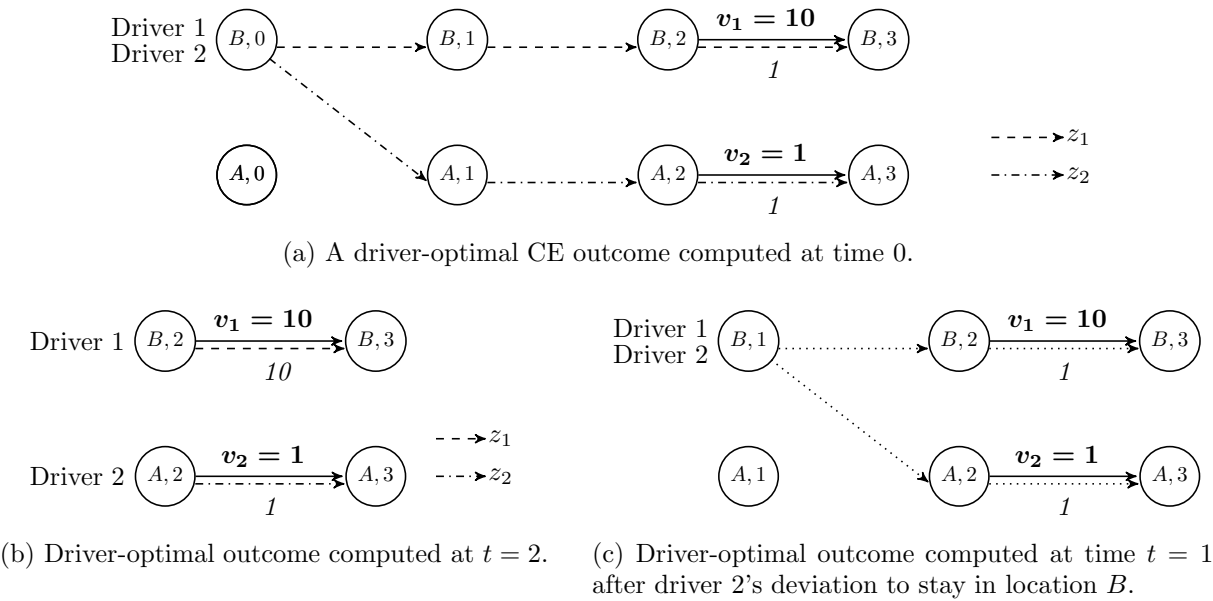
\begin{figure}[htpb!]
\centering
    
\begin{subfigure}[t]{1\textwidth}
\centering
\begin{tikzpicture}[->,>=stealth',shorten >=1pt, auto, node distance=2cm,semithick][font = \small]
\tikzstyle{vertex} = [fill=white,draw=black,text=black,scale=0.9]

\node[state]         (A0) [scale = \nodeScaleV] {$A,0$};
\node[state]         (B0) [above of=A0, node distance = \vdistV cm, scale=\nodeScaleV] {$B,0$};
\node[state]         (A0) [scale = \nodeScaleV] {$A,0$};

\node[state]         (A1) [right of=A0, node distance = \hdistV cm, scale=\nodeScaleV] {$A,1$};
\node[state]         (B1) [right of=B0, node distance = \hdistV cm, scale=\nodeScaleV] {$B,1$};

\node[state]         (A2) [right of=A1, node distance = \hdistV cm, scale=\nodeScaleV] {$A,2$}; 
\node[state]         (B2) [right of=B1, node distance = \hdistV cm, scale=\nodeScaleV] {$B,2$};

\node[state]         (A3) [right of=A2, node distance = \hdistV cm, scale=\nodeScaleV] {$A,3$}; 
\node[state]         (B3) [right of=B2, node distance = \hdistV cm, scale=\nodeScaleV] {$B,3$};

\path (B2) edge	node[pos=0.45, sloped, above] {$\boldsymbol{\val_1 = 10}$} (B3);
\path (A2) edge	node[pos=0.45, sloped, above] {$\boldsymbol{\val_2 = 1}$} (A3);

\node[text width=3cm] at (-0.3, \vdistV + 0.17) {Driver 1};
\node[text width=3cm] at (-0.3, \vdistV -0.17) {Driver 2};

\draw[dashed] (0.4, \vdistV - 0.1) -- (\hdistV - 0.38, \vdistV  - 0.1);
\draw[dashed] (\hdistV + 0.4, \vdistV - 0.1) -- (\hdistV * 2 - 0.38, \vdistV  - 0.1);
\draw[dashed] (\hdistV *2 + 0.4, \vdistV - 0.1) -- (\hdistV * 3 - 0.38, \vdistV  - 0.1);

\draw[dashdotted] ( 0.3, \vdistV  - 0.25) -- (\hdistV *1 - 0.38, - 0.1);
\draw[dashdotted] (\hdistV + 0.4, - 0.1) -- (\hdistV * 2 - 0.38, - 0.1);
\draw[dashdotted] (\hdistV * 2 + 0.4, - 0.1) -- (\hdistV * 3 - 0.38, - 0.1);

\draw[dashed](\hdistV * 3 + 1.5, 0.5) -- (\hdistV * 3 + 2.5, 0.5);
\draw[dashdotted](\hdistV * 3 + 1.5, 0) -- (\hdistV * 3 + 2.5, 0);

\node[text width=0.4cm] at (\hdistV * 3 + 2.7, 0.5) {$z_1$};
\node[text width=0.4cm] at (\hdistV * 3 + 2.7, 0) {$z_2$};

\node[text width=1cm] at (\hdistV * 2 + 1.8, \vdistV - 0.4) {\emph{1}};
\node[text width=1cm] at (\hdistV * 2 + 1.8, -0.4) {\emph{1}};
  
\end{tikzpicture}
\caption{A driver-optimal CE outcome computed at time $0$. \label{fig:exmp_repOpt_time0Plan}}
\end{subfigure}

\vspace{1em}
\begin{subfigure}[t]{0.45\textwidth}
\centering
\begin{tikzpicture}[->,>=stealth',shorten >=1pt, auto, node distance=2cm,semithick][font = \small]
\tikzstyle{vertex} = [fill=white,draw=black,text=black,scale=0.9]

\node[state]         (A2) [scale = \nodeScaleV] {$A,2$};
\node[state]         (B2) [above of=A2, node distance = \vdistV cm, scale=\nodeScaleV] {$B,2$};

\node[state]         (A3) [right of=A2, node distance = \hdistV cm, scale=\nodeScaleV] {$A,3$}; 
\node[state]         (B3) [right of=B2, node distance = \hdistV cm, scale=\nodeScaleV] {$B,3$};

\path (B2) edge	node[pos=0.45, sloped, above] {$\boldsymbol{\val_1 = 10}$} (B3);
\path (A2) edge	node[pos=0.45, sloped, above] {$\boldsymbol{\val_2 = 1}$} (A3);

\node[text width=3cm] at (-0.3, \vdistV) {Driver 1};
\node[text width=3cm] at (-0.3, 0) {Driver 2};

\draw[dashed] (0.4, \vdistV - 0.1) -- (\hdistV - 0.38, \vdistV  - 0.1);
\draw[dashdotted] (0.4, - 0.1) -- (\hdistV - 0.38, - 0.1);

\draw[dashed](\hdistV + 0.8, 0.5) -- (\hdistV * 1 + 1.5, 0.5);
\draw[dashdotted](\hdistV + 0.8, 0) -- (\hdistV * 1 + 1.5, 0);

\node[text width=0.4cm] at (\hdistV  + 1.7, 0.5) {$z_1$};
\node[text width=0.4cm] at (\hdistV + 1.7, 0) {$z_2$};

\node[text width=1cm] at ( 1.8, \vdistV - 0.4) {\emph{10}};
\node[text width=1cm] at (1.8, -0.4) {\emph{1}};

\end{tikzpicture}
\caption{Driver-optimal outcome computed at $t=2$.\label{fig:exmp_repOpt_time2Plan}}
\end{subfigure}%
\hspace{1em}
\begin{subfigure}[t]{0.5\textwidth}
\centering
\begin{tikzpicture}[->,>=stealth',shorten >=1pt, auto, node distance=2cm,semithick][font = \small]
\tikzstyle{vertex} = [fill=white,draw=black,text=black,scale=0.9]

\node[state]         (A1) [scale = \nodeScaleV] {$A,1$};
\node[state]         (B1) [above of=A1, node distance = \vdistV cm, scale=\nodeScaleV] {$B,1$};

\node[state]         (A2) [right of=A1, node distance = \hdistV cm, scale=\nodeScaleV] {$A,2$}; 
\node[state]         (B2) [right of=B1, node distance = \hdistV cm, scale=\nodeScaleV] {$B,2$};

\node[state]         (A3) [right of=A2, node distance = \hdistV cm, scale=\nodeScaleV] {$A,3$}; 
\node[state]         (B3) [right of=B2, node distance = \hdistV cm, scale=\nodeScaleV] {$B,3$};

\path (B2) edge	node[pos=0.45, sloped, above] {$\boldsymbol{\val_1 = 10}$} (B3);
\path (A2) edge	node[pos=0.45, sloped, above] {$\boldsymbol{\val_2 = 1}$} (A3);

\node[text width=3cm] at (-0.3, \vdistV+0.17) {Driver 1};
\node[text width=3cm] at (-0.3, \vdistV-0.17) {Driver 2};

\draw[dotted] ( 0.4, \vdistV - 0.1) -- (\hdistV  - 0.38, \vdistV  - 0.1);
\draw[dotted] (\hdistV *1 + 0.4, \vdistV - 0.1) -- (\hdistV * 2 - 0.38, \vdistV  - 0.1);

\draw[dotted] ( 0.3, \vdistV  - 0.25) -- (\hdistV *1 - 0.38, - 0.1);
\draw[dotted] (\hdistV *1 + 0.4, - 0.1) -- (\hdistV * 2 - 0.38, - 0.1);

\node[text width=1cm] at (\hdistV * 1 + 1.8, \vdistV - 0.4) {\emph{1}};
\node[text width=1cm] at (\hdistV * 1 + 1.8, -0.4) {\emph{1}};

\end{tikzpicture}
\caption{Driver-optimal outcome computed at time $t=1$ after driver 2's deviation to stay in location $B$. \label{fig:exmp_repOpt_deviation}}
\end{subfigure}%
\caption{The economy in Example~\ref{exmp:always_compute_driver_opt}, and driver-optimal CE plans computed at different states.\label{fig:exmp_always_compute_driver_opt} }    
\end{figure}

Consider the economy as illustrated in Figure~\ref{fig:exmp_always_compute_driver_opt} where all costs are zero, and a mechanism that repeatedly computes a driver-optimal CE outcome at all times, regardless of whether any deviation happened. Under the CE outcome computed at time $0$ as in Figure~\ref{fig:exmp_repOpt_time0Plan}, both riders $1$ and $2$ are picked up, and the prices for the trips are both $\price_{B,B,2} = \price_{A,A,2} = 1$. 

Assume that both drivers follow the mechanism until time $2$, the new driver-optimal outcome computed at time $2$ is as illustrated in Figure~\ref{fig:exmp_repOpt_time2Plan}, where the price for trip $(B,B,2)$ becomes $10$, the highest market-clearing price at this state. This shows that the ``time $0$" plan under the mechanism, i.e. the actual outcome where all drivers follow the mechanism's dispatch at all times, is not envy-free for drivers, since the total payment to driver $1$ is $10$, higher than that of driver $2$.

Now consider the scenario where driver $2$ stayed in location $B$ at time $0$ instead of following the dispatch and relocate to $A$. Under the optimal CE outcome from time $1$ onward as in Figure~\ref{fig:exmp_repOpt_deviation}, the two drivers take the paths $((B,B,1),~(B,B,2))$ and $((B,A,1),~(A,A,2))$ respectively, and pick up both riders. The IC property of the mechanism now depends on how the mechanism breaks ties among driver-optimal CE plans, but as long as the mechanism dispatches driver $2$ to take the path $((B,B,1),~(B,B,2))$ with non-zero probability (which would be the case if ties are broken uniformly at random), this would be a useful deviation for driver $2$. Once driver $2$ followed the plan and reach $(B,2)$ as driver $1$ arrived at $(A,2)$, the newly updated price for the trip $(B,B,2)$ would again become $10$, higher than driver $2$'s original payment. \qed
\end{example}

\subsection{On the Smoothness of Prices} \label{appx:price_smoothness}

Competitive equilibrium requires that all feasible paths of the a driver have the same total payoff. Under the myopic pricing mechanism, big gaps in prices for neighboring locations and times violate this and incentivize drivers to strategize, whereas the competitive equilibrium prices are ``more smooth'' in space and time (see Example~\ref{exmp:super_bowl} and the simulation results presented in Appendix~\ref{sec:illustrations}).

Under our model, it is neither necessary nor sufficient for prices to be smooth in the sense that the price gaps are bounded in space and time. 
It is possible to construct economies to show that there does not exist an upper bound on how much the price can change. Consider the example as illustrated in Figure~\ref{fig:gap_time}, where it is welfare-optimal for the only driver to pick up riders $1$ and $2$. The price of the $(A,B)$ trip is upper bounded by $\epsilon$ at time $0$, but needs to be at least $1$ at time $1$ since rider $3$ is not picked up. Similarly, the example illustrated in Figure~\ref{fig:gap_space} shows that big price gap in space is necessary. 
	
\newcommand{\nodeScaleI}{0.8}
\newcommand{\hdistI}{2.5cm}
\newcommand{\vdistI}{1.4cm}

\begin{figure}[t!]
\centering
\begin{subfigure}[t]{0.49\textwidth}
\begin{tikzpicture}[->,>=stealth',shorten >=1pt,auto,node distance=2cm,semithick][font = \small]
\tikzstyle{vertex}=[fill=white,draw=black,text=black,scale=0.9]

\node[state]         (A0) [scale = \nodeScaleI] {$A,0$};
\node[state]         (B0) [above of=A0, node distance = \vdistI, scale=\nodeScaleI] {$B,0$};
\node[state]         (A1) [right of=A0, node distance = \hdistI, scale=\nodeScaleI] {$A,1$};
\node[state]         (B1) [right of=B0, node distance = \hdistI, scale=\nodeScaleI] {$B,1$};
\node[state]         (A2) [right of=A1, node distance = \hdistI, scale=\nodeScaleI] {$A,2$}; 
\node[state]         (B2) [right of=B1, node distance = \hdistI, scale=\nodeScaleI] {$B,2$};

\path (A0) edge	node[pos=0.5, sloped, above] {$\val_1=\epsilon$} (B1);
\path (B1) edge	node[pos=0.5, sloped, above] {$\val_2 = 1$} (B2);
\path (A1) edge	node[pos=0.35, sloped, above] {$\val_3 = 1$} (B2);

\node[text width=3cm] at (-0.3, 0) {Driver 1};

\end{tikzpicture}
\caption{$\price_{A,B,1} \gg \price_{A,B,0}$.\label{fig:gap_time}}
\end{subfigure}
\begin{subfigure}[t]{0.49\textwidth}
\begin{tikzpicture}[->,>=stealth',shorten >=1pt,auto,node distance=2cm,semithick][font = \small]
\tikzstyle{vertex}=[fill=white,draw=black,text=black,scale=0.9]

\node[state]         (A0) [scale = \nodeScaleI] {$A,0$};
\node[state]         (B0) [above of=A0, node distance = \vdistI, scale=\nodeScaleI] {$B,0$};
\node[state]         (A1) [right of=A0, node distance = \hdistI, scale=\nodeScaleI] {$A,1$};
\node[state]         (B1) [right of=B0, node distance = \hdistI, scale=\nodeScaleI] {$B,1$};
\node[state]         (A2) [right of=A1, node distance = \hdistI, scale=\nodeScaleI] {$A,2$}; 
\node[state]         (B2) [right of=B1, node distance = \hdistI, scale=\nodeScaleI] {$B,2$};

\path (A0) edge	node[pos=0.5, sloped, above] {$\val_1 = 1$} (B1);
\path (B1) edge	node[pos=0.5, sloped, above] {$\val_2 = \epsilon$} (B2);
\path (A1) edge	node[pos=0.35, sloped, above] {$\val_3 = 1$} (B2);

\node[text width=3cm] at (-0.3, 0) {Driver 1};

\end{tikzpicture}
\caption{$\price_{A,B,1} \gg \price_{B,B,1}$. \label{fig:gap_space}}
\end{subfigure}
\caption{Example economies illustrating that price gaps can be unbounded in space and time.  
\label{fig:no_bound_on_price_gap}}
\end{figure}

\subsection{Asymmetry and Triangle Inequality}%
\label{appx:triangle_ineq}

For simplicity of notation, we assumed in the body of the paper that the distance $\dist(a,b)$ between a pair of locations $a, b \in \loc$ is fixed. 
Alternatively, we can  allow the distance between a pair of locations to change over time, modeling the changes in traffic conditions, i.e. a trip from $a$ to $b$ starting at time $t$ ends at time $t + \dist(a,b,t)$. This does not affect the results presented in this paper. %

The triangle inequalities in distance and trip costs, i.e.  $\dist(a, a') \leq \dist(a, b) + \dist(b, a')$ and $\cost_{a,a',t} \leq \cost_{a,b,t} + \cost_{b,a',t+\dist(a,b)}$ for all $a, a', b \in \loc$ and all $t \in \timeSet$, are not necessary for our results on drivers' incentives. One concern in practice is that riders may try to break a long trip into a few shorter trips in order to get a lower total price. %
With the triangle inequalities, such strategies are not useful, since  the shorter trips together take a longer time and incur a higher cost, and the total payment, which equals the difference in the welfare gain from an extra driver at the (origin, starting time) and the (destination, ending time), plus the total costs, is higher if the ending time is later.

\section{Relation to the Literature} \label{appx:literature}

\subsection{Dynamic VCG} \label{appx:dynamic_vcg}

The dynamic VCG mechanisms~\cite{bergemann2010dynamic,cavallo2009efficient} truthfully implement efficient decision policies, where agent receive private information over time. Under the dynamic VCG mechanisms, the payment to agent $i$ in each period is equal to the flow marginal externality imposed on the other agents by its presence in this period %
 only~\cite{cavallo2009efficient}. 

The dynamic VCG mechanism can be adapted for the ridesharing problem, where there is no uncertainty in the transition of states (the actions taken by all drivers at time $t$ fully determines the state of the platform at time $t+1$) and no private information from the drivers' side (the location of the driver is known to the mechanism and we assume homogeneous driver costs and no location preferences). We actually show that a variation of the driver-optimal dynamic mechanism that we discussed in Section~\ref{sec:st_pricing}, where driver payments are ``shifted" over time, is equivalent to the dynamic VCG mechanism. 

The dynamic VCG mechanism for ridesharing, however, fails to be incentive compatible, since some drivers may be paid a negative payment for certain periods of time, and the drivers would have incentive to decline the dispatch at such times to avoid making the payments. This is because the existence of some driver for only one period of time may exert negative externality on the rest of the economy by inducing seemingly efficient actions that result in suboptimal positioning of the rest of the drivers in the subsequent time periods. 

\bigskip

We illustrate this via analyzing the economy introduced in Example~\ref{exmp:toy_econ_2}, as shown in Figure~\ref{fig:toy_econ_2}.

\addtocounter{example}{-8}
\begin{example}[Continued]

 Without driver $1$, driver $2$ would be dispatched to pick up riders $1$ and $2$ and achieve a total welfare of $6$. With driver $1$, one %
 welfare-optimal dispatching plan sends driver $1$ to take the path $((B,C,0),~(C,C,1)~,(C,C,2))$ and sends driver $2$ to take the path $((B,A,0),~(A,A,2))$. 

At time $0$, driver $1$ takes trip $(B,C,0)$ and driver $2$ takes trip $(B,A,0)$. Driver $1$ contributes $0$ to welfare at time $0$ since she did not pick up any driver.
When time $1$ comes, if driver $1$ appears for only one period of time, the optimal welfare achieved by the rest of the economy would only be $1$--- driver $1$ disappears and driver $2$ picked up rider $3$. Therefore, the payment to driver $1$ at time $0$ would be $-5$, since exerted a negative externality of $5$ on the rest of the economy by appearing only at time $0$. Similarly, we can compute that the payment to driver $1$ at times $1$ and $2$ would be $1$ and $5$ respectively, giving her a total payment of $-5+ 1 + 5 = 1$.

Now consider the scenario where driver $1$ declines the dispatch, refuses to make the payment and stays in location $B$, and assume that driver $2$ still followed the mechanism and drove to location $A$. When time $1$ comes, driver $1$ would again be dispatched to drive to $C$ at time $1$ and pick up rider $2$ at time $2$. We can check that the payment to driver $1$ at time $1$ would be $0$, and the payment to driver $1$ at time $2$ would be $5$--- the amount the rest of the economy gains from the existence of driver $2$ at that time. This is a useful deviation, thus the dynamic VCG mechanism where drivers are allowed to freely decline trips is not IC.
\qed
\end{example}
\addtocounter{example}{7}

\subsection{Trading Networks} \label{appx:trading_network}

The literature on trading networks studies economic models where agents in a network can trade via bilateral contracts~\citep{hatfield2013stability,hatfield2015chain,ostrovsky2008stability}. Efficient, competitive equilibrium outcomes exist when agents' valuation functions satisfy the \emph{full substitution} property, and the utilities of agents on either end of the trading network form lattices.

Assume that all trip costs and early exit costs are zero, the %
computation of optimal dispatching in a ridesharing platform can be formulated as a trading network problem in the following way: 
\begin{enumerate}[$\bullet$]
	\setlength\itemsep{0.0em}
	\item For each driver or rider, there is a node in the network. 
	\item For each driver $i \in \driverSet$ and each rider $j \in \riders$, there is an edge from $i$ to $j$ if $\tr_j \geq \te_i + \dist(\re_i, \origin_j)$, i.e. driver $i$ is able to pick up rider $j$ if she drivers directly to $\origin_j$ after entering.
	\item For any two riders $j$ and $j'$ in $\driverSet$, there is an edge from $j$ to $j'$ if (i) $\tr_j + \dist(\origin_j, \dest_j) + \dist(\dest_j, \origin_{j'}) \leq \tr_{j'}$ assuming $\dest_j \neq \origin_{j'}$ or (ii) $\tr_j + \dist(\origin_j, \dest_j) \leq \tr_{j'}$ if $\dest_j = \origin_{j'}$. Intuitively, riders $j$ can trade to rider $j'$  if a driver is still able to pick up rider $j'$ after dropping off rider $j$.
\end{enumerate}

What is being traded in the network is the right to use the car over the rest of the planning horizon. Each driver is able to trade with at most one rider. A driver's utilities is zero if she does not trade, and her utility is equal to the her payment if she did trade. Each rider values buying the right to use at most one car, and values it at $\val_j$. If she did buy the right to use one car, her utility is $\val_j$ minus the price she paid to buy the right to use the car, plus the payment she collected from the rider who bought the right to use the car from her. Riders that did not buy a car cannot sell, and this can be modeled by the riders valuing such contracts at $-\infty$.

Drivers costs can also be handled, by introducing additional vertices corresponding to location and time pairs (which can be interpreted as intermediaries in the market at a particular location and time, who need to buy and sell the rights to use the same number of cars, from this location and time onward).

With existing results in the trading network literature, we can show the existence of welfare-optimal, competitive equilibrium outcomes in the trading network, and we can also establish the lattice structure of drivers' utility under all CE outcomes.
But this does not help with proving our main result on the incentive compatibility of the STP mechanism. In the trading network, a driver makes only one decision, which is whether to sell the right to use the car for the rest of the planning horizon, and if so, to which rider or intermediary in the market. This is different from our setting, where the driver actively makes a decision on how to act every period of time (unless she's en-route driving a rider to the destination). 
For the same reason, this mapping to a trading network 
would also complicate arguments in regard to establishing the existence of anonymous CE prices and the structure of CE, and in regard to proving core-equivalence.

\section{Additional Simulation Results} \label{appx:additional_simulations}

We present in this section the additional simulation results that are omitted from the body of the paper.

\subsection{Morning Rush Hour}

Figures~\ref{fig:ex2_drivers_appx} and~\ref{fig:ex2_prices_appx} show the average number of drivers and average prices for each of the five trips that are not analyzed in Section~\ref{sec:sim_morning_rush} for the morning rush hour scenario.

\begin{figure}[hpbt!]
\centering
\begin{subfigure}[t]{0.4\textwidth}
	\centering
    \includegraphics[height= \imageHeight in]{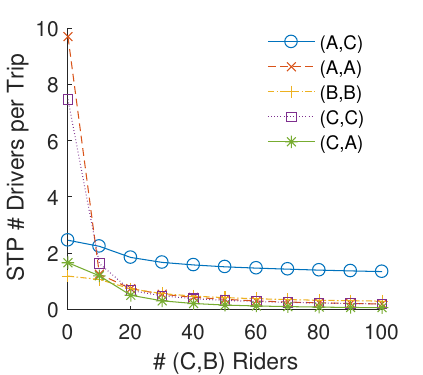}
    \caption{The STP mechanism. \label{fig:ex2_morning_stp_drivers_appx}}
\end{subfigure}%
\begin{subfigure}[t]{0.4\textwidth}
	\centering
    \includegraphics[height=\imageHeight in]{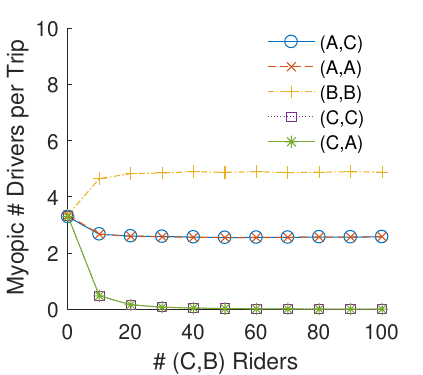}
    \caption{The myopic pricing mechanism.  \label{fig:ex2_morning_myopic_drivers_appx}}
\end{subfigure}%
\caption{Additional comparison of the number of drivers per trip for the morning rush hour. \label{fig:ex2_drivers_appx} }    
\end{figure}

The STP mechanism dispatches a reasonably high number of drivers to the $(A,C)$ trip since there is a high demand for drivers at $C$ (see Figure~\ref{fig:ex2_drivers_appx}).
In contrast, though the myopic pricing mechanism is not sending too many drivers from $C$ to $C$ or $A$, many drivers linger around $B$ due to the excessive supply, and the mechanism did not relocate more driver from $A$ to $C$ than from $A$ to $A$, despite the imbalance in demand in these locations. Prices as shown in Figure~\ref{fig:ex2_prices_appx} are also intuitive and easy to interpret.

\begin{figure}[hpbt!]
\centering
\begin{subfigure}[t]{0.4\textwidth}
	\centering
    \includegraphics[height= \imageHeight in]{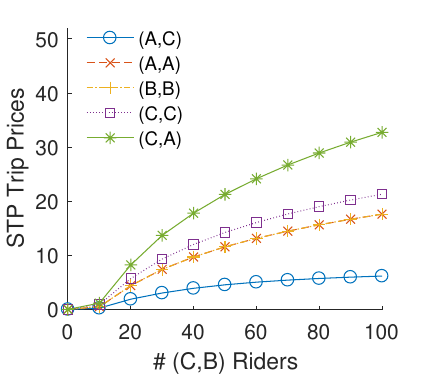}
    \caption{The STP mechanism.\label{fig:ex2_morning_stp_prices_appx}}
\end{subfigure}%
\begin{subfigure}[t]{0.4\textwidth}
	\centering
    \includegraphics[height=\imageHeight in]{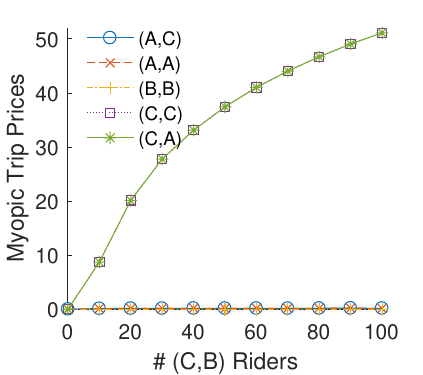}
    \caption{The myopic pricing mechanism. \label{fig:ex2_morning_myopic_prices_appx}}
\end{subfigure}%
\caption{Additional comparison trip prices for the morning rush hour. \label{fig:ex2_prices_appx} }
\end{figure}

\end{document}